\documentclass[conference]{IEEEtran}

\IEEEoverridecommandlockouts

\usepackage{cite}
\usepackage{orcidlink}
\usepackage{amsmath,amssymb,amsthm,amsfonts}

\usepackage{wrapfig}
\usepackage{multirow}

\usepackage{xcolor}
\usepackage{graphicx}

\usepackage{makecell}

\usepackage{bussproofs}
\usepackage{wasysym}
\usepackage{fontawesome}
\usepackage{pxfonts}
\usepackage{mathtools}
\usepackage[font=small]{caption}
\usepackage{subcaption}
\usepackage{tikz}
\usetikzlibrary{positioning,arrows,hobby,backgrounds,calc,trees,shapes,fit}
\pgfdeclarelayer{background}
\pgfsetlayers{background,main}

\definecolor{bordeau}{rgb}{0.3515625,0,0.234375}
\definecolor{darkspringgreen}{rgb}{0.09, 0.45, 0.27}
\definecolor{color_bordeau}{rgb}{0.3515625,0,0.234375}
\definecolor{color_old_paper}{rgb}{.95, .95, .85}
\definecolor{color_slategray}{RGB}{112,128,144}
\definecolor{color_darkslateblue}{RGB}{72,61,139}
\definecolor{CS_Red}{RGB}{150,2,60}
\definecolor{CS_LightRed}{RGB}{207,161,162}
\definecolor{CS_Grey}{RGB}{133,120,148}
\definecolor{color_old_paper}{rgb}{.95, .95, .85}
\definecolor{chromeyellow}{rgb}{1.0, 0.65, 0.0}
\definecolor{amber}{rgb}{1.0, 0.75, 0.0}
\definecolor{hibou_verdictColor_cov}{RGB}{0,0,205}
\definecolor{hibou_verdictColor_short}{RGB}{0,205,205}
\definecolor{hibou_verdictColor_multipref}{RGB}{105,89,205}
\definecolor{hibou_verdictColor_slice}{RGB}{154,50,205}
\definecolor{hibou_verdictColor_inconc}{RGB}{205,16,118}
\definecolor{hibou_verdictColor_lackobs}{RGB}{205,55,0}
\definecolor{hibou_verdictColor_dead}{RGB}{178,34,34}
\definecolor{hibou_verdictColor_out}{RGB}{205,0,0}
\definecolor{col_global_accept}{rgb}{0.1, 0.4, 0.2}
\definecolor{col_global_sem}{rgb}{0.5, 0.8, 0.1}
\definecolor{col_accept}{rgb}{0.0, 0.0, 1}
\definecolor{col_sem}{rgb}{0.0, 0.5, 0.75}
\definecolor{col_multipref}{rgb}{0.5, 0.0, 0.75}
\definecolor{darkspringgreen}{rgb}{0.09, 0.45, 0.27}
\definecolor{hibou_col_newfresh}{RGB}{94, 22, 130}

\definecolor{hibou_col_lf}{RGB}{22, 22, 130}
\newcommand{\hlf}[1]{\textcolor{hibou_col_lf}{#1}}
\definecolor{hibou_col_pr}{RGB}{22, 130, 22}

\definecolor{hibou_col_ms}{RGB}{15, 86, 15}
\newcommand{\hms}[1]{\textcolor{hibou_col_ms}{#1}}
\newcommand{\shortColRed}[1]{\textcolor{red}{#1}}

\newcommand{\shortColBlue}[1]{\textcolor{blue}{#1}}
\newcommand{\shortColOrange}[1]{\textcolor{orange}{#1}}

\definecolor{my_grey}{RGB}{191, 191, 191}
\newcommand{\shortColGrey}[1]{\textcolor{my_grey}{#1}}

\newtheorem{definition}{Definition}
\newtheorem{theorem}{Theorem}
\newtheorem{property}{Property}
\newtheorem{example}{Example}
\newtheorem*{example*}{Example}

\usepackage{scalerel}
\usepackage{centernot}

\DeclareRobustCommand\doubleVerticalTimesDefault{%
  \leavevmode
  {\sbox0{\ddag}%
   \ooalign{\raisebox{\ht0-\height}{$\times$}\cr
            \raisebox{\depth-\dp0}{\scalebox{1}[-1]{$\times$}}\cr}%
  }%
}

\newcommand{\doubleVerticalTimes}{\scalerel*{\doubleVerticalTimesDefault}{b}}

\newcommand{\globalInterleaving}{||}

\newcommand{\globalStrictSeq}{;}

\newcommand{\globalWeakSeq}{\doubleVerticalTimes}

\newcommand{\globalConditionalSequencing}[1]{\doubleVerticalTimes|_{#1}}

\newcommand{\isPruneBase}{\mathrlap{\raisebox{-.125\height}{\doubleVerticalTimes}}\longrightarrow}

\newcommand{\isNotPruneBase}{\centernot{\isPruneBase}}

\newcommand{\isPruneOf}[1]{\mathrlap{\raisebox{-.125\height}{\doubleVerticalTimes}}\xrightarrow{#1}}

\newcommand{\isNotPruneOf}[1]{\centernot{\isPruneOf{#1}}}

\newcommand{\sliceOf}[1]{\widetilde{#1}}

\newcommand{\partitionsOf}[1]{\text{Part}(#1)}

\newcommand{\positionsOf}[1]{\text{pos}(#1)}

\newcommand{\measureDecrements}[1]{\overset{#1}{\rightharpoondown}}
\newcommand{\measureNotDecrements}[1]{\centernot{\measureDecrements{#1}}}

\newcommand{\verdictOk}{\textcolor{hibou_verdictColor_cov}{\text{Ok}}}
\newcommand{\verdictKo}{\textcolor{hibou_verdictColor_out}{\text{Ko}}}

\newcommand{\rulePass}{\shortColBlue{R_p}}
\newcommand{\ruleFail}{\shortColRed{R_f}}
\newcommand{\ruleExec}{\shortColOrange{R_e}}
\newcommand{\ruleSimu}{\shortColGrey{R_s}}

\newcommand{\loopDepthAtPosBase}{\beta}
\newcommand{\loopDepthAtPos}[2]{\loopDepthAtPosBase(#1,#2)}

\newcommand{\numActOutsideBase}{\eta}
\newcommand{\numActOutside}[1]{\numActOutsideBase(#1)}

\newcommand{\multiAppendLeft}{\resizebox{!}{9pt}{$\overset{\rightarrow}{\odot}$}}

\usepackage{listings}

\lstdefinelanguage{hibou_hsf}%
{
    escapeinside={<@}{@>},
    morekeywords = [1]{strict,seq,alt,par,coreg,loopS,loopH,loopW,loopP,synch},
    morekeywords = [2]{@lifeline,@message,@analyze_option,@explore_option}
}

\lstdefinestyle{coloured_hibou_hsf} {
    showstringspaces = false,
    basicstyle = {\ttfamily\footnotesize \color[rgb]{0, 0, 0}},
    columns=fixed,keepspaces=true,tabsize=2,
    backgroundcolor = {\color[rgb]{1, 1, 1}},
    keywordstyle = [1]{\bfseries},
    keywordstyle = [2]{\bfseries\color[rgb]{.5,0,.5}},frame=single
}

\lstdefinestyle{coloured_hibou_hsf_inline} {
    showstringspaces = false,
    basicstyle = {\scriptsize\ttfamily \color[rgb]{0, 0, 0}},
    columns=fixed,keepspaces=true,
    keywordstyle = [1]{\bfseries},
    keywordstyle = [2]{\bfseries\color[rgb]{.5,0,.5}},frame=single
}

\lstdefinelanguage{hibou_htf}%
{
    escapeinside={<@}{@>},
    morekeywords = [1]{[,]}
}

\lstdefinestyle{coloured_hibou_htf} {
    showstringspaces = false,
    basicstyle = {\ttfamily\footnotesize \color[rgb]{0, 0, 0}},
    columns=fixed,keepspaces=true,
    backgroundcolor = {\color[rgb]{1, 1, 1}},
    keywordstyle = [1]{\color[rgb]{.5,0,.5}\bfseries},frame=single
}

\lstdefinestyle{coloured_hibou_htf_inline} {
    showstringspaces = false,
    basicstyle = {\scriptsize\ttfamily \color[rgb]{0, 0, 0}},
    columns=fixed,keepspaces=true,
    keywordstyle = [1]{\color[rgb]{.5,0,.5}\bfseries},frame=single
}

\newenvironment{scprooftree}[1]%
  {\gdef\scalefactor{#1}\begin{center}\proofSkipAmount \leavevmode}%
  {\scalebox{\scalefactor}{\DisplayProof}\proofSkipAmount \end{center} }

\title{Tooling Offline Runtime Verification against Interaction Models : recognizing sliced behaviors using parameterized simulation}

\author{

\IEEEauthorblockN{
Erwan Mahe \orcidlink{0000-0002-5322-4337},\\
Boutheina Bannour \orcidlink{0000-0002-4943-7807},
Christophe Gaston \orcidlink{0000-0001-6865-5108},
Arnault Lapitre \orcidlink{0000-0002-2185-4051}
}
\IEEEauthorblockA{\textit{Université Paris Saclay,}
\textit{CEA, LIST}\\
F-91120, Palaiseau, France}

\and

\IEEEauthorblockN{Pascale Le Gall \orcidlink{0000-0002-8955-6835}}
\IEEEauthorblockA{\textit{Université Paris Saclay,} \\
\textit{CentraleSupélec}\\
F-91192, Gif-sur-Yvette, France}

}

\begin{document}

\maketitle

\begin{abstract}
Offline runtime verification involves the static analysis of executions of a system against a specification. For distributed systems, it is generally not possible to characterize executions in the form of global traces, given the absence of a global clock.
To account for this, we model executions as collections of local traces called {\em multi-traces}, with one local trace per group of co-localized actors that share a common clock.
Due to the difficulty of synchronizing the start and end of the recordings of local traces, events may be missing at their beginning or end.
Considering such {\em partially observed} multi-traces is challenging for runtime verification.
To that end, we propose an algorithm that verifies the conformity of such traces against formal specifications called  {\em Interactions} (akin to Message Sequence Charts). 
It relies on parameterized simulation to reconstitute unobserved behaviors.
\end{abstract}


\section{Introduction}

Runtime Verification (RV)~\cite{introduction_to_runtime_verification_BartocciFFR18,surveyRV_SanchezSABBCFFK19} refers to a group of techniques aiming at confronting observed system executions to formal references, specifying legal system executions, in order to identify non-conformance. Executions are observed via instrumentation and collected in {\em traces} consisting of sequences of atomic events. Such events often correspond to communication actions, consisting of emissions or receptions of messages observed at the system's interfaces under observation.
Most approaches for RV can be described as either {\em offline} or {\em online}. In online approaches, events are processed on the fly whenever they are observed, while in offline approaches (abbrv. ORV) - which are the focus of this paper -, traces are logged a priori to their analysis.
Capturing executions of a Distributed System (DS) as traces is possible if it is observed via a unique interface deployed on a single machine. 
In practice, however, as a DS may be distributed across distinct machines, so is the instrumentation that observes its execution.
Also, events observed on different and geographically distant computers cannot be easily temporally ordered as there is no common clock to label them with comparable dates.
For these reasons, instead of a single global trace, one rather observes a set of local traces occurring on specific sub-systems or groups of {\em co-localized} (i.e. sharing a common clock) sub-systems.
In such a situation, an execution is characterized as a structured collection of local traces, which we call a {\em multi-trace}.

While most RV techniques are based on formal references given in the form of automata \cite{constraint_based_oracles_for_timed_distributed_systems} or temporal logic formulas \cite{Falcone16}, we use an {\em interaction language}. 
Interactions are models whose most well-known instances are UML Sequence Diagrams (UML-SD) \cite{UML} or Message Sequence Charts (MSC) \cite{MSC}. Interactions specify the communication flow between entities constituting a system. They are particularly adapted to specify DS behaviors, as DS are, by nature, composed of sub-systems interacting via message passing.  

The graphical representation of interactions provides an intuitive vision of a DS's expected behaviors. Each sub-system is represented by a vertical line called a {\em lifeline} while message passing between sub-systems is represented by horizontal arrows drawn between the corresponding lifelines. As time flows from top to bottom, behaviors expected to occur on a given lifeline are sequences of emissions and receptions of messages that match the horizontal arrows entering or exiting the lifeline from top to bottom. More complex behaviors can be specified via operators drawn as annotated boxes.

In previous works~\cite{equivalence_of_denotational_and_operational_semantics_for_interaction_languages}, we defined the semantics of interactions without the need for translations to other formalisms. In particular, an operational semantics in the style introduced by Plotkin for process algebras \cite{a_structural_approach_to_operational_semantics} can be used to animate interaction models, explore their semantics, and define ORV algorithms \cite{revisiting_semantics_of_interactions_for_trace_validity_analysis,a_small_step_approach_to_multi_trace_checking_against_interactions}.
This semantics is based on the key notion of {\em follow-up interaction}. 
Given an initial interaction $i$ which specifies a set of expected behaviors, if a certain communication action $a$ (either an emission or a reception of a message) inside $i$ can immediately occur, then there exists a follow-up $i'$, which we denote by $i \xrightarrow{a} i'$, such that $i'$ specifies continuations of behaviors of $i$ that start with the occurrence of $a$.
Such atomic execution steps can be used to display graphically the semantics of an interaction in a tree-like structure.
Being grounded by this small-step semantics, our approach is, to the best of our knowledge, the first tooled approach to offer interaction animation without going through translation mechanisms to intermediate formalisms like automata~\cite{sd_to_tiosts_BannourGS11}, Petri-nets~\cite{sd_to_petrinet_FariaP16} or others as overviewed in~\cite{othertool_interaction_MouakherDA22}.

The main contributions of this paper are extensions of the work in \cite{a_small_step_approach_to_multi_trace_checking_against_interactions,equivalence_of_denotational_and_operational_semantics_for_interaction_languages} under two aspects:\\

\noindent\textbf{(1)} The observability constraints imposed by monitoring:
\begin{itemize}
    \item We define a finer notion of multi-trace, whose component local traces are defined on groups of co-localized lifelines (e.g. that share a common clock) rather than on single lifelines.
    Analyzing those richer multi-traces allows taking advantage of the additional information provided by the existence of common clocks;
    \item We define a new ORV algorithm tolerant to the absence of synchronization related to both the beginning and the end of observation across distant monitors.
    To do so we rely on parameterized simulation for guessing missing/unobserved behaviors.
    With this new algorithm, logging on each co-localization can start and end independently and at any time when the system is in operation; 
\end{itemize}

\noindent\textbf{(2)} Our specification formalism and tool implementation:
\begin{itemize}
    \item We introduce a $coreg_r$ operator for specifying behaviors that ought to be concurrent on specific sub-systems (the concurrent region $r$) while being weakly sequential on others. This new operator allows defining more expressive (w.r.t. the language from \cite{equivalence_of_denotational_and_operational_semantics_for_interaction_languages}) specifications. It also simplifies the definition of the language, the weakly sequential $seq$ and interleaving $par$ operators being covered by its definition (being resp. equivalent to $coreg_\emptyset$ and $coreg_L$ where $L$ is the set of all lifelines).
    \item We also propose a reformulation of the operational semantics from \cite{equivalence_of_denotational_and_operational_semantics_for_interaction_languages} which has the advantage of using fewer inductive predicates;
    \item We extend the tool implementation mentioned in \cite{a_small_step_approach_to_multi_trace_checking_against_interactions} to include our new contributions and present its interface in more details. This tool, called HIBOU, allows designing interactions, exploring their semantics, generating multi-traces and performing RV. We also propose an experimental evaluation of our simulation-based ORV approach using this tool.
\end{itemize}

The paper is organized as follows.
In Sec.\ref{sec:multitraces} we introduce the notions of {\em multi-trace} for characterizing behaviors of DS and of {\em multi-trace slice} for characterizing partial observations of such behaviors.
Then, in Sec.\ref{sec:interactions} we define our language of {\em interactions} and its semantics in terms of multi-traces.
Following those definitions, we introduce in Sec.\ref{sec:algo} a generic algorithm for verifying partially observed distributed behaviors (i.e. multi-trace slices) against formal specifications given in the form of interactions. 
This algorithm uses simulation steps in order to complete optimistically behaviors that might be missing from the slice (because it is not observed). It is generic in so far as the manner with which simulation is performed is parameterized by a certain criterion.
After that, we propose one such criterion in Sec.\ref{sec:criterion}, apply the resulting algorithm on an example, discuss its advantages and limitations and present results from experiments.
Related works and the position of our contribution are then discussed in Sec.\ref{sec:related}.
Finally, in Sec.\ref{sec:tool}, we present our tool implementation HIBOU.

\section{Characterizing observed DS executions\label{sec:multitraces}}

\subsection{Multi-traces to model executions}

\begin{figure*}[ht]
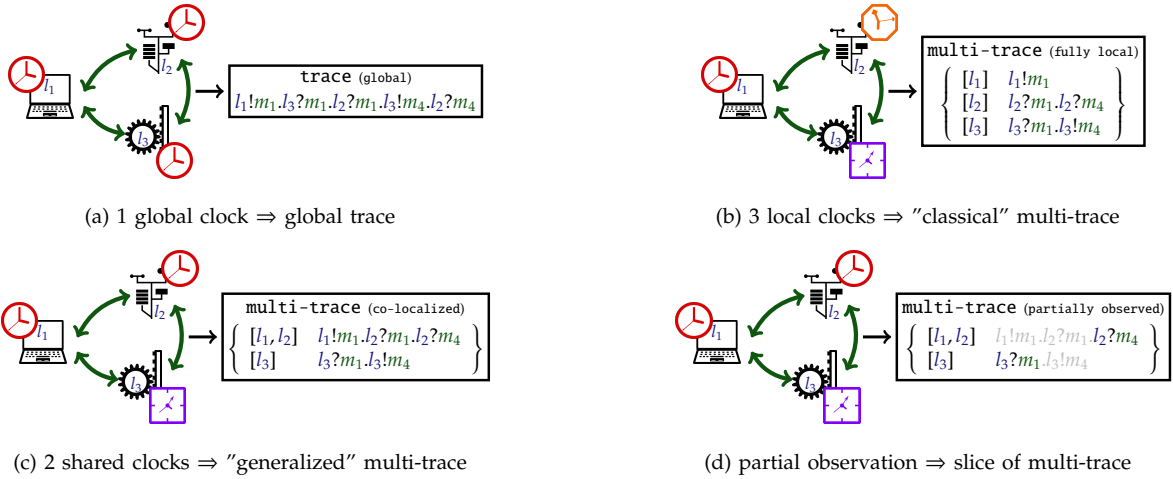

    \centering
    \begin{minipage}{.49\linewidth}
        \centering
        \begin{subfigure}[t]{\linewidth}
            \centering
            \scalebox{.75}{\begin{tikzpicture}
\tikzstyle{trace}=[draw,rectangle,line width=1.25pt,inner sep=.1cm,outer sep=.1cm]
\node (ds) at (0,0)
{
    \resizebox{4cm}{!}{
    \input{figures/2_multitrace/trio_1clock}
    }
};
\node[trace,right=-.1cm of ds] (tr) {
        \begin{tikzpicture}[every node/.style={minimum height=0pt,minimum width=0pt,inner sep=0,outer sep=0}]
        \node (lab) at (0,0) {\texttt{trace} \scriptsize(\texttt{global})};
        \node[below=.15cm of lab] (gltr) {$\hlf{l_1}!\hms{m_1}.\hlf{l_3}?\hms{m_1}.\hlf{l_2}?\hms{m_1}.\hlf{l_3}!\hms{m_4}.\hlf{l_2}?\hms{m_4}$};
        \end{tikzpicture}
};
\draw[->,line width=1.5pt] ($(ds.east) + (-.6,0)$) -- (tr.west); 
\end{tikzpicture}}
            \caption{$1$ global clock $\Rightarrow$ global trace} 
            \label{fig:mu_global1clock}
        \end{subfigure}
    \end{minipage}
    \begin{minipage}{.49\linewidth}
        \centering
        \begin{subfigure}[t]{\linewidth}
            \centering
            \scalebox{.75}{\begin{tikzpicture}
\tikzstyle{trace}=[draw,rectangle,line width=1.25pt,inner sep=.1cm,outer sep=.1cm]
\node (ds) at (0,0)
{
    \resizebox{4cm}{!}{
    \input{figures/2_multitrace/trio_3clock}
    }
};
\node[trace,right=-.1cm of ds] (tr) {
        \begin{tikzpicture}[every node/.style={minimum height=0pt,minimum width=0pt,inner sep=0,outer sep=0}]
        \node (lab) at (0,0) {\texttt{multi-trace} \scriptsize(\texttt{fully local})};
        \node[below=.15cm of lab] (gltr) {
$
\left\{
\begin{array}{ll}
\lbrack\hlf{l_1}\rbrack & \hlf{l_1}!\hms{m_1} \\
\lbrack\hlf{l_2}\rbrack & \hlf{l_2}?\hms{m_1}.\hlf{l_2}?\hms{m_4} \\
\lbrack\hlf{l_3}\rbrack & \hlf{l_3}?\hms{m_1}.\hlf{l_3}!\hms{m_4}
\end{array}
\right\}
$
        };
        \end{tikzpicture}
};
\draw[->,line width=1.5pt] ($(ds.east) + (-.6,0)$) -- (tr.west); 
\end{tikzpicture}}
            \caption{$3$ local clocks $\Rightarrow$ "classical" multi-trace} 
            \label{fig:mu_local3clocks}
        \end{subfigure}
    \end{minipage}
    \\[.1cm]
    \begin{minipage}{.49\linewidth}
        \centering
        \begin{subfigure}[t]{\linewidth}
            \centering
            \scalebox{.75}{\begin{tikzpicture}
\tikzstyle{trace}=[draw,rectangle,line width=1.25pt,inner sep=.1cm,outer sep=.1cm]
\node (ds) at (0,0)
{
    \resizebox{4cm}{!}{
    \input{figures/2_multitrace/trio_2clock}
    }
};
\node[trace,right=-.1cm of ds] (tr) {
        \begin{tikzpicture}[every node/.style={minimum height=0pt,minimum width=0pt,inner sep=0,outer sep=0}]
        \node (lab) at (0,0) {\texttt{multi-trace} \scriptsize(\texttt{co-localized})};
        \node[below=.15cm of lab] (gltr) {
$
\left\{
\begin{array}{ll}
\lbrack\hlf{l_1},\hlf{l_2}\rbrack & \hlf{l_1}!\hms{m_1}.\hlf{l_2}?\hms{m_1}.\hlf{l_2}?\hms{m_4}\\
\lbrack\hlf{l_3}\rbrack & \hlf{l_3}?\hms{m_1}.\hlf{l_3}!\hms{m_4}
\end{array}
\right\}
$
        };
        \end{tikzpicture}
};
\draw[->,line width=1.5pt] ($(ds.east) + (-.6,0)$) -- (tr.west); 
\end{tikzpicture}}
            \caption{$2$ shared clocks $\Rightarrow$ "generalized" multi-trace} 
            \label{fig:mu_coloc2clocks}
        \end{subfigure}
    \end{minipage}
    \begin{minipage}{.49\linewidth}
        \centering
        \begin{subfigure}[t]{\linewidth}
            \centering
            \scalebox{.75}{\begin{tikzpicture}
\tikzstyle{trace}=[draw,rectangle,line width=1.25pt,inner sep=.1cm,outer sep=.1cm]
\node (ds) at (0,0)
{
    \resizebox{4cm}{!}{
    \input{figures/2_multitrace/trio_2clock}
    }
};
\node[trace,right=-.1cm of ds] (tr) {
        \begin{tikzpicture}[every node/.style={minimum height=0pt,minimum width=0pt,inner sep=0,outer sep=0}]
        \node (lab) at (0,0) {\texttt{multi-trace} \scriptsize(\texttt{partially observed})};
        \node[below=.15cm of lab] (gltr) {
$
\left\{
\begin{array}{ll}
\lbrack\hlf{l_1},\hlf{l_2}\rbrack & \shortColGrey{l_1!m_1.l_2?m_1.}\hlf{l_2}?\hms{m_4}\\
\lbrack\hlf{l_3}\rbrack & \hlf{l_3}?\hms{m_1}\shortColGrey{.l_3!m_4}
\end{array}
\right\}
$
        };
        \end{tikzpicture}
};
\draw[->,line width=1.5pt] ($(ds.east) + (-.6,0)$) -- (tr.west); 
\end{tikzpicture}}
            \caption{partial observation $\Rightarrow$ slice of multi-trace} 
            \label{fig:mu_coloc_partial}
        \end{subfigure}
    \end{minipage}
    \caption{Trace collection}
    \label{fig:trace_collection}
\end{figure*}

The asynchronous exchanges of messages are at the heart of the behaviors of Distributed Systems (DS). Those exchanges can be modeled using discrete communication actions (abbrv. as {\em actions}) corresponding to atomic emissions and receptions of messages. Those actions occur at the communication interfaces of specific sub-systems (those that emit and/or receive the corresponding messages) within the DS.

To formalize this, we describe the sub-systems constituting a DS using a set $L$ of {\em lifelines}, and the messages that can transit through it using a set $M$ of {\em messages}.

Elements of the set $\mathbb{A}$ of communication actions are~:
\begin{itemize}
    \item either of the form $\hlf{l}!\hms{m}$, corresponding to the emission of the message $m$ in $M$ from the lifeline $l$ in $L$ 
    \item or of the form $\hlf{l}?\hms{m}$, corresponding to the reception of the message $m$ in $M$ by the lifeline $l$ in $L$.
\end{itemize} 
For any action $a \in \mathbb{A}$ of the form $\hlf{l}!\hms{m}$ or $\hlf{l}?\hms{m}$, we denote by $\theta(a)$ the lifeline $l$ on which $a$ occurs.

An execution of a DS can be characterized by the actions that occurred in its span and by the order between their occurrences. Depending on the architecture of the DS our ability to reorder those actions may vary.

Fig.\ref{fig:mu_global1clock} describes (on the left) a DS with three lifelines: $\hlf{l_1}$ (the computer icon), $\hlf{l_2}$ (the sensor icon) and $\hlf{l_3}$ (the gear icon) which all share the same global clock (as indicated by the drawn clocks).
Thanks to the global clock, actions can be ordered globally, whichever is the sub-system on which they occur. 
Hence, an execution of the DS can be characterized by a global sequence of actions, which we call a {\em trace}. Traces are sequences of actions where $\varepsilon$ represents the empty sequence and are concatenated using the "$.$" operator.
We denote by $\mathbb{T} = \mathbb{A}^*$ the set of all global traces\footnote{Given a set $X$, $X^*$ denotes the set of all sequences with elements in $X$ (this is the Kleene star notation).}.
The right side of Fig.\ref{fig:mu_global1clock} describes an execution of our example DS in the form of a global trace: \[\hlf{l_1}!\hms{m_1}.\hlf{l_3}?\hms{m_1}.\hlf{l_2}?\hms{m_1}.\hlf{l_3}!\hms{m_4}.\hlf{l_2}?\hms{m_4}\]
This execution can be understood as follows: $\hlf{l_1}$ broadcasts message $\hms{m_1}$ to both $\hlf{l_3}$ and $\hlf{l_2}$ and then $\hlf{l_3}$ sends $\hms{m_4}$ to $\hlf{l_2}$.

Characterizing executions as global traces requires that all sub-systems share a common clock which we may call the global clock.
However, in all generality, as the sub-systems of a DS can be distributed across distant machines, they may not share a common clock \cite{Lamport19b} and such a centralization and reordering of logging is not possible.
Fig.\ref{fig:mu_local3clocks} describes a similar system as that of Fig.\ref{fig:mu_global1clock}  except that all three lifelines have different local clocks (as indicated by the drawn clocks).
Let us suppose however, that the same execution occurred in both cases. 
Then, because it is not possible to reorder actions occurring on distinct lifelines, instead of a global trace, the execution is rather characterized by a set of three local traces (one for each sub-system), which we call a {\em multi-trace}:
\[
\begin{array}{ll}
\left[\hlf{l_1} \right] & \hlf{l_1}!\hms{m_1}  \\
\left[\hlf{l_2}\right] & \hlf{l_2}?\hms{m_1}.\hlf{l_2}?\hms{m_4} \\
\left[\hlf{l_3}\right] & \hlf{l_3}?\hms{m_1}.\hlf{l_3}!\hms{m_4}
\end{array}
\]
This notion of multi-trace can be found e.g. in \cite{constraint_based_oracles_for_timed_distributed_systems,a_small_step_approach_to_multi_trace_checking_against_interactions} as well as in \cite{trace_partitioning_and_local_monitoring_for_asynchronous_components_AttardF17} (called partitioned traces) and \cite{passive_conformance_testing_of_service_choreographies,Falcone16} (as sets of logs/local traces).

Still, it may be so that groups of sub-systems do share a common clock. 
We call those groups {\em co-localizations} \cite{modeling_concurrency_with_partial_orders}. 
Fig.\ref{fig:mu_coloc2clocks} describes a variant of our example where lifelines $\hlf{l_1}$ and $\hlf{l_2}$ share a common clock (as indicated by the drawn clocks).
In this case it is possible to order an action occurring on $\hlf{l_1}$ w.r.t. another occurring on $\hlf{l_2}$.
As a result, the execution (the same as in the previous two cases) can be characterized by a {\em generalized} multi-trace where lifelines $\hlf{l_1}$ and $\hlf{l_2}$ constitute together a co-localization and where lifeline $\hlf{l_3}$ alone represents another co-localization. The multi-trace is then composed of two collected local traces, each representing a local order of actions on one of the two co-localizations. 
\[
\begin{array}{ll}
\left[\hlf{l_1}, \hlf{l_2} \right] & \hlf{l_1}!\hms{m_1}.\hlf{l_2}?\hms{m_1}.\hlf{l_2}?\hms{m_4} \\
\left[\hlf{l_3}\right] & \hlf{l_3}?\hms{m_1}.\hlf{l_3}!\hms{m_4}
\end{array}
\]

More formally, a {\em co-localization} $c$ is defined by a subset of lifelines $c \subseteq L$. We introduce $\mathbb{A}_{|c} = \{a \in \mathbb{A} ~|~ \theta(a) \in c\}$ the set of actions occurring on a lifeline in $c$ and $\mathbb{T}_{|c} = \mathbb{A}_{|c}^*$ the set of local traces defined on $c$.

For a set $X$, $\partitionsOf{X}$ denotes the set of partitions of $X$, where a partition $C \in \partitionsOf{X}$ is defined as a collection $C \subset \mathcal{P}(X)$ s.t. $\bigcup_{c\in C} c = X$ and $\forall~(c,c') \in C$, $c \neq c' \Rightarrow c \cap c' = \emptyset$. 

\begin{definition}[Multi-traces]
\label{def:multi-trace}
Given a partition $C \in \partitionsOf{L}$ of lifelines, we denote by $\mathbb{M}_{C}$ the set of multi-traces up to $C$.
A multi-trace $\mu \in \mathbb{M}_{C}$ is defined as a tuple of traces, each defined over events occurring on a specific co-localization from $C$, hence 
$\mathbb{M}_{C} = \prod_{c\in C} \mathbb{T}_{|c}$.
\end{definition}

Given a multi-trace $\mu \in \mathbb{M}_{C}$ and given any $c \in C$, we denote by:
\begin{itemize}
    \item $\mu_{|c}$ the local component of $\mu$ on $c$, 
    \item for any $t \in \mathbb{T}_{|c}$, $\mu[t]_c$ the multi-trace $\mu$ in which $t$ substitutes the $c$ component, i.e. $\forall~ c' \in C \setminus\{c\}, (\mu[t]_{c})_{|c'} = \mu_{|c'}$ and $(\mu[t]_{c})_{|c} = t$. 
\end{itemize}
We also extend the notation $\theta$ s.t. for any $a \in \mathbb{A}$ and $C \in \partitionsOf{L}$, $\theta_C(a)$ designates the unique co-localization $c \in C$ on which $a$ occurs i.e. s.t. $\theta(a) \in c$.

$\varepsilon_C$ denotes the empty multi-trace s.t. $\forall~c\in C$, $\varepsilon_{C|c} = \varepsilon$.
We define a left concatenation operator for multi-traces as follows: for any action $a$ and multi-trace $\mu$, $a ~ \multiAppendLeft ~ \mu = \mu[a.\mu_{|\theta_C(a)}]_{\theta_C(a)}$ is the multi-trace obtained by prepending action $a$ on the corresponding component of $\mu$ (i.e. $\mu_{|\theta_C(a)}$). 

Finally, we denote by $C_t= \{L\}$ the trivial partition (in which all lifelines are co-localized) and by $C_d = \{\{l\} ~|~ l \in L\}$ the discrete partition (in which no two lifelines are co-localized). 

Multi-traces defined up to $C_t$ and $C_d$ respectively correspond to the notions of global traces \cite{revisiting_semantics_of_interactions_for_trace_validity_analysis} and multi-traces as defined in \cite{a_small_step_approach_to_multi_trace_checking_against_interactions}.
Co-localizations allow us to generalize and bridge the gap between those two notions. 
Hence, any RV approach that can handle those generalized multi-traces can also handle both global traces (trivial partition $C_t$) and classical multi-traces (discrete partition $C_d$) which are particular cases of generalized multi-traces.

As a side note, it is possible to define projections from coarser multi-traces to finer multi-traces. For any set $X$, a partition $C' \in \partitionsOf{X}$ is a refinement of a partition $C \in \partitionsOf{X}$, denoted by $C' \leq C$, if for any $c' \in C'$, there exists $c \in C$ s.t. $c' \subset c$. 
For instance, the multi-traces from Fig.\ref{fig:mu_coloc2clocks} and Fig.\ref{fig:mu_local3clocks} can be obtained by projecting that from Fig.\ref{fig:mu_global1clock} because $\{ \{ \hlf{l_1}, \hlf{l_2} \}, \{ \hlf{l_3} \} \} \leq \{ \{ \hlf{l_1}, \hlf{l_2},\hlf{l_3} \} \}$ and $\{ \{ \hlf{l_1} \}, \{\hlf{l_2} \}, \{ \hlf{l_3} \} \leq \{ \{ \hlf{l_1}, \hlf{l_2},\hlf{l_3} \} \}$.

\subsection{Slices to model partially observed executions\label{ssec:slices}}

As we have seen, the example execution from Fig.\ref{fig:mu_global1clock} is characterized by the global trace :
\[
t = \hlf{l_1}!\hms{m_1}.\hlf{l_3}?\hms{m_1}.\hlf{l_2}?\hms{m_1}.\hlf{l_3}!\hms{m_4}.\hlf{l_2}?\hms{m_4}
\]
With the aim of performing Offline Runtime Verification we would then analyze this trace w.r.t. a certain behavioral model. However, in practice, the instrumentation that is used to collect this trace (from an observation of the execution that occurred) is not always perfect.
If we suppose that this trace $t$ characterizes the entirety of the execution, then, depending on the quality of the instrumentation, the collected trace $t'$ may be a suffix of $t$ (if the observation started too late), a prefix of $t$ (it it ended too early) or a sub-word of $t$ (if both).

This remark can equally be applied in the absence of a global clock. In the case of Fig.\ref{fig:mu_coloc2clocks}, $\hlf{l_1}$ and $\hlf{l_2}$ share a common clock while $\hlf{l_3}$ has its own local clock. Thus, the instrumentation must have at least two distinct observers to log actions on those two co-localizations.
In ideal conditions, the observation of this execution would yield the following multi-trace (see also Fig.\ref{fig:mu_coloc2clocks}) defined over $C=\{\{\hlf{l_1},\hlf{l_2}\},\{\hlf{l_3}\}\}$:
\[
\begin{array}{ll}
\lbrack\hlf{l_1},\hlf{l_2}\rbrack & \hlf{l_1}!\hms{m_1}.\hlf{l_2}?\hms{m_1}.\hlf{l_2}?\hms{m_4}\\
\lbrack\hlf{l_3}\rbrack & \hlf{l_3}?\hms{m_1}.\hlf{l_3}!\hms{m_4}
\end{array}
\]
In practice however, it may be difficult to synchronize the periods of observation between distant observers. Hence it may be so that some actions are missing at the beginning and/or at the end of local components of the multi-trace (as compared with the multi-trace that would have been observed in ideal conditions of observation).

This is illustrated on Fig.\ref{fig:mu_coloc_partial}. Here, we suppose that observation on component $\{\hlf{l_1},\hlf{l_2}\}$ started too late, leading to $\shortColGrey{l_1!m_1}$ and $\shortColGrey{l_2?m_1}$ not being observed while that on component $\{\hlf{l_3}\}$ ceased too early, leading to $\shortColGrey{l_3!m_4}$ not being observed. With the convention that the missing actions are greyed-out, this leads to the following partial multi-trace:
\[
\begin{array}{ll}
\left[\hlf{l_1}, \hlf{l_2} \right] & \shortColGrey{l_1!m_1}.\shortColGrey{l_2?m_1}.\hlf{l_2}?\hms{m_4} \\
\left[\hlf{l_3}\right] & \hlf{l_3}?\hms{m_1}.\shortColGrey{l_3!m_4}
\end{array}
\]
This notion of partial observation (i.e. some elements may be missing at the beginning or at the end of traces composing the multi-trace) corresponds to considering \emph{slices of multi-traces}, as defined in Def.\ref{def:slices}.

\begin{definition}[Slices]
\label{def:slices}
For any traces $t,t' \in \mathbb{T}_{|c}$, we say that $t'$ is a slice of $t$ iff there exists $t_{+}$ and $t_{-} \in \mathbb{T}_{|c}$ s.t. $t = t_{-}.t'.t_{+}$ and we denote by $\sliceOf{t}$ the set of slices of $t$.\\
For any multi-traces $\mu,\mu' \in \mathbb{M}_{C}$, we say that $\mu'$ is a slice of $\mu$ iff for all $c \in C$, $\mu'_{|c} \in \sliceOf{\mu_{|c}}$ and $\sliceOf{\mu}$ denotes the set of slices of $\mu$.
\end{definition}

In our example, if we denote by $\mu_0$ the multi-trace from Fig.\ref{fig:mu_coloc2clocks} (observed in ideal conditions) then the multi-trace described on Fig.\ref{fig:mu_coloc_partial} (collected in conditions of partial observation) is a slice of $\mu_0$ i.e. $\mu_0' \in \sliceOf{\mu_0}$.

\section{Interactions and their semantics\label{sec:interactions}}

\subsection{Syntax}

Formal behavioral models enable users to
\begin{enumerate}
    \item specify systems which may exhibit infinitely many distinct behaviors as finite expressions and
    \item automate Verification and Validation processes such as RV.
\end{enumerate}
The challenge of modeling and adapting RV to complex DS requires rich and intuitive formalisms.
Interactions \cite{equivalence_of_denotational_and_operational_semantics_for_interaction_languages} are well-suited to specify distributed behaviors thanks to their intuitive graphical representation (in the fashion of UML-Sequence Diagrams \cite{the_many_meanings_of_uml2_sd_a_survey}) while at the same time being formal models on which RV techniques can be applied \cite{revisiting_semantics_of_interactions_for_trace_validity_analysis,a_small_step_approach_to_multi_trace_checking_against_interactions}.
In that spirit, an interaction is both described as a syntactic term which takes the form of a binary tree (see Fig.\ref{fig:interaction_term}) and visualized as a diagram (see Fig.\ref{fig:interaction_diagram}) which may be familiar to many software engineers thanks to the wide use of UML-SD and MSC.
Fig.\ref{fig:interaction_example} depicts an example of interaction defined on the lifelines $\hlf{l_1}$, $\hlf{l_2}$ and $\hlf{l_3}$ with messages $\hms{m_1}$, $\hms{m_2}$, $\hms{m_3}$ and $\hms{m_4}$. It is represented as a diagram on the left (Fig.\ref{fig:interaction_diagram}), and as a term on the right (Fig.\ref{fig:interaction_term}).

\begin{figure}[h]
    \centering
    \begin{minipage}{.38\linewidth}
        \centering
        \begin{subfigure}[t]{\linewidth}
            \centering
            \resizebox{\textwidth}{!}{
                \input{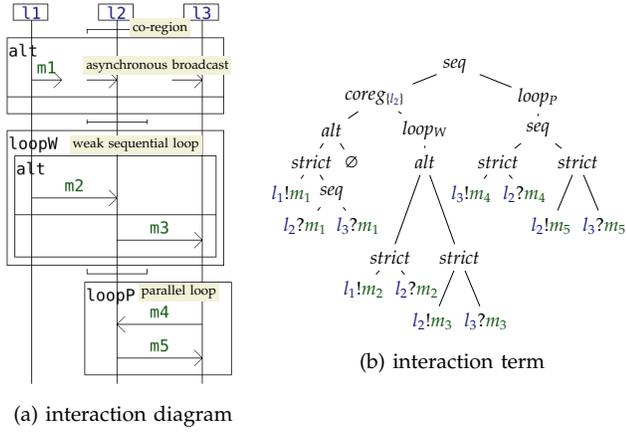}
            }
            \caption{interaction diagram\label{fig:interaction_diagram}}
        \end{subfigure}
    \end{minipage}
    \begin{minipage}{.58\linewidth}
        \centering
        \begin{subfigure}[t]{\linewidth}
            \centering
            \resizebox{\textwidth}{!}{
                \begin{tikzpicture}[every node/.style = {shape=rectangle, align=center}]
\node (o) { $seq$ } [sibling distance=3cm,level distance=0.6cm]
  child { node (o1) { $coreg_{\{\hlf{l_2}\}}$ } [sibling distance=1.75cm,level distance=0.6cm]
  child { node (o11) {$alt$} [sibling distance=.75cm]
    child { node (o111) {$strict$} [sibling distance=.75cm]
      child { node (o1111) {$\hlf{l_1}!\hms{m_1}$} }
      child { node (o1112) {$seq$} [sibling distance=1cm]
        child { node (o11121) {$\hlf{l_2}?\hms{m_1}$} }
        child { node (o11122) {$\hlf{l_3}?\hms{m_1}$} }
      }
    }
      child { node (o112) {$\varnothing$} }
  }
  child { node (o12) {$loop_W$}
    child { node (o121) {$alt$} [sibling distance=1.3cm,level distance=1.8cm]
      child { node (o1211) {$strict$} [sibling distance=1cm,level distance=.6cm]
        child { node (o12111) {$\hlf{l_1}!\hms{m_2}$} }
        child { node (o12112) {$\hlf{l_2}?\hms{m_2}$} }
      }
      child { node (o1212) {$strict$} [sibling distance=1cm,level distance=1.2cm]
        child { node (o12121) {$\hlf{l_2}!\hms{m_3}$} }
        child { node (o12122) {$\hlf{l_3}?\hms{m_3}$} }
      }
    }
  }
  }
  child { node (o2) {$loop_P$}
    child { node (o21) {$seq$} [sibling distance=1.5cm,level distance=.6cm]
      child { node (o211) {$strict$} [sibling distance=1cm,level distance=.6cm]
        child { node (o2111) {$\hlf{l_3}!\hms{m_4}$} }
        child { node (o2112) {$\hlf{l_2}?\hms{m_4}$} }
      }
      child { node (o212) {$strict$} [sibling distance=1cm,level distance=1.2cm]
        child { node (o2121) {$\hlf{l_2}!\hms{m_5}$} }
        child { node (o2122) {$\hlf{l_3}?\hms{m_5}$} }
      }
    }
  };
\end{tikzpicture}
            }
            \caption{interaction term\label{fig:interaction_term}}
        \end{subfigure}
    \end{minipage}
    \caption{An example of an interaction}
    \label{fig:interaction_example}
\end{figure}

Interaction models correspond to expressions built over:
\begin{itemize}
    \item the empty interaction, denoted by $\varnothing$, with the empty multi-trace $\varepsilon_C$ as the only accepted one 
    \item and actions $a$ (of the form $\hlf{l}!\hms{m}$ or $\hlf{l}?\hms{m}$) with a multi-trace reduced to a single action as the only accepted one (i.e. the multi-trace $a ~ \multiAppendLeft ~ \varepsilon_C$).
\end{itemize}

We then use operators to compose interactions into more complex expressions.
Let us consider two interactions $i_1$ and $i_2$:
\begin{itemize}
    \item $alt$ stands for alternative and a behavior of $alt(i_1,i_2)$ is either a behavior of $i_1$ or one of $i_2$ according to a non-deterministic and exclusive choice between the two alternatives.
    \item $strict$ stands for strict sequencing and a behavior of $strict(i_1,i_2)$ is such that a behavior from $i_1$ must be entirely expressed before any action from $i_2$ can occur;
    \item $coreg$ stands for concurrent region and corresponds to a family of operators $(coreg_r)_{r \in \mathcal{P}(L)}$. For a given subset $r \subseteq L$ of lifelines, $coreg_r(i_1,i_2)$ specifies behaviors composed from behaviors of $i_1$ and $i_2$.
    
    \indent
     (a) In the first case $r = L$, actions occurring in $i_1$ and $i_2$ can occur in any order in behaviors expressed by $coreg_r(i_1,i_2)$. 
        This definition coincides with a classical interleaving \cite{uml_interactions_meet_state_machines_an_institutional_approach,equivalence_of_denotational_and_operational_semantics_for_interaction_languages} or orthocurrence \cite{modeling_concurrency_with_partial_orders} operator. As such we denote by $par$ the corresponding derivable construct.
        
        \indent 
        (b) In the second case $r = \emptyset$, interleaving is only possible between actions that occur on different lifelines, i.e. the behaviors of $coreg_\emptyset(i_1,i_2)$ are defined as with the $strict$ operator for actions occurring on the same lifeline (whatever it may be) and as with the $par$ operator for actions occurring on different lifelines.
        This definition coincides with weak sequencing \cite{uml_interactions_meet_state_machines_an_institutional_approach,equivalence_of_denotational_and_operational_semantics_for_interaction_languages} which is a key operator for sequence diagrams. As such we denote by $seq$ the corresponding derivable construct.
        
         \indent (c) In the last case $r \not\in \{\emptyset,L\}$, $coreg_r$ behaves as $par$ on $r$ and as $seq$ on $L \setminus r$.
    
\end{itemize}

The $coreg$ operator is new w.r.t. the language from \cite{equivalence_of_denotational_and_operational_semantics_for_interaction_languages}. This operator is inspired by the co-regions of UML-SD \cite{UML}, also found in some papers on MSCs \cite{pomsets_for_message_sequence_charts}.
The $coreg$ construct allows certain patterns of communications that would be difficult to model ergonomically otherwise.
For instance, on Fig.\ref{fig:interaction_example} we have that (1) $\hlf{l_1}$ has to emit $\hms{m_1}$ (if it ever does) before it can emit $\hms{m_2}$ (if it ever does) and (2) $\hlf{l_2}$ can receive $\hms{m_1}$ or $\hms{m_2}$ in any order. To specify this, we use a co-region on lifeline $\hlf{l_2}$. This could not have been done using a $seq$ (as it would forbid $\hms{m_2}$ to be received before $\hms{m_1}$) or a $par$ (as it would allow $\hms{m_2}$ to be emitted before $\hms{m_1}$). 

$seq$ and $par$ are not primitive operators as they can be derived from $coreg$. However, because those two operators are familiar to users of sequence-diagram-like models and widely used, we keep them to denote the corresponding $coreg$ variants.

$strict$ and $coreg$ are binary scheduling operators i.e. they can be used to schedule behaviors w.r.t. one another. As a result, they can be used for defining repetition operators in the same manner as concatenation can be used to define the Kleene star for regular expressions.
In \cite{equivalence_of_denotational_and_operational_semantics_for_interaction_languages} we have defined $loop_S$, $loop_W$, $loop_P$ as repetition operators using resp. $strict$, $seq$ and $par$. In the same fashion we can define a family $(loop_{C_r})_{r \in \mathcal{P}(L)}$ of repetition operators.
In contrast to MSC and UML-SD, which only have a single loop construct, those loops enable us to specify a variety of behaviors.
 
$loop_S$ is a strict sequential loop meaning that any instance of the repeated behavior must be entirely executed (globally) before any other instance of the behavior might be started.

$loop_C$, used as $loop_{C_{r}}(i)$ with $r \subseteq L$ is a repetition using a $coreg_{r}$ operator. It is a middleground between:
    \begin{itemize}
        \item $loop_W = loop_{C_{\emptyset}}$ which corresponds to repetitions using the weak sequential operator. Several instances of the repeated behavior might exist at the same time because there is no synchronization between lifelines as for the beginning and end of the executions of those instances. Moreover, with $loop_W$, it might be so that the first action that is executed does not come from the first instance of the loop. For instance, in the example from Fig.\ref{fig:interaction_example}, after a first occurrence of $\hlf{l_1}!\hms{m_2}$, lifeline $\hlf{l_2}$ may emit $\hms{m_3}$ several times before receiving the $\hms{m_2}$ initially sent by $\hlf{l_1}$.
        \item and $loop_P = loop_{C_{L}}$, which is more akin to the bang operator of pi-calculus \cite{an_introduction_to_the_pi_calculus} and signifies the parallel composition of an arbitrary number of instances of the same behavior. It can be used to model services of which, at any given time, many instances may run in parallel.
    \end{itemize}

When modeling DS, communications between sub-systems can be defined up to a certain communication medium. In formalisms based on communicating automata, this often takes the form of buffers which assume certain policies (FIFO, bag, etc.) \cite{a_hierarchy_of_communication_models_for_message_sequence_charts}.
In our case, loops used in combination with asynchronous message passing, may be used to abstract away those communication media. For instance, while $loop_W(strict(\hlf{l_1}!\hms{m},\hlf{l_2}?\hms{m}))$ corresponds to having a FIFO buffer receiving messages from $\hlf{l_1}$ on $\hlf{l_2}$, in contrast, by using $loop_P(strict(\hlf{l_1}!\hms{m},\hlf{l_2}?\hms{m}))$ we rather have a bag buffer in so far as instances of message $\hms{m}$ can be received in any order.

\begin{definition}
\label{def:interactions}
The set $\mathbb{I}$ of interactions is the least term set s.t.:\\
\noindent $-$ $\varnothing$ and actions in $\mathbb{A}$ belong to $\mathbb{I}$\\
\noindent $-$ for any $i_1,i_2 \in \mathbb{I}^2$ and any $r \subseteq L$:\\
\noindent $\phantom{~~~~}-$ $\forall~ f \in \{strict,alt,coreg_{r}\}$, $f(i_1,i_2) \in \mathbb{I}$\\
\noindent $\phantom{~~~~}-$ $\forall~k \in \{S,C_{r}\}$, $loop_k(i_1) \in \mathbb{I}$
\end{definition}

In Def.\ref{def:interactions}, we formalize our interaction language. Let us keep in mind that the $seq$, $par$, $loop_W$ and $loop_P$ constructs are derivable from $coreg$ and $loop_C$.

Via their recursive definition, interactions have a tree-like structure, as illustrated on Fig.\ref{fig:interaction_term}. Those trees are binary-trees and we can pinpoint unambiguously each sub-tree via its position as a word $p \in \{1,2\}^*$ (with $\varepsilon$ the empty position which designates the root node). More precisely, $1$ (resp. $2$) allows access to the left direct sub-interaction or the unique direct sub-interaction (resp. the right direct sub-interaction). For any interaction $i$, $\positionsOf{i}$ designates the set of its positions, and, for any $p \in \positionsOf{i}$, $i_{|p}$ designates the sub-interaction at position $p$. For example, for the interaction $i =  seq(alt(\hlf{l_1}!\hms{m_1},\hlf{l_2}?\hms{m_2}),\hlf{l_1}!\hms{m_3})$, $i_{|1}$ is the interaction $alt(\hlf{l_1}!\hms{m_1},\hlf{l_2}?\hms{m_2})$, $i_{|12}$ is the interaction $\hlf{l_2}?\hms{m_2}$ and $\{\varepsilon,1,2,11,12\}$ is the set of positions of $i$.

\subsection{Semantics\label{ssec:semantics}}

Given a partition $C \in \partitionsOf{L}$ defining co-localizations, each interaction $i \in \mathbb{I}$ characterizes a (potentially infinite) set $\sigma_C(i)$ of multi-traces according to $C$. This semantics can be defined in an operational-style using either inductive rules (in the style of Plotkin \cite{a_structural_approach_to_operational_semantics}) as in \cite{equivalence_of_denotational_and_operational_semantics_for_interaction_languages} or through the definition of an execution function (in the style of functional programming languages) as in \cite{revisiting_semantics_of_interactions_for_trace_validity_analysis,a_small_step_approach_to_multi_trace_checking_against_interactions}.
In the following we propose a reworked (w.r.t. that of \cite{equivalence_of_denotational_and_operational_semantics_for_interaction_languages}) operational-style formulation which includes the coregion operator $coreg$ and involves fewer inductive predicates.

Interactions provide detailed control structures as for the occurrences and orders of actions, going beyond the simple linear order of events made available by multi-traces. An operational-style semantics defines accepted behaviors via concatenations of actions $a$ which occurrences are associated to term transformations of the form \[i \xrightarrow{a@p} i'\] 
Here interaction $i'$ specifies all the continuations of the behaviors specified by $i$ which start with the occurrence of action $a \in \mathbb{A}$, which, by construction, is a sub-term of $i$ at a certain position $p \in \positionsOf{i}$.

In order to ensure that the follow-up interaction $i'$ specifies the right order of actions following the first selected action ($a$ at position $p$) according to the interaction $i$, the execution relation $\rightarrow$ performs transformations on the initial term $i$ so as to obtain $i'$. Those transformations may include pruning operations (related to the notion of "permission" in \cite{high_level_message_sequence_charts}) which  clean the term with regard to the lifeline of the action which is executed.

\begin{figure}[h]
    \centering
    \begin{minipage}{.38\linewidth}
        \centering
        \begin{subfigure}[t]{\linewidth}
            \centering
            \resizebox{\textwidth}{!}{
                \input{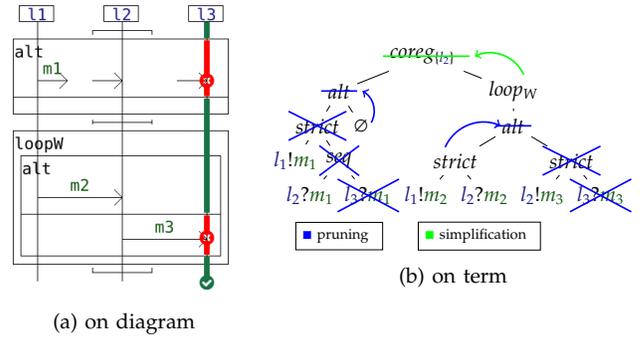}
            }
            \caption{on diagram\label{fig:pruning_diagram}}
        \end{subfigure}
    \end{minipage}
    \begin{minipage}{.58\linewidth}
        \centering
        \begin{subfigure}[t]{\linewidth}
            \centering
            \resizebox{\textwidth}{!}{
                \begin{tikzpicture}[every node/.style = {shape=rectangle, align=center}]
\node (o) { $coreg_{\{\hlf{l_2}\}}$ } [sibling distance=3cm,level distance=0.6cm]
  child { node (o1) {$alt$} [sibling distance=.75cm]
    child { node (o11) {$strict$} [sibling distance=.75cm]
      child { node (o111) {$\hlf{l_1}!\hms{m_1}$} }
      child { node (o112) {$seq$} [sibling distance=1cm]
        child { node (o1121) {$\hlf{l_2}?\hms{m_1}$} }
        child { node (o1122) {$\hlf{l_3}?\hms{m_1}$} }
      }
    }
      child { node (o12) {$\varnothing$} }
  }
  child { node (o2) {$loop_W$}
    child { node (o21) {$alt$} [sibling distance=2cm]
      child { node (o211) {$strict$} [sibling distance=1cm]
        child { node (o2111) {$\hlf{l_1}!\hms{m_2}$} }
        child { node (o2112) {$\hlf{l_2}?\hms{m_2}$} }
      }
      child { node (o212) {$strict$} [sibling distance=1cm]
        child { node (o2121) {$\hlf{l_2}!\hms{m_3}$} }
        child { node (o2122) {$\hlf{l_3}?\hms{m_3}$} }
      }
    }
  }
;
\draw[blue,thick] (o1122.south east) -- (o1122.north west);
\draw[blue,thick] (o1122.south west) -- (o1122.north east);
\draw[blue,thick] (o112.south east) -- (o112.north west);
\draw[blue,thick] (o112.south west) -- (o112.north east);
\draw[blue,thick] (o11.south east) -- (o11.north west);
\draw[blue,thick] (o11.south west) -- (o11.north east);
\draw[blue,thick] (o1.east) -- (o1.west);
\draw[->,blue,thick] ([xshift=5pt,yshift=-5pt] o12.north) to [bend right=45] ([xshift=2.5pt] o1.east);
\draw[blue,thick] (o2122.south east) -- (o2122.north west);
\draw[blue,thick] (o2122.south west) -- (o2122.north east);
\draw[blue,thick] (o212.south east) -- (o212.north west);
\draw[blue,thick] (o212.south west) -- (o212.north east);
\draw[blue,thick] (o21.east) -- (o21.west);
\draw[->,blue,thick] ([xshift=-5pt,yshift=-.5pt] o211.north) to [bend left=45] ([xshift=2.5pt] o21.west);
\draw[green,thick] (o.east) -- (o.west);
\draw[->,green,thick] ([xshift=5pt,yshift=-.5pt] o2.north) to [bend right=45] ([xshift=2.5pt] o.east);
\node[draw] (legP) at (-1.5,-3.1) { {\scriptsize \textcolor{blue}{$\blacksquare$} pruning} };
\node[draw,right=.6 of legP] (legS) { {\scriptsize \textcolor{green}{$\blacksquare$} simplification} };
\end{tikzpicture}
            }
            \caption{on term\label{fig:pruning_term}}
        \end{subfigure}
    \end{minipage}
    \caption{Pruning an interaction}
    \label{fig:pruning_example}
\end{figure}

In particular, the mechanism of pruning enters into play for handling weak sequencing. 
For example, let us consider executing $\hlf{l_3}!\hms{m_4}$ in the interaction from Fig.\ref{fig:interaction_example}. 
If $\hlf{l_3}!\hms{m_4}$ is the first action to occur then, so as to respect the top to bottom order of the diagram (i.e. weak sequencing), this means that neither $\hlf{l_3}?\hms{m_1}$ nor $\hlf{l_3}?\hms{m_3}$ can occur. 
Indeed, they appear above $\hlf{l_3}!\hms{m_4}$ along the lifeline $\hlf{l_3}$ in Fig.\ref{fig:interaction_example}, more precisely they are scheduled with weak sequencing w.r.t. to it and must precede it.
Hence, if those actions were to occur they would have done so before $\hlf{l_3}!\hms{m_4}$. 
As a result, both actions must be eliminated or, in other words, to better conform to our vocabulary, pruned from the follow-up interaction. This is possible because they are within alternatives and loops. 
The general idea is to transform the sub-terms preceding (i.e. with sequencing) the action that is executed (here $\hlf{l_3}!\hms{m_4}$) in such a way as to eliminate from these sub-terms actions that involve the lifeline on which this executed action occurs (here $\hlf{l_3}$).
In this process, which we call {\em pruning}, pertinent actions are eliminated and, from the bottom up, the {\em pruned} sub-terms are reconstructed so as to keep all the behaviors that do not involve a certain lifeline (here $\hlf{l_3}$).

This process of pruning is illustrated on Fig.\ref{fig:pruning_example}. 
The interaction term with $coreg_{\{\hlf{l_2}\}}$ as root in Fig.\ref{fig:pruning_term} is simplified with two purposes: eliminate all traces with an action on lifeline $\hlf{l_3}$ and preserve all other accepted traces. As the sub-interaction $strict(\hlf{l_1}!\hms{m_1},seq(\hlf{l_2}?\hms{m_1},\hlf{l_3}?\hms{m_1}))$ at position $1$ only accepts traces containing the action $\hlf{l_3}?\hms{m_1}$, the first alternative of the $alt$ operator (position $11$) is no longer allowed and the sub-interaction with $alt$ as top operator is reduced to its second alternative, here $\varnothing$ (position $12$), which by definition accepts only the empty trace and consequently avoids lifeline $\hlf{l_3}$. As can be seen from Fig.\ref{fig:pruning_diagram}, the process of pruning an interaction is a local transformation guided by the lifelines to be avoided.

In Def.\ref{def:pruning} below, we define pruning w.r.t. a subset $L' \subseteq L$ of lifelines. The two pruning relations $\isPruneOf{L'}$ and $\isNotPruneOf{L'}$ are defined inductively on the term structure of interactions. For any interactions $i$ and $i'$:
\begin{itemize}
    \item $i \isPruneOf{L'} i'$ signifies that $i'$ is an interaction which specifies exactly all the behaviors specified by $i$ that do not involve any action  occurring on a lifeline of $L'$.
    \item $i \isNotPruneOf{L'}$ signifies that it is impossible to find such an interaction $i'$ because all the behaviors specified by $i$ involve at least one action occurring on a lifeline of $L'$.
\end{itemize}

\begin{definition}[Pruning]
\label{def:pruning}
The pruning relations $\isPruneBase \; \subset \mathbb{I} \times \mathcal{P}(L) \times \mathbb{I}$ and $\isNotPruneBase \; \subset \mathbb{I} \times \mathcal{P}(L)$ are s.t. for any $L' \subseteq L$, any $f \in \{strict\} \cup \bigcup_{r \subseteq L} \{coreg_{r}\}$ and $k \in \{S\} \cup \bigcup_{r \subseteq L} \{C_{r}\}$:

{
\centering
\begin{minipage}{2.5cm}
\centering
\begin{scprooftree}{.9}
\AxiomC{\vphantom{$\theta(a) \not\in L'$} \vphantom{$\isPruneOf{L'}$}}
\UnaryInfC{$\varnothing \isPruneOf{L'} \varnothing$}
\end{scprooftree}
\end{minipage}
\begin{minipage}{2.75cm}
\centering
\begin{scprooftree}{.9}
\AxiomC{\vphantom{$\theta(a) \not\in L'$} \vphantom{$\isPruneOf{L'}$}}
\RightLabel{\scriptsize $\theta(a) \not\in L'$}
\UnaryInfC{$a \isPruneOf{L'} a$}
\end{scprooftree}
\end{minipage}
\begin{minipage}{2.75cm}
\centering
\begin{scprooftree}{.9}
\AxiomC{\vphantom{$\theta(a) \not\in L'$} \vphantom{$\isPruneOf{L'}$}}
\RightLabel{\scriptsize $\theta(a) \in L'$}
\UnaryInfC{$a \isNotPruneOf{L'}$}
\end{scprooftree}
\end{minipage}

\begin{minipage}{3.9cm}
\centering
\begin{scprooftree}{.9}
\AxiomC{$i_1 \isPruneOf{L'} i_1'$}
\AxiomC{$i_2 \isPruneOf{L'} i_2'$}
\BinaryInfC{$alt(i_1,i_2) \isPruneOf{L'} alt(i_1',i_2')$}
\end{scprooftree}
\end{minipage}
\begin{minipage}{3.9cm}
\centering
\begin{scprooftree}{.9}
\AxiomC{$i_1 \isNotPruneOf{L'}$}
\AxiomC{$i_2 \isNotPruneOf{L'}$}
\BinaryInfC{$alt(i_1,i_2) \isNotPruneOf{L'}$}
\end{scprooftree}
\end{minipage}

\begin{minipage}{5cm}
\centering
\begin{scprooftree}{.9}
\AxiomC{$i_j \isPruneOf{L'} i_j'$}
\AxiomC{$i_{j'} \isNotPruneOf{L'}$}
\RightLabel{\scriptsize $\{j,j'\} = \{1,2\}$}
\BinaryInfC{$alt(i_1,i_2) \isPruneOf{L'} i_j'$}
\end{scprooftree}
\end{minipage}

\begin{minipage}{3.9cm}
\centering
\begin{scprooftree}{.9}
\AxiomC{$i_1 \isPruneOf{L'} i_1'$}
\AxiomC{$i_2 \isPruneOf{L'} i_2'$}
\BinaryInfC{$f(i_1,i_2) \isPruneOf{L'} f(i_1',i_2')$}
\end{scprooftree}
\end{minipage}
\begin{minipage}{3.9cm}
\centering
\begin{scprooftree}{.9}
\AxiomC{$i_j \isNotPruneOf{L'}$}
\RightLabel{\scriptsize $j \in \{1,2\}$}
\UnaryInfC{$f(i_1,i_2) \isNotPruneOf{L'}$}
\end{scprooftree}
\end{minipage}

\begin{minipage}{3.9cm}
\centering
\begin{scprooftree}{.9}
\AxiomC{$i_1 \isPruneOf{L'} i_1'$}
\UnaryInfC{$loop_k(i_1) \isPruneOf{L'} loop_k(i_1')$}
\end{scprooftree}
\end{minipage}
\begin{minipage}{3.9cm}
\centering
\begin{scprooftree}{.9}
\AxiomC{$i_1 \isNotPruneOf{L'}$}
\UnaryInfC{$loop_k(i_1) \isPruneOf{L'} \varnothing$}
\end{scprooftree}
\end{minipage}

}
\end{definition}

The pruning relations are defined inductively in the style of Plotkin \cite{a_structural_approach_to_operational_semantics}:
\begin{itemize}
    \item the empty interaction $\varnothing$ can always be pruned into $\varnothing$ w.r.t. any subset $L' \subseteq L$ of lifelines (i.e. $\varnothing \isPruneOf{L'} \varnothing$) because it expresses no action occurring on $L'$
    \item for any action $a$, we have $a \isPruneOf{L'} a$ if $\theta(a) \not\in L'$ and $a \isNotPruneOf{L'}$ otherwise because $a$ must be expressed
    \item having $alt(i_1,i_2) \isNotPruneOf{L'}$ (resp. $strict(i_1,i_2) \isNotPruneOf{L'}$) requires that both (resp. one of the two) $i_1 \isNotPruneOf{L'}$ and $i_2 \isNotPruneOf{L'}$ hold
    \item all other cases are handled similarly.
\end{itemize}

We define the execution relation $i \xrightarrow{a@p} i'$ - which makes interactions executable - in the same way as the pruning predicates. Executing an atomic action $a \in \mathbb{A}$ simply consists in replacing it with the empty interaction $\varnothing$ because once action $a$ is expressed, nothing remains to be expressed. If an action can be expressed within a branch of an alternative, then it can also be expressed from the alternative itself but its expression forces the choice of the branch on which it occurs to be made. An action which can be expressed on the left branch of a scheduling operator ($strict$ or any $coreg_{r}$ which implies also $seq$ and $par$) can always be expressed from the scheduling operator itself and what remains to be expressed is the scheduling of what remains of the left branch w.r.t. the initial right branch. 

Defining $i \xrightarrow{a@p} i'$ for executing actions from the right branch of $i$ can be more challenging.
Indeed, it is not always possible to express an action on the right branch of a scheduling operator, and, if it is possible, then it often requires pruning the left branch so as to remove inconsistencies in the follow-up interaction. 

When an action is expressed inside a loop, we have in the follow-up a certain scheduling of what remains to be executed of the instance of the sub-behavior w.r.t. the initial loop (which, serving as a specification of the repeatable instance, remains the same). Due to the peculiarities of weak sequencing (as evoked in \cite{equivalence_of_denotational_and_operational_semantics_for_interaction_languages}), in particular to the fact that the first action that is executed does not necessarily come from the first instance of the loop (as ordered by the operator which schedule different instances of the loop), the rule for $loop_{C_{r}}$ is somewhat more complex. It involves scheduling a pruned version of the initial loop before the remainder of the executed instance.

\begin{definition}
[Execution]
\label{def:execution}
The execution relation $\rightarrow \subset \mathbb{I} \times (\mathbb{A} \times \{1,2\}^*) \times \mathbb{I}$ is defined as follows:

{
\centering
\begin{minipage}{2cm}
\begin{scprooftree}{.9}
\AxiomC{\phantom{$\xrightarrow{a@p}$}}
\UnaryInfC{$a \xrightarrow{a@\varepsilon} \varnothing$}
\end{scprooftree}
\end{minipage}
\begin{minipage}{3.3cm}
\begin{scprooftree}{.9}
\AxiomC{$i_1 \xrightarrow{a@p} i'_1$}
\UnaryInfC{$alt(i_1,i_2) \xrightarrow{a@1.p} i'_1$}
\end{scprooftree}
\end{minipage}
\begin{minipage}{3.3cm}
\begin{scprooftree}{.9}
\AxiomC{$i_2 \xrightarrow{a@p} i'_2$}
\UnaryInfC{$alt(i_1,i_2) \xrightarrow{a@2.p} i'_2$}
\end{scprooftree}
\end{minipage}

\begin{minipage}{4.375cm}
\begin{scprooftree}{.9}
\AxiomC{$i_1 \xrightarrow{a@p} i'_1$}
\UnaryInfC{$strict(i_1,i_2) \xrightarrow{a@1.p} strict(i'_1,i_2)$}
\end{scprooftree}
\end{minipage}
\begin{minipage}{4.3cm}
\begin{scprooftree}{.9}
\AxiomC{$i_1 \xrightarrow{a@p} i'_1$\vphantom{$\isPruneOf{\{\theta(a)\}\setminus r}$}}
\UnaryInfC{$coreg_{r}(i_1,i_2) \xrightarrow{a@1.p} coreg_{r}(i'_1,i_2)$}
\end{scprooftree}
\end{minipage}

\begin{minipage}{4.275cm}
\begin{scprooftree}{.9}
\AxiomC{$i_1 \isPruneOf{L} \varnothing$}
\AxiomC{$i_2 \xrightarrow{a@p} i'_2$}
\BinaryInfC{$strict(i_1,i_2) \xrightarrow{a@2.p} i'_2$}
\end{scprooftree}
\end{minipage}
\begin{minipage}{4.4cm}
\begin{scprooftree}{.9}
\AxiomC{$i_1 \isPruneOf{\{\theta(a)\}\setminus r} i_1'$}
\AxiomC{$i_2 \xrightarrow{a@p} i'_2$}
\BinaryInfC{$coreg_{r}(i_1,i_2) \xrightarrow{a@2.p} coreg_{r}(i'_1,i'_2)$}
\end{scprooftree}
\end{minipage}

\begin{minipage}{6cm}
\begin{scprooftree}{.9}
\AxiomC{$i_1 \xrightarrow{a@p} i_1'$}
\UnaryInfC{$loop_S(i_1) \xrightarrow{a@1.p} strict(i_1',loop_S(i_1))$}
\end{scprooftree}
\end{minipage}
 
\begin{minipage}{7cm}
\begin{scprooftree}{.9}
\AxiomC{$i_1 \xrightarrow{a@p} i_1'$}
\AxiomC{$loop_{C_{r}}(i_1) \isPruneOf{\{\theta(a)\} \setminus r} i'$}
\BinaryInfC{$loop_{C_{r}}(i_1) \xrightarrow{a@1.p} coreg_{r}(i',coreg_{r}(i_1',loop_{C_{r}}(i_1)))$}
\end{scprooftree}
\end{minipage}

}
\end{definition}

In \cite{equivalence_of_denotational_and_operational_semantics_for_interaction_languages}, the reader can find detailed explanations on the rules concerning most operators ($strict$, $par$, $alt$, $seq$ and $loop$). In addition, Def.\ref{def:execution} includes the use of positions which is not the case in \cite{equivalence_of_denotational_and_operational_semantics_for_interaction_languages}. Using those decorations $@p$ comes at no cost given the inductive nature of $\rightarrow$'s definition. Moreover, it removes all ambiguities related to having several executable occurrences of the same action (i.e. $i \xrightarrow{a@p_1} i_1$ and $i \xrightarrow{a@p_2} i_2$ with $p_1 \neq p_2$ and $i_1 \neq i_2$).

\begin{figure}[h]
    \centering
\resizebox{.475\textwidth}{!}{\input{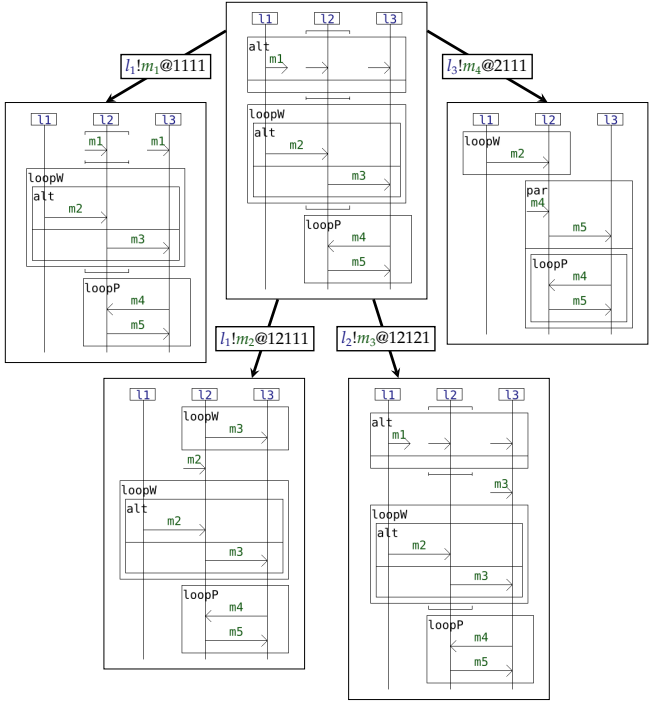}}
    \caption{Follow-ups of the interaction from Fig.\ref{fig:interaction_example}}
    \label{fig:followup}
\end{figure}

Fig.\ref{fig:followup} illustrates the use of the execution relation on the example from Fig.\ref{fig:interaction_example}. Four distinct actions can be immediately executed, leading to four follow-up interactions.

This process of computing follow-up interactions can be repeated recursively so that we obtain a tree which root is the initial interaction $i_0$. This tree, called an {\em execution tree} represents the semantics is $i_0$ i.e. its set of accepted behaviors. behaviors expressed by $i_0$ can indeed be observed via the succession of the actions that are executed on any path of the tree starting from $i_0$ and ending with an interaction that can express the empty behavior.

Finally, the set of multi-traces $\sigma_C(i_0)$ accepted by an interaction $i_0$ can be built using the execution relation $\rightarrow$. A multi-trace $\mu$ belongs to $\sigma_C(i_0)$ iff it can be written as:
\[
\mu = a_1 ~ \multiAppendLeft ~ a_2 ~ \multiAppendLeft ~ \ldots ~ \multiAppendLeft ~ a_n ~ \multiAppendLeft ~ \varepsilon_C
\] 
and if there exist $n$ interactions $i_1,\ldots,i_n$ s.t.:
\begin{itemize}
    \item $\forall~j \in \{0,1,\ldots,n-1\}$,  $i_j \xrightarrow{a_{j+1}@p_{j+1}} i_{j+1}$ \item and $\varepsilon_C \in \sigma_C(i_n)$. 
\end{itemize}
This last point (i.e. whether or not $\varepsilon_C \in \sigma_C(i_n)$) can be determined statically using the pruning relation. Indeed, if there exists $i_n'$ such that $i_n \isPruneOf{L} i'_n$ then this means that $i_n$ has at least a behavior which does not involve any action that occurs on lifelines of $L$. $L$ being the set of all lifelines, this behavior can only correspond to the empty multi-trace $\varepsilon_C$ and hence $\varepsilon_C \in \sigma_C(i_n)$.

\begin{definition}
\label{def:semantics}
Let $C \in \partitionsOf{L}$ and let $i \in \mathbb{I}$. 

The semantics $\sigma_C(i)$ of $i$ is the least subset of $\mathbb{M}_C$ s.t.:

\begin{minipage}{3.5cm}
\begin{scprooftree}{.9}
\AxiomC{$i \isPruneOf{L} i'$}
\UnaryInfC{$\varepsilon_C \in \sigma_C(i)$}
\end{scprooftree}
\end{minipage}
\begin{minipage}{4.5cm}
\begin{scprooftree}{.9}
\AxiomC{$\mu \in \sigma_C(i')$}
\AxiomC{$i \xrightarrow{a@p} i'$}
\BinaryInfC{$a~ \multiAppendLeft ~\mu \in \sigma_C(i)$}
\end{scprooftree}
\end{minipage}

with $\mu \in \mathbb{M}_C$, $a \in \mathbb{A}$, $p \in \{1,2\}^*$ and $i' \in \mathbb{I}$.
\end{definition}

\subsection{Soundness of the operational semantics}

In \cite{equivalence_of_denotational_and_operational_semantics_for_interaction_languages} we have given a denotational semantics for the interaction language without the $coreg$ and $loop_C$ operators. The semantics is based on the use of composition and algebraic operators as in \cite{uml_interactions_meet_state_machines_an_institutional_approach}.
In \cite{equivalence_of_denotational_and_operational_semantics_for_interaction_languages}, we used the $\globalInterleaving$ (interleaving) and $\globalWeakSeq$ (weak sequencing) operators on sets of traces.
In order to include $coreg$ we have to define a new operator on sets of traces as follows.
The first step is to define a conditional conflict predicate $t \; \globalConditionalSequencing{r} \; l$ meaning that the trace $t$ contains an action on a lifeline $l \not\in r$:
\[
\begin{array}{lll}
\varepsilon ~\globalConditionalSequencing{r}~ l &=& \bot \\
(a.t) ~\globalConditionalSequencing{r}~ l &=&  ((\theta(a) = l)\wedge (l \not\in r)) \vee (t ~\globalConditionalSequencing{r}~ l)
\end{array}
\]
If $t ~\globalConditionalSequencing{r}~ l = \top$, we say that the trace $t$ has conflicts w.r.t. the lifeline $l$ in the region covered by $L \setminus r$ where $r \subseteq L$ is the concurrent region. By overloading the symbol $\globalConditionalSequencing{r}$, the set $t_1 ~\globalConditionalSequencing{r}~ t_2$ of conditional sequencing of $t_1$ and $t_2$ is defined by:
\[
\begin{array}{lll}
\varepsilon ~\globalConditionalSequencing{r}~ t_2 & =  & \{ t_2 \}\\
t_1 ~\globalConditionalSequencing{r}~ \varepsilon & =  & \{ t_1 \} \\
(a_1.t_1) ~\globalConditionalSequencing{r}~ (a_2.t_2) & = &
\{ a_1.t ~|~ t \in t_1 ~\globalConditionalSequencing{r}~ (a_2.t_2) \} \\
& & 
\cup \left\{ 
a_2.t
\middle|
\begin{array}{l}
t \in (a_1.t_1) ~\globalConditionalSequencing{r}~ t_2,\\
\neg (a_1.t_1 ~\globalConditionalSequencing{r}~ \theta(a_2))
\end{array}
\right\}
\end{array}
\]

This conditional sequencing operator entirely covers previous notions of weak sequencing and interleaving as we have $\globalInterleaving = \globalConditionalSequencing{L}$ and $\globalWeakSeq = \globalConditionalSequencing{\emptyset}$. Formal proofs of those statements are given in Appendix \ref{anx:proof_conditional_sequencing} which follows the structure of the Coq proof available in \cite{coq_hibou_label_eqsem_with_coregions}.

We denote by $\globalStrictSeq$ the concatenation (strict sequencing) operator on multi-traces, and by $\cup$ the set-theoretic union.

Then, in order to define the semantics of loops, we extend the Kleene star notation for our new scheduling operator.
For any $\diamond \in \{ \globalStrictSeq ,~\globalConditionalSequencing{r} \}$ and any set of traces $T$, the Kleene closure $T^{\diamond *}$ of $T$ is defined by $T^{\diamond *} = \bigcup_{\substack{j \in \mathbb{N}}} T^{\diamond j}$ 
with $T^{\diamond 0}= \{ \varepsilon \}$ and $T^{\diamond j} = T \diamond T^{\diamond (j-1)}$ for $j > 0$.

Finally, we define a denotational semantics $\rho : \mathbb{I} \rightarrow \mathcal{P}(\mathbb{T}_{|L})$ of interactions as a set of global traces - i.e. traces defined on the partition $C_t = \{L\}$ - associated to a specific interaction term:
\begin{itemize}
    \item $\rho(\varnothing) = \{\varepsilon\}$ and $\rho(a) = \{a\}$ for any action $a \in \mathbb{A}$
    \item and for any $i_1$ and $i_2$ in $\mathbb{I}$:
    \begin{itemize}
        \item $\rho(alt(i_1,i_2) = \rho(i_1) \cup \rho(i_2)$
        \item $\rho(strict(i_1,i_2)) = \rho(i_1) \globalStrictSeq \rho(i_2)$
        \item $\rho(coreg_r(i_1,i_2)) = \rho(i_1) \globalConditionalSequencing{r} \rho(i_2)$
        \item $\rho(loop_S(i_1)) = \rho(i_1)^{\globalStrictSeq *}$
        \item $\rho(loop_{C_r}(i_1)) = \rho(i_1)^{\globalConditionalSequencing{r} *}$
    \end{itemize}
\end{itemize}

The pruning and execution relations are then characterized w.r.t. this denotational formulation in Th.\ref{th:prune_denotational} and Th.\ref{th:exec_denotational}.

\begin{theorem}
\label{th:prune_denotational}
For any $L' \subseteq L$ and any $i$ and $i'$ from $\mathbb{I}$ we have:
\[
\begin{array}{ccc}
(i \isPruneOf{L'} i')
&
\Rightarrow 
&
(\rho(i') = \{ t \in \rho(i) ~|~ \forall~l\in L',~ \neg (t ~\globalConditionalSequencing{\emptyset}~ l) \})
\\
(i \isNotPruneOf{L'})&
\Rightarrow 
&
(\forall~t \in \rho(i),~\exists~l \in L',~ t ~\globalConditionalSequencing{\emptyset}~ l)
\end{array}
\]
\end{theorem}

\begin{proof}
A detailed proof is given in Appendix \ref{anx:proof_pruning} and corresponds to the Coq proof available in \cite{coq_hibou_label_eqsem_with_coregions}.
\end{proof}

Th.\ref{th:prune_denotational} states that transformations $i \isPruneOf{L'} i'$ characterize interactions $i'$ which specify behaviors that are exactly those specified by $i$ with no action occurring on a lifeline of $L'$.
By contrast, $i \isNotPruneOf{L'}$ stands for "all the behaviors of $i$ involve at least one action occurring on a lifeline of $L'$".
Let us observe that if we choose $L' = L$ in Def.\ref{def:pruning}, the expression $i \isPruneOf{L} \varnothing$ (resp. $i \isNotPruneOf{L}$) means that the interaction $i$ accepts (resp. does not accept) the empty multi-trace $\varepsilon_C$. We take advantage of this observation in Def.\ref{def:semantics} to provide a more compact presentation of the operational semantics than the one in \cite{equivalence_of_denotational_and_operational_semantics_for_interaction_languages}.

\begin{theorem}
\label{th:exec_denotational}
For any $i \in \mathbb{I}$, $a \in \mathbb{A}$ and $t \in \mathbb{T}_{|L}$:
\[
(a.t \in \rho(i))
\Leftrightarrow
\left(
\exists~i' \in \mathbb{I},~
\begin{array}{c}
(i \xrightarrow{a} i')
\wedge 
(t \in \rho(i'))
\end{array}
\right)
\]
\end{theorem}

\begin{proof}
A detailed proof is given in Appendix \ref{anx:proof_execution} and corresponds to the Coq proof available in \cite{coq_hibou_label_eqsem_with_coregions}.
\end{proof}

Th.\ref{th:exec_denotational} characterizes the execution relation $\rightarrow$ w.r.t. the denotational-style semantics $\rho$. It states that the follow ups $i'$ s.t. $i \xrightarrow{a@p} i'$ indeed specify all the continuations of the behaviors specified by $i$.

Finally, in Th.\ref{th:equivalence_semantics} we justify the correctness of our operational-style semantics  from Def.\ref{def:semantics} for the particular case of the partition $C_t = \{L\}$ w.r.t. the denotational semantics $\rho$ inspired by \cite{uml_interactions_meet_state_machines_an_institutional_approach}.

\begin{theorem}
\label{th:equivalence_semantics}
For any $i \in \mathbb{{I}}$, we have $\sigma_{C_t}(i) = \rho(i)$
\end{theorem}

\begin{proof}
Implied by Th.\ref{th:prune_denotational} (for $L'=L$) and by Th.\ref{th:exec_denotational}. See Appendix \ref{anx:proof_eqsem} and the Coq proof in \cite{coq_hibou_label_eqsem_with_coregions}.
\end{proof}

Th.~\ref{th:equivalence_semantics} states that both definitions of $\rho$ and $\sigma_{C_t}$ coincide on the trivial partition $C_t = \{L\}$ which preserves the most information on partial orders between events. For other partitions $C \in \partitionsOf{L}$, the analogy of denotational and semantic semantics is obtained by observing that $\sigma_{C}(i)$ corresponds to a projection of $\sigma_{C_t}(i)$ from finer to coarser multi-traces, and thus corresponds also to the same projection applied on $\rho(i)$.

\subsection{Application to multi-trace analysis}

Accepted multi-traces of a certain interaction $i$ are defined via their semantics $\sigma_C(i)$ in Def.\ref{def:semantics}. If, in a practical setting, a multi-trace $\mu$ is observed during a system execution, the conformance of $\mu$ to $i$ can be brought back to a problem of membership as in \cite{a_small_step_approach_to_multi_trace_checking_against_interactions} i.e. verifying whether or not $\mu \in \sigma_C(i)$.

\begin{figure*}[ht]
    \centering

\scalebox{.5}{\input{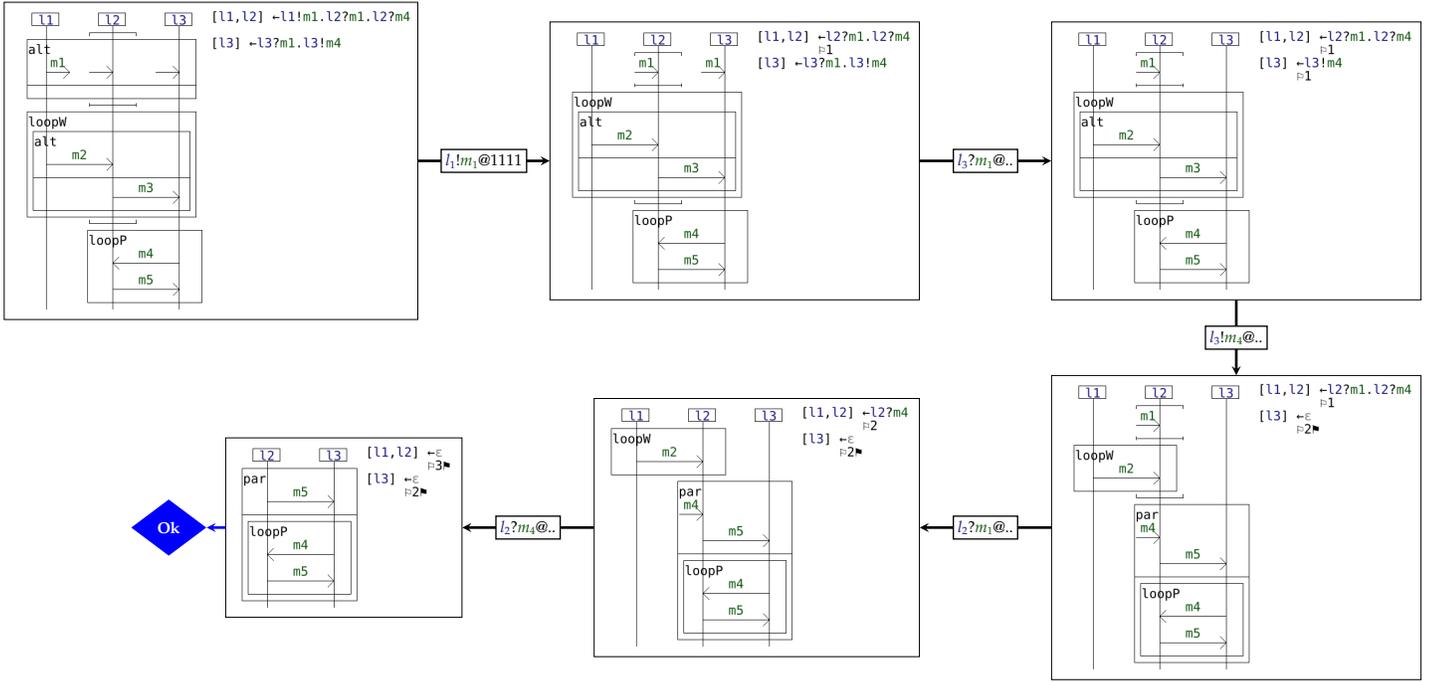}}

    \caption{Illustrating multi-trace analysis {\scriptsize algorithm adapted (to co-localizations) from \cite{a_small_step_approach_to_multi_trace_checking_against_interactions}}.}
    \label{fig:previous_algo}
\end{figure*}

In \cite{revisiting_semantics_of_interactions_for_trace_validity_analysis}, we have proposed an algorithm for analyzing traces w.r.t. interactions, which corresponds to a single co-localization in our present framework, i.e. with the trivial partition $C_t = \{L\}$. These analyses determine whether or not a behavior given as a trace is accepted. The principle of the algorithm is to consider the first element $a_1$ of the trace $t$ to be analyzed (therefore of the form $a_1.t'$), to execute it in the reference interaction $i_0$ and to remove it from the trace. This allows us to start again with all the interactions $i_1$ verifying $i_0 \xrightarrow{a_1@p} i_1$ for (possibly several) positions $p$ and the remaining trace $t'$.
If the original trace $t$ of length $n$ can be emptied (via $n$ such steps), then it means that it is accepted by the original interaction $i_0$ iff $\varepsilon$ is accepted by the last interaction $i_n$. If this is not possible then this means that the behavior $t$ deviates from $i_0$.

This principle can be directly adapted to multi-traces as demonstrated with the algorithm defined for the discrete partition $C_d = \{\{l\} ~|~ l \in L\}$ in \cite{a_small_step_approach_to_multi_trace_checking_against_interactions}. 
The latter algorithm can be easily extended to consider any partition $C \in \partitionsOf{L}$ and to consider global prefixes. Let us illustrate this with our running example i.e. let us analyze the multi-trace from Fig.\ref{fig:mu_coloc2clocks} w.r.t. the interaction from Fig.\ref{fig:interaction_example} which semantics we have illustrated in Fig.\ref{fig:followup}.

The application of the algorithm is represented on Fig.\ref{fig:previous_algo}.
In order to analyze the multi-trace, we try to reconstruct a global behavior (global trace) from the execution tree (e.g. Fig.\ref{fig:followup}) of the interaction which can be projected into the multi-trace.
To that end, we use the execution relation $\rightarrow$ from Def.\ref{def:execution} (the operational semantics).
If an action is at the beginning of one of the multi-trace's local traces, and if it is immediately executable in the interaction model, the algorithm performs a step in which it both consumes it from the multi-trace and executes it in the interaction.
This enables one to replay a behavior characterized by the multi-trace in the model. If it is possible then it means that the multi-trace satisfies the specifying interaction. Otherwise this means that the multi-trace violates it.
The analysis itself is represented graphically under the form of a graph on Fig.\ref{fig:previous_algo}.
Every node on this graph contains both an interaction (on the left) and a multi-trace (on the right), the initial node containing the initial interaction $i_0$ (Fig.\ref{fig:interaction_example}) and multi-trace $\mu_0$ (Fig.\ref{fig:mu_coloc2clocks}). As a visual aid, a \faFlagO~ (white flag) symbol appears under components $c \in C$ of the multi-trace whenever, at this moment in the reproduction of the behavior in the model, observation has started on $c$. The number following the flag then represents the number of actions which have been observed. A \faFlag~ (black flag) symbol marks the end of observation.

\begin{figure}[h]
    \centering
    
    \begin{subfigure}[t]{\linewidth}
        \centering
        $\shortColGrey{l_1!m_1}.\hlf{l_3}?\hms{m_1}.\shortColGrey{l_2?m_1.l_3!m_4}.\hlf{l_2}?\hms{m_4}$
        \caption{Missing actions in the multi-trace from Fig.\ref{fig:mu_coloc_partial} transposed in the global scenario from Fig.\ref{fig:mu_global1clock}\label{fig:missing_actions}}
    \end{subfigure}

    \vspace*{.5cm}
    
    \begin{subfigure}[t]{\linewidth}
        \centering
        \scalebox{.5}{\input{figures/3_4_limit_ana/graph}}
        \caption{The analysis yields a $Fail$ because of unobserved actions\label{fig:limits_ana}}
    \end{subfigure}
    
    \caption{Limitation of the approach from \cite{a_small_step_approach_to_multi_trace_checking_against_interactions} under partial observation.}
    \label{fig:limits}
\end{figure}

However, this general principle is insufficient for analyzing multi-traces in the case of partial observability.
In this context, {\em partial observation} signifies that the multi-trace logged by the instrumentation does not characterize the entire execution of the DS. More concretely, some events may be missing from the multi-trace w.r.t. an ideal multi-trace which would have been observed with ideal conditions of observation. The notion of {\em multi-trace slice} from Sec.\ref{ssec:slices} proposes a specific definition of partial observation, where events may be missing at the beginning and/or the end of each local trace component of the multi-trace (independently).

As a means to understand how partial observation is challenging for RV we consider the example from Fig.\ref{fig:mu_coloc_partial}, which is a partial observation of the multi-trace from Fig.\ref{fig:mu_coloc2clocks}. 
If we were to reorder actions globally, this observation could be described as in Fig.\ref{fig:missing_actions}. 
Missing actions, i.e. actions that are not observed by the instrumentation, are not necessarily at the beginning or the end globally but there may be several sub-words missing from the global trace (here $\shortColGrey{l_1!m_1}$ and $\shortColGrey{l_2?m_1.l_3!m_4}$ inserted in light gray in the global trace given in Fig.\ref{fig:missing_actions}).
The fact that such missing actions may be located anywhere in a globally sequential behavior (here the broadcast of $\hms{m_1}$ from $\hlf{l_1}$ to $\hlf{l_2}$ and $\hlf{l_3}$ followed by the passing of $\hms{m_4}$ from $\hlf{l_3}$ to $\hlf{l_2}$) is particularly challenging for ORV.
Using the algorithm from \cite{a_small_step_approach_to_multi_trace_checking_against_interactions} this would yield a $Fail$ verdict - as illustrated on Fig.\ref{fig:limits_ana} - because this algorithm cannot differentiate between the system going out of specification and it being partially observed.

This motivates the definition of an ORV algorithm that is tolerant to partial observation, which is the object of the next section.

\section{ORV algorithm with bounded simulation\label{sec:algo}}

As we have seen with the example from Fig.\ref{fig:mu_coloc_partial} and Fig.\ref{fig:limits}, DS executions can be partially observed due to issues of synchronization between local observers. As a result, a correct behavior may be observed as a slice $\mu' \in \sliceOf{\mu}$ of an accepted multi-trace $\mu \in \sigma_C(i)$ with missing elements at the beginning and/or the end of the traces corresponding to co-localizations of $C$. Then, because we might have $\mu' \not\in \sigma_C(i)$, membership is not enough to verify conformance. The property which we have to verify is rather whether or not $\mu' \in \sliceOf{\sigma_C(i)}$.

Simulation is a straightforward answer to partial observation in so far as actions missing from the multi-trace may simply be simulated in the model.
In order to identify a slice $\mu'$ of an accepted multi-trace, we may simulate actions $a$ that occur either before the first action of the corresponding component $\mu'_{|\theta_C(a)}$ or after its last action i.e. outside of the period of (continuous) observation of the component.
Simulating such actions hopefully enables the consumption of further actions in the multi-trace (in the same component or any other).
The approach that we propose is optimistic in so far as that it suffices that there exists some missing actions that can be be simulated which explain the observed behavior. In other words, this means that if, during an execution, what we observe of it can be explained by a behavior without violations of the specification, then, it is accepted even though a violation might have happened in some unoberseved part of the behavior.

Simulation explores possible missing actions that could have been executed in order to explain the behavior observed via the multi-trace w.r.t. the specifying interaction.
However, the presence of loops in interaction models makes it possible to simulate arbitrarily many actions, making a naive simulation-based algorithm non-terminating.
As a practical solution, we should bind simulation up to a certain criterion so that we can hope to find missing actions within a finite search space.

Defining pertinent stopping criteria on simulation being no trivial matter, we firstly formalize a simulation-based algorithm using an abstract criterion.

\subsection{Algorithm initialization\label{ssec:orv_prelim}}
\label{sec:algo-init}

Our algorithm relies on two mechanisms: one for executing actions and consuming them from the multi-trace and one for simulating actions without a corresponding consumption. It then consists in exploring a graph to find possible explanations of an observed multi-trace.
From a finite number of starting nodes, the mechanism of execution can only yield a finite number of steps given that the multi-trace is finite (in terms of number of actions) and, for each such action, there can only be finitely many manners to interpret it in the interaction term (see \cite{a_small_step_approach_to_multi_trace_checking_against_interactions}). 
However, this is not the case for the mechanism of simulation. 
In order for the algorithm to be terminating, we  have to limit the simulation steps using a  criterion.

We define a generic algorithm which relies on three notions:
\begin{itemize}
    \item flags defined by a function $\gamma$ which aim is to keep track of whether or not observation has started on the co-localizations $c \in C$;
    \item a space of measures $\mathbb{J}$ fitted with a strict order relation $<_\mathbb{J}$ which can be parameterized and serves as a means to limit the number of simulation steps;
    \item an initialization function $\kappa$ which sets the initial value of the measure at the start of the algorithm process.
\end{itemize}

\subsubsection*{Flags $\gamma$}

In the algorithm, execution steps and simulation steps may be interleaved.
Execution steps can be taken at any moment, provided that there is a match between an action at the beginning of the multi-trace and an immediately executable action of the interaction. However, an additional condition is required for applying a simulation step.
Indeed, the goal of simulation is to reconstruct parts of behaviors that were not observed at the beginning and the end of the period of observation on a given co-localization. As a result, we need an additional condition to ascertain that we are outside this period of observation.

To that end, we define $\gamma : C \rightarrow \mathbb{B}$ (where $\mathbb{B} = \{\bot,\top\}$ is the usual set of Boolean values) such that for any $c \in C$, $\gamma(c) = \bot$ if observation has not started and $\gamma(c) = \top$ if it has. We denote by $\gamma_\bot$ the case where observation has not started on any co-localization i.e. $\forall c \in C$, $\gamma_\bot(c) = \bot$.
For any $\gamma \in \mathbb{B}^C$ and $c \in C$, we denote by $\gamma + c$ the function s.t. $\forall~ c' \in C\setminus \{c\}$, $(\gamma + c)(c') = \gamma(c')$ and $(\gamma + c)(c) = \top$. We use this notation to update the observation status on $c$ (i.e. that it has started).

\subsubsection*{Parameterizable measure $\mathbb{J}$}

In the sequel, we consider: 
\begin{itemize}
    \item  a set $\mathbb{J}$ of measures fitted with a strict order relation $<_\mathbb{J}$ which admits no infinite descending chains, i.e. which admits no infinite sequences $(j_i)_{\in \mathbb{N}}$ of elements in $\mathbb{J}$ verifying $\forall i \in \mathbb{N}, j_{i+1} < j_i$;
    \item and a relation $\measureDecrements{} \subseteq \mathbb{J} \times (\mathbb{I} \times \{1,2\}^*) \times \mathbb{J}$ verifying that for any $j,j' \in \mathbb{J}$, $i \in \mathbb{I}$ and $p \in \{1,2\}^*$, if $j \measureDecrements{i,p} j'$ then $j'<_\mathbb{J} j$. We note $j \measureNotDecrements{i,p}$ if there does not exist $j'$ s.t. $j \measureDecrements{i,p} j'$.
\end{itemize}

From a pair $(i,j)$ where $i$ is an interaction and $j$ a measure, given an action at position $p$ in $i$, if $j \measureNotDecrements{i,p}$ then we cannot simulate this action starting form $i$. 
If however there exists $j'$ s.t. $j \measureDecrements{i,p} j'$ then we may simulate it and, given $i \xrightarrow{a@p} i'$ with $a$ action at position $p$ in $i$, reach the pair $(i',j')$. 
In this manner, we can bind simulation to a strictly decreasing measure. In the context of our ORV algorithm, and because $(\mathbb{J},<_\mathbb{J})$ has no infinite descending chain, this will imply that, in any run of our algorithm, there can only be finitely many consecutive steps of simulation.

\begin{example}
\label{ex:measure}
For instance, we can consider $\mathbb{J} = \mathbb{N}$ (positive integers), $<$ the classical inequality on positive integers, and the relation $\measureDecrements{}$ s.t. for any $j \in \mathbb{N}$, any $i \in \mathbb{I}$ and $p$ s.t. $\exists i' \in \mathbb{I}$ s.t. $i \xrightarrow{i_{|p}} i'$ we have $j \measureDecrements{i,p} j-1$. With this example, we decrement by $1$ each time we simulate any action.
\end{example}

\subsubsection*{Parameterizable measure (re-)initialization $\kappa$}

The measure being now defined, it still needs to be initialized. For that, we consider any arbitrary function $\kappa : \mathbb{I} \rightarrow \mathbb{J}$. Because consecutive sequences of unobserved actions can occur in between observed actions, whenever an action is executed and consumed in the multi-trace instead of being simulated, we may reset the measure to $\kappa(i')$ given the execution $i \xrightarrow{a@p} i'$.

\begin{example}
\label{ex:kappa}
For instance, following Ex.\ref{ex:measure}, we can consider $\kappa$ s.t. for any $i \in \mathbb{I}$ we have 
$\kappa(i) = 5$. 
With this example criterion, we may simulate successively a number of actions which is at most $5$.
\end{example}

In practice, however, so as to be able to define a wide variety of criteria, the definition of $\kappa$ may also depend on other variables (besides the interaction $i$) such as the (size of the) multi-trace $\mu$ which is analyzed, the current state of the flags $\gamma$ etc.
This is reflected in the tool implementation. However, for the sake of simplicity, we consider $\mathbb{I} \rightarrow \mathbb{J}$ as the  signature of $\kappa$.

\subsection{Analysis graph\label{ssec:orv_algo}}

Let us define a directed search graph $\mathbb{G}$ which set of vertices is $\mathbb{V} = (\mathbb{I} \times \mathbb{M}_{C} \times \mathbb{B}^C \times \mathbb{J}) \cup \{\verdictOk,\verdictKo\}$ i.e. which vertices are:
\begin{itemize}
    \item either of the form $(i,\mu,\gamma,j)$ with $i$ an interaction, $\mu$ a multi-trace, $\gamma$ a flag and $j$ a measure (the latter two defined as in Section~\ref{sec:algo-init});
    \item or can be one of the two verdicts $\verdictOk$ or $\verdictKo$.
\end{itemize}     
The arcs of $\mathbb{G}$ are defined using 4 rules: $\rulePass$ (for "pass verdict"), $\ruleFail$ (for "fail verdict"), $\ruleExec$ (for "execute") and $\ruleSimu$ (for "simulate") which are defined in Def.\ref{def:orv_rules}.

\begin{definition}\label{def:orv_rules}
The graph $\mathbb{G} = (\mathbb{V},\leadsto)$ is defined by:\\
\noindent $-$ the set $\mathbb{V} = (\mathbb{I} \times \mathbb{M}_{C} \times \mathbb{B}^C \times \mathbb{J}) \cup \{\verdictOk,\verdictKo\}$\\
\noindent $-$ the relation $\leadsto \subseteq \mathbb{V} \times \mathbb{V}$ s.t. $\forall~v,v' \in \mathbb{V}$, $v \leadsto v'$ holds iff it can be derived by applying\footnote{the notation $(R) \frac{H}{ v \leadsto v'}$ signifies that $v \leadsto v'$ can be inferred by applying rule $R$ if we suppose that hypothesis $H$ holds.} the four following rules:


{
\small

\begin{scprooftree}{.9}
\AxiomC{}
\LeftLabel{($\rulePass$)}
\RightLabel{}
\UnaryInfC{$(i,\varepsilon_C,\gamma,j) \leadsto \verdictOk$}
\end{scprooftree}

\begin{scprooftree}{.9}
\AxiomC{$\exists~ p ~s.t.~ i \xrightarrow{a@p} i'$}
\LeftLabel{($\ruleExec$)}
\UnaryInfC{$
(i,a  ~\multiAppendLeft~ \mu,\gamma,j)
\leadsto  
(i',\mu,\gamma + \theta_C(a),\kappa(i')) 
$}
\end{scprooftree}

\begin{scprooftree}{.9}
\AxiomC{
$
(\mu \neq \varepsilon_C)
\wedge
\left(
\exists~ a,p ~s.t.~
    \left|
    \begin{array}{l}
        (i \xrightarrow{a@p} i')\\
        \wedge (j \measureDecrements{i,p} j')\\
        \wedge 
            \left(
                \begin{array}{l}
                (\gamma(\theta_C(a)) = \bot)\\
                \vee (\mu_{|\theta_C(a)} = \varepsilon)
                \end{array}
            \right)
    \end{array}
    \right.
\right)
$
}
\LeftLabel{($\ruleSimu$)}
\UnaryInfC{$
(i,\mu,\gamma,j) 
\leadsto 
(i', \mu , \gamma , j')$}
\end{scprooftree}

\begin{scprooftree}{.9}
\AxiomC{
$
(\mu \neq \varepsilon_C)
\wedge
\left(
\left(
\begin{array}{c}
\forall~ a,p,i' ~s.t.\\
(i \xrightarrow{a@p} i')
\end{array}
\right)
    \left|
    \begin{array}{l}
        (\not\exists~\mu'~s.t.~ \mu = a ~\multiAppendLeft~\mu')\\
        \wedge
        \left(
            \begin{array}{l}
            (j \measureNotDecrements{i,p})\\
            \vee 
                \left(
                \begin{array}{l}
                    (\gamma(\theta_C(a)) = \top)\\
                    \wedge (\mu_{|\theta_C(a)} \neq \varepsilon)
                \end{array}
                \right)
            \end{array}
        \right)
    \end{array}
    \right.
\right)
$
}
\LeftLabel{($\ruleFail$)}
\UnaryInfC{$
(i,\mu,\gamma,j)
\leadsto
\verdictKo
$}
\end{scprooftree}
}
where $i$ and $i'$ are interactions, $\mu$ and $\mu'$ are multi-traces, $\gamma$ is a flag, $j$ and $j'$ are measures, $a$ is an action and $p$ is a position. 
\end{definition}

Rules $\rulePass$ and $\ruleFail$ define edges from nodes of the form $(i,\mu,\gamma,j)$ to verdicts (i.e. resp. $\verdictOk$ and $\verdictKo$).
Their conditions of application are exclusive to that of all the other rules which imply that if $v \leadsto \verdictOk$ (or $v \leadsto \verdictKo$) then there is no other edge originating from $v$.
\begin{itemize}
    \item Given that the objective of the algorithm is to recognize slices of accepted multi-traces and because the empty multi-trace $\varepsilon_C$ is a slice of any other multi-trace, rule $\rulePass$ (for "pass") yields a $\verdictOk$ verdict.
    \item $\ruleFail$ (for "fail") yields a $\verdictKo$ verdict if the multi-trace is not empty and if it is not possible to apply any of the other two rules $\ruleExec$ or $\ruleSimu$.
\end{itemize}

Based on the machinery of execution $i \xrightarrow{a@p} i'$ (Def.\ref{def:execution}), rules $\ruleExec$ and $\ruleSimu$ specify edges of $\mathbb{G}$ of the form $(i,\mu,\gamma,j) \leadsto (i',\mu',\gamma',j')$ in which an action occurs: 
\begin{itemize}
    \item the application of $\ruleExec$ corresponds to the identification of an action $a$ which can simultaneously be consumed at the head of a component of $\mu$ and be executed from $i$. When rule $\ruleExec$ applies, $\gamma$ is updated into $\gamma + \theta_C(a)$ to reflect that observation has started on the co-localization $\theta_C(a)$ on which action $a$ has been observed.
    \item with $\ruleSimu$ the action is simulated in the interaction without a corresponding consumption in the multi-trace. Action $a$ can be simulated if and only if at this moment in the global scenario we are outside of the period of observation on the corresponding co-localization $\theta_C(a)$ i.e. it has either not started ($\gamma(\theta_C(a))=\bot$) or has already finished ($\mu_{|\theta_C(a)} = \varepsilon$).
\end{itemize}
Note that the conditions of application of $\ruleExec$ and $\ruleSimu$ are not mutually exclusive. The same action $a$ may be either executed or simulated in case observation has not yet started on $\theta_C(a)$. This allows considering any missing prefix of the multi-trace.

\begin{theorem}[Finite reachable sub-graph]
\label{th:finite_search_space}
From any vertex $v=(i,\mu,\gamma,j)$, the sub-graph of $\mathbb{G}$ reachable from $v$ is finite.
\end{theorem}

\begin{proof}
We prove this in two steps: \textbf{(1)} all paths in $\mathbb{G}$ from $v$ are of finite length and \textbf{(2)} there is a finite number of distinct paths in $\mathbb{G}$ from $v$.

To prove \textbf{(1)} this let us consider a measure on vertices of $\mathbb{V}$ given as a tuple by $|\verdictOk| = |\verdictKo| = (-1,j_0)$ (with $j_0$ any element of $\mathbb{J}$) and, for any vertex $v$ of the form $(i,\mu,\gamma,j)$, $|v| = (|\mu|,j)$ where $|\mu|$ is the length of the multi-trace (in total number of actions). Considering the lexicographic order on the tuple (with the relations $<$ on integers for the first element and $<_\mathbb{J}$ for the second), each one of the $4$ rules decreases strictly this measure:
\begin{itemize}
    \item for $\rulePass$ and $\ruleFail$ because $(-1,j_0) < (|\mu|,j)$ given that $|\mu| \geq 0$ for any multi-trace
    \item for $\ruleExec$ because $(|\mu|,\kappa(i')) < (|a ~\multiAppendLeft~\mu|,j)$ given that $|a ~\multiAppendLeft~\mu| = |\mu| + 1$
    \item for $\ruleSimu$ because $(|\mu|,j') < (|\mu|,j)$ given that $j \measureDecrements{i,p} j'$ by the definition of the rule $\ruleSimu$.
\end{itemize}
In addition $(-1,j_0)$ is a global minimum for this measure and because there are no infinite descending chains in $\mathbb{J}$, there are also none for this measure on nodes of $\mathbb{G}$.
Hence, by construction any path in $\mathbb{G}$ is finite.

To prove \textbf{(2)} we remark that, from any given node $v=(i,\mu,\gamma,j)$ there are at most $2*|i| + 1$ (where $|i|$ designates the total number of actions in $i$) edges of the form $v \leadsto v'$ with $v' \not\in \{\verdictOk,\verdictKo\}$ (this is in the worst case, when every action can be both executed and simulated).
\end{proof}

Th.\ref{th:finite_search_space} states that only a finite sub-graph of $\mathbb{G}$ can be reached from any given vertex $v \in \mathbb{G}$.
Let us also remark that any sink (i.e. any vertex without any outgoing transition) of $\mathbb{G}$ must either be $\verdictOk$ or $\verdictKo$. They are indeed sinks because there are no rules specifying edges of the form $\verdictOk \leadsto v$ or $\verdictKo \leadsto v$ and they are the only ones because for any vertex of the form $v = (i,\mu,\gamma,j)$, if $\mu = \varepsilon_C$ then $\rulePass$ applies and $v$ is not a sink and if $\mu \neq \varepsilon_C$ then:
\textbf{(1)} if there is a match between an action that can be executed from $i$ and the head of a component of the multi-trace then rule $\ruleExec$ applies. \textbf{(2)} if there is some action of $i$ that can be simulated, then rule $\ruleSimu$ applies. \textbf{(3)} if neither condition 1 nor condition 2 hold then rule $\ruleFail$ applies. Indeed, by construction, the conditions of application of $\ruleFail$  are defined as complementary to the conditions of application of the other 3 rules $\rulePass$, $\ruleExec$ and $\ruleSimu$.

\subsection{Verdicts and properties of the generic algorithm}

The ORV algorithm consists in exploring reachable vertices of a graph $\mathbb{G}$ using $\leadsto$ (cf. Def.\ref{def:orv_rules}) from an initial vertex $v = (i,\mu,\gamma_\bot,\kappa(i))$ where $\mu$ is the multi-trace which we want to analyze, $i$ is the interaction which serves as the reference specification and $\gamma_\bot$ is the flag set to $\bot$ on each co-localization $c \in C$. In this starting node, we choose the flag $\gamma_\bot$ since observation has not started on any component and we initialize the measure for simulation using any arbitrary function $\kappa : \mathbb{I} \rightarrow \mathbb{J}$. 

This algorithm is generic given that the measures $\mathbb{J}$, the relation $\measureDecrements{}$ and the functions $\kappa$ are kept generic i.e. are only defined through their profiles and properties.

Because of Th.\ref{th:finite_search_space} and because $\verdictOk$ and $\verdictKo$ are the only two possible sinks, we can conclude that at least one of them is reachable from $v$. In any case, because the reachable part of $\mathbb{G}$ is finite, it is always possible to determine in finite time if $\verdictOk$ is reachable from $v$ (which we may denote by $v \overset{*}{\leadsto} \verdictOk$) or not.

In the context of our ORV algorithm we return a $Pass$ if we can ascertain that $v \overset{*}{\leadsto} \verdictOk$ and an $Inconc$ (meaning an inconclusive verdict) otherwise. The algorithm is defined as a function $\omega$ in Def.\ref{def:anaslice_def} and is well-founded given the previous remark linked to Th.\ref{th:finite_search_space}.

\begin{definition}\label{def:anaslice_def}
For any $C \in \partitionsOf{L}$, we define $\omega_C : \mathbb{I} \times \mathbb{M}_C \rightarrow \{Pass,Inconc\}$ s.t. for any $i \in \mathbb{I}$ and $\mu \in \mathbb{M}_C$:
\begin{itemize}
    \item $\omega_C(i,\mu) = Pass$ iff there exists a path in $\mathbb{G}$ from $(i,\mu,\gamma_\bot,\kappa(i))$ to $\verdictOk$
    \item $\omega_C(i,\mu) = Inconc$ otherwise
\end{itemize}
\end{definition}

Th.\ref{th:correctness} states that a $Pass$ verdict ensure the identification of a multi-trace as a slice of an accepted multi-trace of the initial interaction. However this property is not as strong as e.g. the correctness of the algorithm from \cite{a_small_step_approach_to_multi_trace_checking_against_interactions} i.e. we are not guaranteed to have a $Pass$ for any and all slices of an accepted multi-trace.

\begin{theorem}
[Soundness]
\label{th:correctness}
For any interaction $i \in \mathbb{I}$, partition $C \in \partitionsOf{L}$ and multi-trace $\mu \in \mathbb{M}_C$, we have:
\[
\begin{array}{c}
(\mu \in \sigma_C(i)) \Rightarrow (\omega_C(i,\mu) = Pass)\\
(\omega_C(i,\mu) = Pass) \Rightarrow (\mu \in \sliceOf{\sigma_C(i)})
\end{array}
\]
\end{theorem}

\begin{proof}
For the first point: if the multi-trace $\mu$ that is analyzed is in $\sigma_C(i)$ then there is a corresponding path in the execution tree of $i$ (via Def~\ref{def:semantics}, see also Fig.\ref{fig:followup}).
Using only rule $\ruleExec$ (and never $\ruleSimu$) it is then possible to consume $\mu$ in its entirety.
 As the application of $\ruleExec$ is not constrained by the value of the measure $j$, the analysis is  close to that of the algorithm from \cite{a_small_step_approach_to_multi_trace_checking_against_interactions}.

For the second point we can reason as follows.
Given the nature of $\ruleExec$ and $\ruleSimu$, which both execute actions in the current interaction, any path $v = (i,\mu,\gamma_\bot,\kappa(i)) \overset{*}{\leadsto} \verdictOk$ exactly corresponds to a path in the execution tree (see Fig.\ref{fig:followup}) of $i$ and hence to a full multi-trace $\mu_0$ accepted by $i$ i.e. s.t. $\mu_0 \in \sigma_C(i)$. 
Then, $\mu$ is a partial observation (i.e. a slice) of $\mu_0$ given that: \textbf{(1)} it contains all the actions of $\mu$ corresponding to $\ruleExec$ steps while those corresponding to $\ruleSimu$ are missing and \textbf{(2)} $\ruleSimu$ being only applicable outside the period of observation of the components of $\mu$, those missing actions are either before the start or after the end of those corresponding components. Hence $\mu \in \sliceOf{\mu_0}$ and therefore $\mu \in \sliceOf{\sigma_C(i)}$.
\end{proof}

While the $Pass$ verdict ensures that the multi-trace analyzed is indeed a slice of the considered interaction, the other global verdict $Inconc$ only means an inconclusive verdict and not a $Fail$ one because it does not necessarily mean that the multi-trace which is analyzed is not a slice of an accepted multi-trace. Depending on the structure of the input interaction and depending on the criterion that is selected to bind the simulation, we may not have simulated enough the interaction to bring it to states which would allow the entire consumption of the multi-trace.

\section{Implementation and assessment \label{sec:criterion}}

\subsection{Instantiating the parameters}
\label{ssec:inst-param}

Our approach for simulation-based analysis, as presented in Sec.\ref{sec:algo}, requires a parameterization to be concretized. It consists in defining $\mathbb{J}$, $\measureDecrements{}$ and $\kappa$ (in Sec.\ref{sec:algo} those were only characterized through their properties).

A variety of criteria could be used. For instance, with $\mathbb{J} = \mathbb{N}$ (positive integers), we could use an arbitrary maximum number of actions that can be successively simulated. And whenever an action $i_{|p}$, at position $p$ in interaction $i$ is executed, we would have $j \measureDecrements{i,p} j-1$.

Many such trivial criteria may be defined. However, in the following, we propose a slightly more subtle criteria which tries to strike a balance between a good coverage rate (the ability to identify most if not all slices) and efficiency (via taking care of not simulating too many actions, and hence decreasing the complexity/size of graph $\mathbb{G}$).

Let us consider a concrete criterion for binding simulation in the form of a tuple of integers (i.e. $\mathbb{J} = \mathbb{N}^2$) which we denote by $(\lambda,\alpha)$ such that: 
\begin{itemize}
    \item $\lambda$ represents a maximum number of loops which can be instantiated in a consecutive sequence of simulation steps 
    \item and $\alpha$ relates to a number of actions.
\end{itemize} 
This set $\mathbb{J}$ is fitted with the lexicographic order.

In order to initialize and update this criterion, let us consider two functions: $\numActOutsideBase: \mathbb{I} \rightarrow \mathbb{N}$ which gives the total number of actions occurring outside loops and $\loopDepthAtPosBase : \mathbb{I} \rightarrow \mathbb{N}$ which gives the maximum depth of nested loops. More precisely, $\beta(i) = max_{p\in pos(i)} \; \beta(i,p)$ where $\beta(i,p)$ is the number of nested loops above position $p$ in interaction $i$ (see Appendix \ref{anx:details_criterion} for a complete definition). We then define $\kappa(i)$ as the couple $( \loopDepthAtPosBase(i),~ \numActOutside{i} )$.

We define the relation $\measureDecrements{} \subseteq \mathbb{J} \times (\mathbb{I} \times \{1,2\}^*) \times \mathbb{J}$ by: for any $(\lambda,\alpha) \in \mathbb{J}$, any $i,i' \in \mathbb{I}$, $a \in \mathbb{A}$ and $p \in pos(i)$ s.t. $i \xrightarrow{a@p} i'$,
\begin{itemize}
    \item if $\loopDepthAtPos{i}{p} = 0$, we are in the case where we simulate an action outside a loop. Here $\lambda$ stays the same while $\alpha$ decreases. Indeed at least one action (the one which is executed) is removed from $i$ (more may be removed due to pruning or choosing alternatives etc) and none are added (because no loops are instantiated). Here we have $(\lambda,~\alpha) \measureDecrements{i,p} (\lambda,~\numActOutside{i'})$ with $0 \leq \numActOutside{i'} \leq \alpha - 1$ by construction.
    \item if $\loopDepthAtPos{i}{p} > 0$, we instantiate $\loopDepthAtPos{i}{p}$ loops which requires that $\lambda \geq \loopDepthAtPos{i}{p}$. Then $(\lambda,~\alpha) \measureDecrements{i,p} (\lambda - \loopDepthAtPos{i}{p},~\numActOutside{i'})$. Here we reset the value of $\alpha$ because loop instantiation may change the total number of actions outside loops.
\end{itemize}

$\lambda$ guarantees that we can instantiate at least once every loop in the interaction (although not necessarily in the same path). 
This limit on the number of loops is sufficient to guarantee termination because there can only be a finite number of actions existing outside loops and each of those can be simulated at least once (which corresponds to $\alpha$).
Our definition of $\mathbb{J}$ and $\measureDecrements{}$ ensures the required properties, i.e. that it is strictly decreasing within a space which has no infinite descending chains. 
This guarantees the termination of the parameterized algorithm.

This proposal for $\kappa$ is independent of the size of the analyzed multi-trace. It only depends on the reference interaction. This is advantageous in so far as, in practice, the size of the interaction is small compared to the size of the multi-trace. Thus, this measure allows one to calibrate the number of simulation steps according to the complexity of the interaction while drastically limiting their number.

Once $\kappa$ is properly defined, the multi-trace analysis algorithm essentially boils down to a traversal of a finite graph. Different traversal heuristics (depth/breadth-first, with priorities on the application of the rules etc.) can be implemented. This is mentioned in Sec.\ref{ssec:tool_explo} in the context of our tool HIBOU. 

In the following, we illustrate on a small example how $\kappa$ comes into play for the construction of the graph $\mathbb{G}$.

\begin{figure*}[h]
    \centering
    \resizebox{\textwidth}{!}{\input{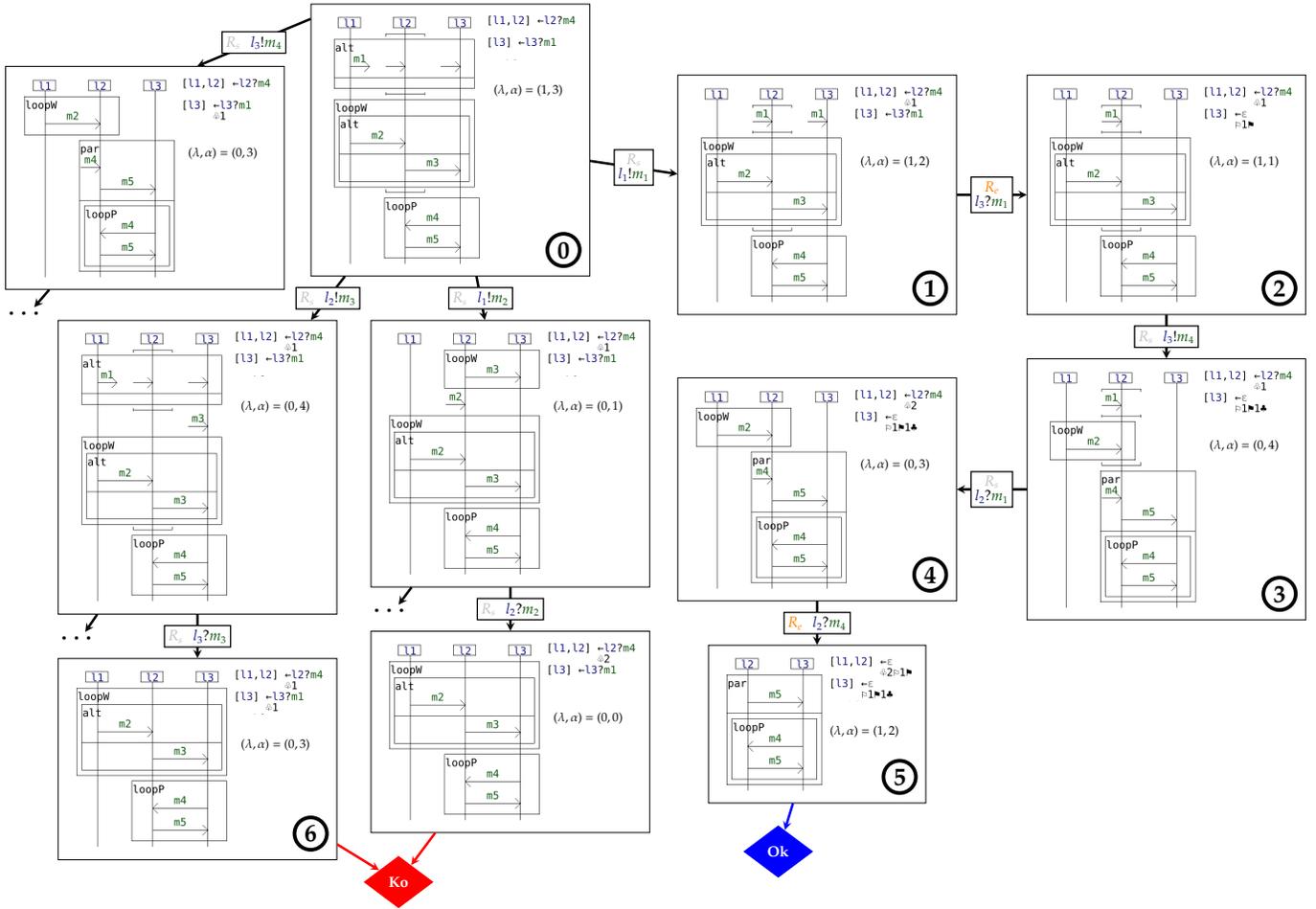}}
    \caption{Multi-trace slice analysis\label{fig:sliceana}}
\end{figure*}

\subsection{Illustration on the running example}

Applying our algorithm with this criterion on our initial example i.e. analyzing the multi-trace $\mu$ from Fig.\ref{fig:mu_coloc_partial} against the interaction $i$ from Fig.\ref{fig:interaction_example} yields the graph (partially drawn) on Fig.\ref{fig:sliceana}. In this particular case, we effectively conclude that $\mu$ is a slice of a behavior accepted by $i$. Here, the global scenario which is reconstructed during the analysis and that match $\mu$ corresponds to the path leading to $\verdictOk$ displayed on Fig.\ref{fig:sliceana}.
This path corresponds to the trace $\shortColGrey{l_1!m_1}.\hlf{l_3}?\hms{m_1}.\shortColGrey{l_3!m_4.l_2?m_1}.\hlf{l_2}?\hms{m_4}$. Notice that is uses a different interleaving of the simulated actions  $\shortColGrey{l_3!m_4}$ and $\shortColGrey{l_2?m_1}$ w.r.t. the global scenario described on Fig.\ref{fig:mu_global1clock} and Fig.\ref{fig:missing_actions}.

Let us comment this path in the graph. Each vertex in this path is annotated with a circled number (from \textcircled{0} to \textcircled{5}). Vertex \textcircled{0} corresponds to the initial vertex $(i_0,\mu_0,\gamma_\bot,\kappa(i_0) = (\lambda_0,\alpha_0))$ from which the analysis starts, where $\mu_0$ is the input multi-trace from Fig.\ref{fig:mu_coloc_partial} which we want to analyze and $i_0$ is the input interaction from Fig.\ref{fig:interaction_example} which serves as a specification.
We initialize the measure to $\kappa(i_0)=(\beta(i_0),\eta(i_0))$, hence $\lambda_0 = \beta(i_0)=1$ because the maximum depth of nested loops is $1$ and $\alpha_0 = \eta(i_0) = 3$ because there are $3$ actions outside loops.

From the initial vertex \textcircled{0}, rules $\rulePass$ and $\ruleExec$ cannot be applied. Indeed, $\rulePass$ is not applicable because the multi-trace is not empty. As for $\ruleExec$, it is because we cannot execute in the interaction any of the actions that are at the heads of the components of the multi-trace (i.e. neither $\hlf{l_2}?\hms{m_4}$ nor $\hlf{l_3}?\hms{m_1}$ can be immediately executed, cf. Fig.\ref{fig:followup} which enumerates all immediately executable actions).
However, $\ruleSimu$ can be applied on any of the four immediately executable actions (cf. Fig.\ref{fig:followup}). Actions $\hlf{l_3}!\hms{m_4}$, $\hlf{l_2}!\hms{m_3}$ and $\hlf{l_1}!\hms{m_2}$ are at depth $1$ w.r.t. loops. Hence, because $\lambda_0 = 1$ we can simulate them. Action $\hlf{l_1}!\hms{m_1}$ is outside all loops and can be simulated because $\alpha_0 = 3$. 
Because $\ruleSimu$ can be applied, $\ruleFail$ cannot be applied from \textcircled{0}.

The simulation of action $\hlf{l_1}!\hms{m_1}$ (at position $1111$) leads to vertex \textcircled{1}. In this vertex $(i_1,\mu_1,\gamma_1,(\lambda_1,\alpha_1))$ we have $i_1$ s.t. $i_0 \xrightarrow{\hlf{l_1}!\hms{m_1}@1111} i_1$ because the corresponding action is executed in the model, $\mu_1=\mu_0$ because no action has been consumed in the multi-trace, $\gamma_1 = \gamma_0$ because observation has not started on any additional co-localization, $\lambda_1 = \lambda_0$ because no action inside a loop has been simulated and $\alpha_1 = \alpha_0 - 1 = 2$ because one action outside loops has been simulated.
Here, we have one step of simulation before the start of observation on component $\{\hlf{l_1},\hlf{l_2}\}$. As a visual aid, on Fig.\ref{fig:sliceana} a $\varclubsuit$~(white clover) symbol indicates the beginning of simulation before the start of observation. This $\varclubsuit$~is then followed by the number of simulated actions.

From vertex \textcircled{1}, rule $\ruleExec$ can be applied on action $\hlf{l_3}?\hms{m_1}$ i.e. we can execute this action and consume it from the multi-trace. The previous step of simulation has put the initial model $i_0$ to a state $i_1$ from which this action can be executed. In the next vertex \textcircled{2}, $\gamma_2(\{\hlf{l_3}\}) = \top$ because observation has then started on the co-localization $\{\hlf{l_3}\}$. Moreover, because the entire trace component on $\{\hlf{l_3}\}$ has been consumed, we have on Fig.\ref{fig:sliceana} both the \faFlagO~ and \faFlag~ visual aids.
In \textcircled{2} the measure is also reset to $\lambda_2 = 1$ and $\alpha_2 = 1$ because of the loop depth and number of actions outside loops in $i_2$.

From vertex \textcircled{2}, neither $\rulePass$ nor $\ruleExec$ can be applied. However, it is possible to simulate actions $\hlf{l_1}!\hms{m_2}$, $\hlf{l_2}?\hms{m_1}$, $\hlf{l_2}!\hms{m_3}$ and $\hlf{l_3}!\hms{m_4}$ because of the current values of $\lambda_2$ and $\alpha_2$ and for the following reasons: for the first three actions, simulation is possible because observation has not yet started on co-localization $\{\hlf{l_1},\hlf{l_2}\}$; for action $\hlf{l_3}!\hms{m_4}$, it is possible because observation has already ended on co-localization $\{\hlf{l_3}\}$ (because the corresponding trace component is already empty).

The simulation of $\hlf{l_3}!\hms{m_4}$ leads to vertex \textcircled{3}. On Fig.\ref{fig:sliceana}, the visual aid $\clubsuit$~(black clover) denotes the number of simulation steps after the end of observation. The measure is updated so that $\lambda_3 = \lambda_2 - 1 = 0$ because $\hlf{l_3}!\hms{m_4}$ was at depth $1$ (hence it is not possible to instantiate new loops in simulation) and $\alpha_3 = 4$ because there are $4$ actions outside loops in $i_3$.

From vertex \textcircled{3} an additional simulation step leads to \textcircled{4}. Because the simulated action $\hlf{l_2}?\hms{m_1}$ is outside all loops, it is $\alpha$ which is decremented. Finally, from vertex \textcircled{4}, rule $\ruleExec$ can be applied so that the entire multi-trace is emptied in vertex \textcircled{5}. Then, $\verdictOk$ can be reached by the application of $\rulePass$ from \textcircled{5}.

Due to having several choices in the applications of rules $\ruleExec$ and $\ruleSimu$, several distinct paths may be opened during the analysis (as illustrated by the paths towards the left of Fig.\ref{fig:sliceana} which leads to $\verdictKo$ and by the $\cdots$ representing other paths which are not drawn). However, via the use of some heuristics and by terminating the analysis as soon as a $\verdictOk$ is reach, one can limit the size of the part of $\mathbb{G}$ which is explored.

At the bottom of Fig.\ref{fig:sliceana} we have also annotated one of the vertices as \textcircled{6} so as to illustrate the application of rule $\ruleFail$. Here, the two previous simulation steps have lead to a vertex $(i_6,\mu_6,\gamma_6,(\lambda_6,\alpha_6))$ in which: \textbf{(1)} the multi-trace is not empty i.e. $\mu_6 \neq \varepsilon_C$, \textbf{(2)} there are no immediately executable actions that match the heads of $\mu_6$ and \textbf{(3)} we have $\lambda_6 = 0$ and there are no actions outside loops remaining in the interaction. Therefore, neither $\rulePass$, $\ruleExec$ nor $\ruleSimu$ can be applied. As a result, we apply rule $\ruleFail$ which leads to $\verdictKo$.


\subsection{Further remarks on the approach\label{ssec:discuss_simu}}

\begin{figure}[h]
    \centering
    \begin{minipage}{.325\linewidth}
        \centering
        \begin{subfigure}[t]{\linewidth}
            \centering
            \hspace*{-.4cm}\includegraphics[scale=.3]{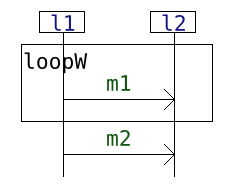}
            \caption{Specifying diagram\label{fig:sections_diagram}}
        \end{subfigure}
    \end{minipage}
    \begin{minipage}{.65\linewidth}
        \centering
        \begin{subfigure}[t]{\linewidth}
            \centering
            $\shortColGrey{l_1!m_1.}\hlf{l_1}!\hms{m_1}.\shortColGrey{l_2?m_1.l_2?m_1.l_1!m_2.}\hlf{l_2}?\hms{m_2}$
            \caption{Global execution which occurred\label{fig:sections_occurred}}
        \end{subfigure}
        \vspace*{.1cm}
        \begin{subfigure}[t]{\linewidth}
            \centering
$
\left\{
\begin{array}{ll}
\lbrack\hlf{l_1}\rbrack & \hlf{l_1}!\hms{m_1}\\
\lbrack\hlf{l_2}\rbrack & \hlf{l_2}?\hms{m_2}
\end{array}
\right.
$
            \caption{Observed multi-trace\label{fig:sections_observed}}
        \end{subfigure}
        \vspace*{.1cm}
        \begin{subfigure}[t]{\linewidth}
            \centering
            $\hlf{l_1}!\hms{m_1}.\shortColGrey{l_1!m_2.l_2?m_1}.\hlf{l_2}?\hms{m_2}$
            \caption{Minimal reconstructible global execution\label{fig:sections_reconstructed}}
        \end{subfigure}
    \end{minipage}

    \begin{subfigure}[t]{\linewidth}
        \centering
        \scalebox{.45}{\input{figures/5_2_sections/graph}}
        \caption{Path found in the analysis graph\label{fig:sections_ana}}
    \end{subfigure}
    
    \caption{Filling-in missing sections optimistically \& a minima}
    \label{fig:sections}
\end{figure}

Simulation steps are used to find possible replacements for missing sections in what is observed of the global scenario. As we have seen, depending on the architecture of the system, those missing sections can be temporally interspersed in-between sections that are observed in the global scenario.

Another example is described on Fig.\ref{fig:sections}. Here a global execution characterized by the trace from Fig.\ref{fig:sections_occurred} was observed and recorded into the multi-trace from Fig.\ref{fig:sections_observed}. The observation on $\hlf{l_1}$ both started too late and ended too early while that on $\hlf{l_2}$ started too late. The actions which have not been observed are grayed-out in Fig.\ref{fig:sections_occurred}.
The global execution is correct w.r.t. the input interaction from Fig.\ref{fig:sections_diagram}
and hence this multi-trace is an accepted slice. Analyzing it w.r.t. the interaction yields the graph given on Fig.\ref{fig:sections_ana}.
The algorithm tries to find a global scenario which fits the observed multi-trace. However, this global scenario may not necessarily correspond to the one which effectively occurred. It is sufficient that it both conforms to the specification and explains the observed multi-trace. 
Ideally it should then look for such a minimal reconstructible scenario such as the one from Fig.\ref{fig:sections_reconstructed} which indeed corresponds to the one unveiled by the analysis graph on Fig.\ref{fig:sections_ana}.

We may remark that our simulation-based approach is optimistic. Indeed, it suffices to find a reconstructible global scenario that both fits the interaction model and the input multi-trace. 
It may be possible that unwanted behaviors occurred but their detection is not possible from what was observed of the execution. From another perspective, the simulation is optimistic because we only simulate actions which do not deviate from the specification.


The conditions for the application of rules $\ruleExec$ and $\ruleSimu$ are not mutually exclusive. Hence the same action at the same position within an interaction might be both simulated or executed from the same vertex.
We consider an example global scenario in which a first message $\hms{m_1}$ is transmitted from $\hlf{l_1}$ to $\hlf{l_2}$ and then a second one transits from $\hlf{l_1}$ to $\hlf{l_2}$ before a final message $\hms{m_2}$ is send from $\hlf{l_2}$ to $\hlf{l_1}$. We then suppose that this behavior is observed through the multi-trace which is analyzed on Fig.\ref{fig:duplicate_act_in_sim}. In this multi-trace, only the second instance of $\hlf{l_1}!\hms{m_1}$ is observed.
If we start the analysis by executing $\hlf{l_1}!\hms{m_1}$, we consume the second instance of $\hlf{l_1}!\hms{m_1}$ from the multi-trace instead of the first. This later leads to a $\verdictKo$ as illustrated by the path on the left of Fig.\ref{fig:duplicate_act_in_sim}.
The correct first step in Fig.\ref{fig:duplicate_act_in_sim} is to simulate $\hlf{l_1}!\hms{m_1}$ instead of executing it. It is therefore important to allow the same action to be both simulated and executed, which explains having non mutually-exclusive conditions for the application of rules $\ruleExec$ and $\ruleSimu$.

\begin{figure}[h]
    \centering
    \scalebox{.45}{\input{figures/5_2_duplicate/graph}}
    \caption{Pertinence of non mutually exclusive $\ruleExec$ and $\ruleSimu$}
    \label{fig:duplicate_act_in_sim}
\end{figure}

We may also remark that the problem of speculating which actions to simulate incurs a high complexity. The solution which we have presented here tackles this problem using brute force because we explore possible simulated actions exhaustively up to a certain bound which makes the search space finite. Various optimizations can be envisioned e.g. via a static analysis of the multi-trace when choosing between alternative branches, loops, etc. 


In the algorithm, transitions allowed by rule $\ruleExec$ are of the form 
$(i,a  ~\multiAppendLeft~ \mu,\gamma,j) \leadsto (i',\mu,\gamma + \theta_C(a),\kappa(i'))$ 
i.e. we reset the measure $j$ to a new value which depends on the follow-up interaction $i'$.
With the examples on Fig.\ref{fig:reset}, we illustrate the motivation behind this reset of the measure.

In Fig.\ref{fig:reset_based}, we simulate two instances of $\hlf{l_1}!\hms{m_1}$ in order to be able to execute the two instances of $\hlf{l_2}?\hms{m_1}$ in the multi-trace. Instead of using an arbitrary criterion of size $2$ for the number of loops, we can use our proposal criterion on the maximum loop depth (which is $1$ here) but reset it every time an action is executed. Here, because lifelines $\hlf{l_1}$ and $\hlf{l_2}$ are on different co-localizations, this allows alternating between steps of simulation on co-localization $\{\hlf{l_1}\}$ and steps of execution on co-localization $\{\hlf{l_2}\}$ so as to consume the multi-trace in its entirety. 
Here, the reset of the measure enables us to recognize any repetition of $\hlf{l_2}?\hms{m_1}$. With this reset, the initialization of the measure does not need to depend on the size of the multi-trace $\mu$ but only on the structure of the interaction $i$.

\begin{figure}[h]
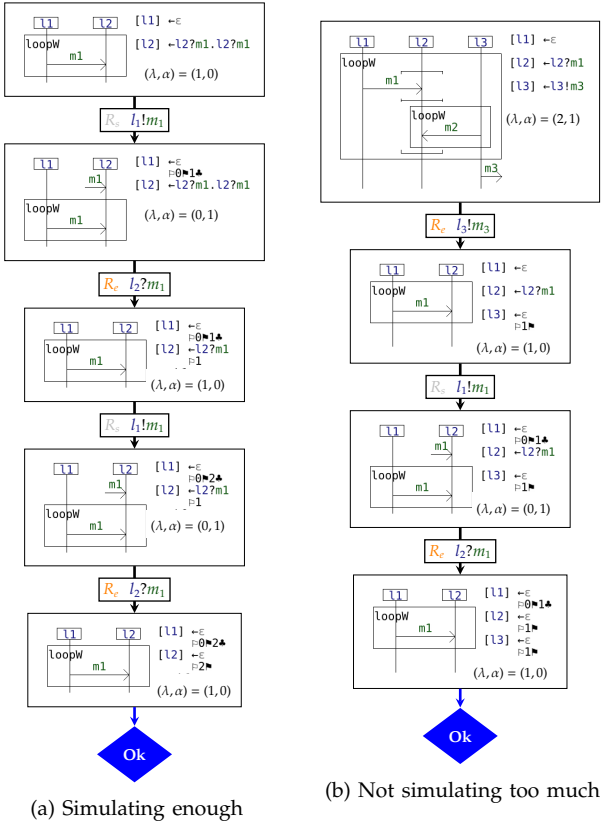

    \centering
    
    \begin{minipage}{.45\linewidth}
        \centering
        \begin{subfigure}[t]{\linewidth}
            \centering
            \scalebox{.5}{\input{figures/5_2_reset/enough/graph}}
            \caption{Simulating enough}
            \label{fig:reset_based}
        \end{subfigure}
    \end{minipage}
    \begin{minipage}{.5\linewidth}
        \centering
        \begin{subfigure}[t]{\linewidth}
            \centering
            \scalebox{.5}{\input{figures/5_2_reset/nottoomuch/graph}}
            \caption{Not simulating too much}
            \label{fig:reset_limit}
        \end{subfigure}
    \end{minipage}
    
    \caption{Motivation behind the reset of the measure}
    \label{fig:reset}
\end{figure}

In the analysis from Fig.\ref{fig:reset_limit}, the initial interaction has a maximum loop depth of $2$. However, due to pruning operations, after the first step, which is an execution (rule $\ruleExec$), the follow-up interaction only has a maximum loop depth of $1$. Hence it is pertinent to reset the measure to reflect this change and so as not to allow more simulation than necessary.

\subsection{Limitations related to inconclusiveness\label{ssec:limitations_simulation}}

Our algorithm returns either a $Pass$ verdict or an $Inconc$ verdict.
When a $Pass$ verdict is returned, the analyzed multi-trace is indeed a slice of an accepted multi-trace.
However, the $Inconc$ verdict do not necessarily reflect a failure. This is because, by bounding the number of simulation steps, it may well be that a correct slice can not be recognized, as its recognition would have required more simulation steps.
In particular, the use of certain specific constructs of the interaction language, in combination with certain architectures of observation, may yield to situations where some correct slices are misidentified (i.e. an $Inconc$ is returned instead of a $Pass$). We provide three such examples in Fig.\ref{fig:failsim}.

\begin{figure*}[ht]
    \centering
    \begin{minipage}{.425\linewidth}
        \centering
        \begin{subfigure}[t]{\linewidth}
            \centering
            \scalebox{.5}{\input{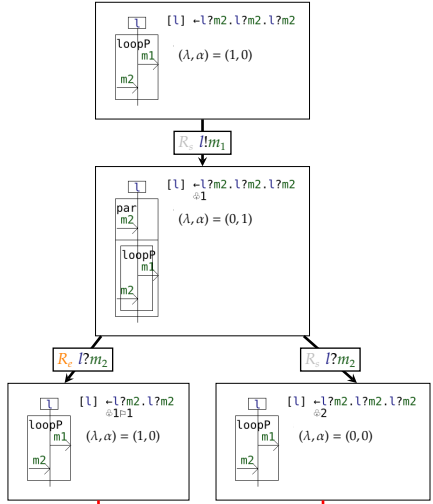}}
            \caption{Example with a $loop_P$ construct.}
            \label{fig:failsimP}
        \end{subfigure}

        \vspace*{.5cm}

        \begin{subfigure}[t]{\linewidth}
            \centering
            \scalebox{.5}{\input{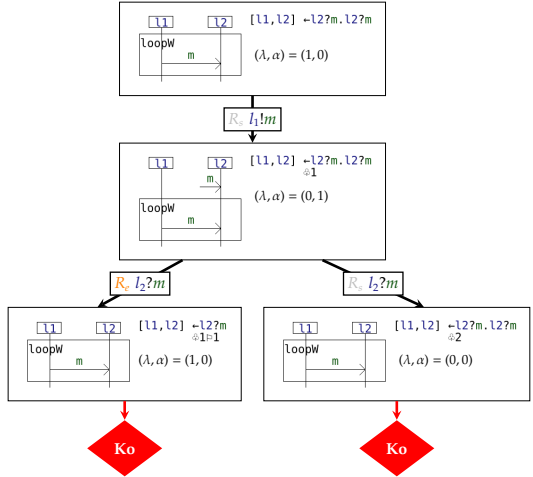}}
            \caption{Example with a $loop_W$ construct under a specific architecture ($\hlf{l_1}$ and $\hlf{l_2}$ co-localized).}
            \label{fig:failsimW}
        \end{subfigure}
    \end{minipage}
    \begin{minipage}{.525\linewidth}
        \centering
        \begin{subfigure}[t]{\linewidth}
            \centering
            \scalebox{.6}{\input{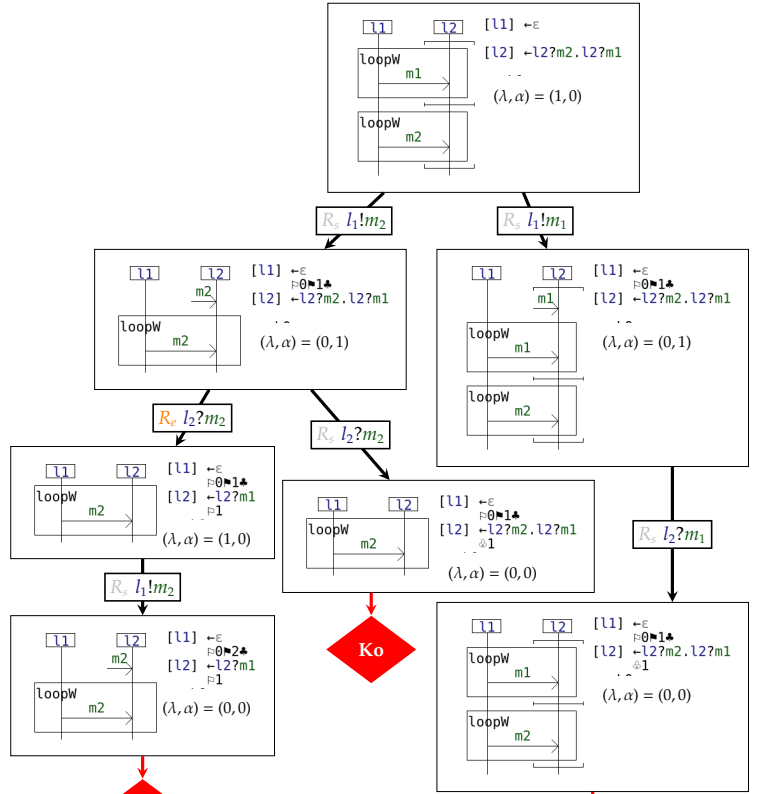}}
            \caption{Example with a $coreg$ construct and two loops.}
            \label{fig:failsim_coreg}
        \end{subfigure}
    \end{minipage}
    \caption{Examples where the proposed criterion does not allow enough simulation to recognise correct slices.}
    \label{fig:failsim}
\end{figure*}

Let us consider the example from Fig.\ref{fig:failsimP}. Because of the use of the parallel loop $loop_P$, several instances of $seq(\hlf{l}!\hms{m_1},\hlf{l}?\hms{m_2})$ can be executed in parallel. In particular the trace $\hlf{l}!\hms{m_1}.\hlf{l}!\hms{m_1}.\hlf{l}!\hms{m_1}.\hlf{l}?\hms{m_2}.\hlf{l}?\hms{m_2}.\hlf{l}?\hms{m_2}$ is thus specified by the interaction. Hence, $\hlf{l}?\hms{m_2}.\hlf{l}?\hms{m_2}.\hlf{l}?\hms{m_2}$ is a slice of an accepted behavior. However, in order to recognize this slice, the action $\hlf{l}!\hms{m_1}$ should be simulated three times consecutively in order to reproduce the prefix missing from the slice. Using the criterion from Sec.\ref{sec:criterion}, we can only execute the action once because the maximum loop depth is $1$ (it would be the same if we used the total number of loops). Hence, as illustrated on Fig.\ref{fig:failsimP}, the analysis fails. This problem is inherent to the $loop_P$ construct as long as the content of the loop (the sub-interaction within it) may express a behavior which contains several distinct actions (here $\hlf{l}!\hms{m_1}$ and $\hlf{l}?\hms{m_2}$).

The $loop_W$ construct is more forgiving in the general case. In particular it may not pose any problem when analyzing classical multi-traces (defined up to the discrete partition, as illustrated on Fig.\ref{fig:reset_based}). However, if distinct lifelines appearing in the sub-interaction underneath a $loop_W$ are co-localized it may pose a problem, as illustrated with the example from Fig.\ref{fig:failsimW}. Here the fact that $\hlf{l_1}$ and $\hlf{l_2}$ are co-localized prevents the algorithm from alternating between simulation steps and execution steps (as it is done in Fig.\ref{fig:reset_based}) because simulation steps can only be taken outside the period of observation on a given co-localization.

The example from Fig.\ref{fig:failsim_coreg} shows a similar problem. Due to the presence of the co-region, lifeline $\hlf{l_2}$ may receive incoming $\hms{m_1}$ and $\hms{m_2}$ messages in any order. However, lifeline $\hlf{l_1}$ must emit all the $\hms{m_1}$ messages before it can emit the first $\hms{m_2}$ message. In the specific multi-trace which we analyze on Fig.\ref{fig:failsim_coreg}, $\hlf{l_2}$ receives one message $\hms{m_2}$ and then one $\hms{m_1}$. Because the maximum loop depth is $1$, we can only simulate one emission consecutively. If we simulate $\hlf{l_1}!\hms{m_1}$ first then we can't do anything afterwards because $\hms{m_1}$ can't be received by $\hlf{l_2}$ yet and we cannot simulate $\hlf{l_1}!\hms{m_2}$ due to the limitation on loops. If we simulate $\hlf{l_1}!\hms{m_2}$ first then, due to transformations of the interaction (which specifies that $\hlf{l_1}$ must emit all the $\hms{m_1}$ messages before it can emit the first $\hms{m_2}$ message), we can not simulate $\hlf{l_1}!\hms{m_1}$ afterwards and hence the analysis fails. In order for this specific analysis to succeed we would need to be able to instantiated $2$ loops. 

Let us remark that taking the total number of loops instead of the maximum depth of loops would help for this specific example (as the total number of loops is $2$) but would not help in all cases. For instance, if a multi-trace $\mu$ with $\mu_{|\{\hlf{l_1}\}} = \varepsilon$ and $\mu_{|\{\hlf{l_2}\}} = \hlf{l_2}?\hms{m_2}.\hlf{l_2}?\hms{m_2}.\hlf{l_2}?\hms{m_1}$ is analyzed, 
instantiating $3$ loops is required.

\subsection{Experimental assessment\label{ssec:experiments}}

\begin{figure*}
    \centering

    \begin{subfigure}[t]{\linewidth}
        \centering
        \resizebox{\textwidth}{!}{\input{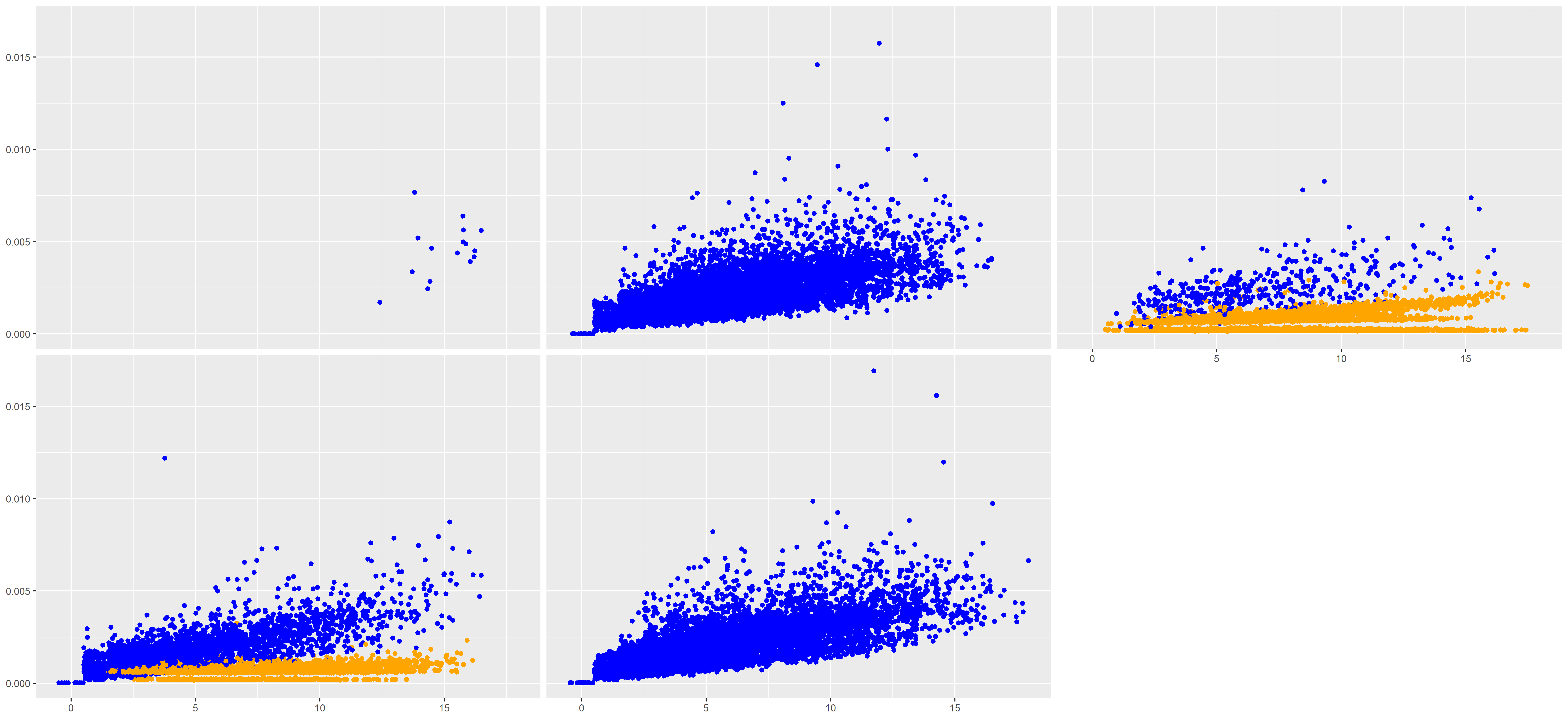}}
        \caption{Example with shallow but exhaustive exploration and exhaustive slicing.}
        \label{fig:experiment_i2}
    \end{subfigure}
    
    \begin{subfigure}[t]{\linewidth}
        \centering
        \resizebox{\textwidth}{!}{\input{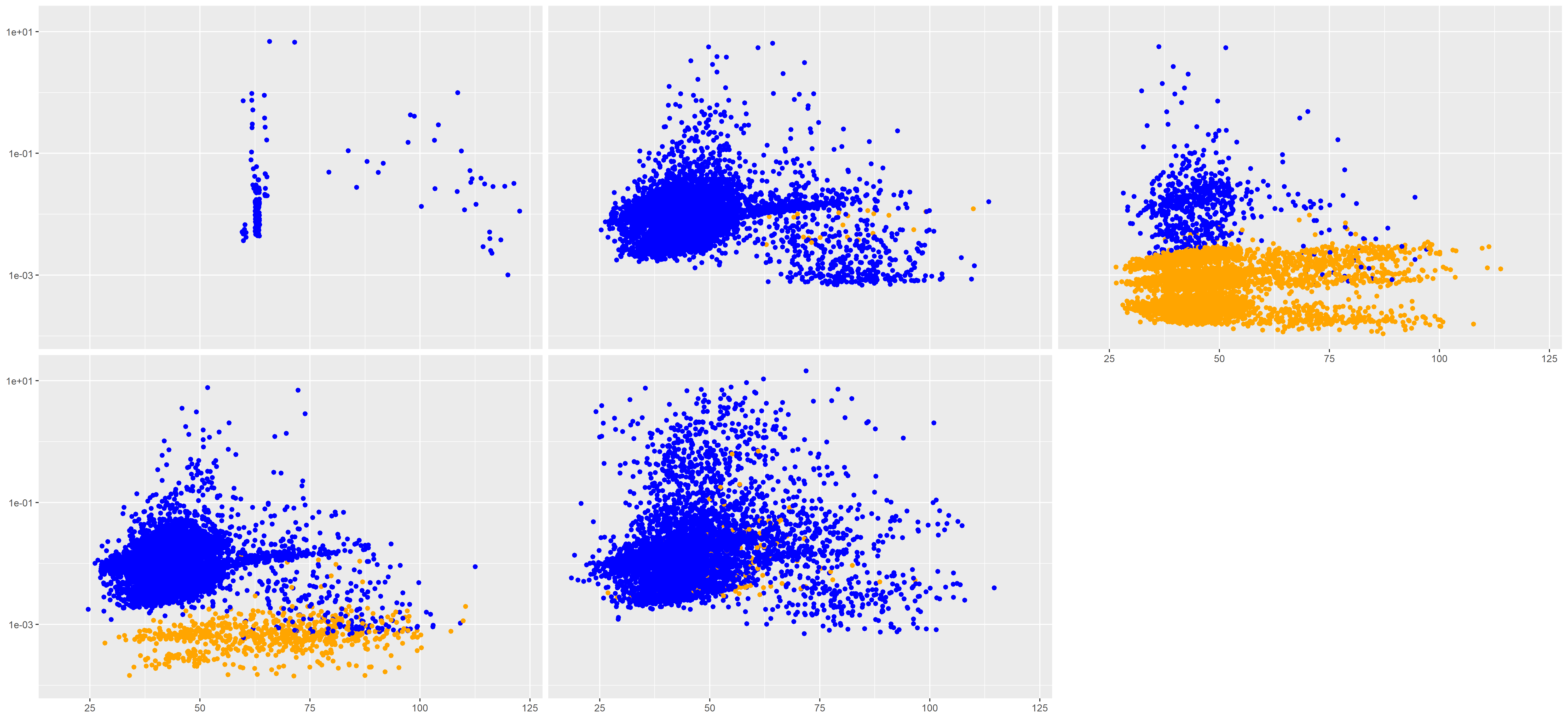}}
        \caption{Example with partial and random but in depth exploration and random selection of wide slices (length $\geq 1/3$ of original). Time in log scale.}
        \label{fig:experiment_i3}
    \end{subfigure}
    
    \caption{Some experiments (length and time resp. in $x$ and $y$ axes, blue and orange representing $Pass$ and $Inconc$ verdicts).}
    \label{fig:experiments}
\end{figure*}

We present a small experimental assessment of an implementation of our algorithm parameterized with the criterion shown in Sec.\ref{ssec:inst-param}.
The results of those experiments are summarized on Fig.\ref{fig:experiments}. 
We exploit two interactions and for each one of them which we denote by $i$:
\begin{enumerate}
    \item we generate a finite set $T(i) \subseteq \sigma_C(i)$ of accepted multi-traces given a certain partition $C$ of lifelines.
    \item for every accepted multi-trace $\mu$, we generate a set of slices. Hence, we obtain a certain $S(i) \subseteq \sliceOf{T(i)}$.
    \item we then generate three sets of mutants from those slices $M_{sa}(i)$, $M_{sc}(i)$, and $M_{ia}(i)$. Those sets of mutants are obtained from $S(i)$ as follows:

\begin{itemize}

    \item $M_{sa}(i)$ is obtained by swapping the relative positions of actions within a local trace. For instance, given $\mu = (\hlf{l_1}!\hms{m_1}.\hlf{l_1}!\hms{m_2},\hlf{l_2}?\hms{m_1}.\hlf{l_2}?\hms{m_2})$, $\mu' = (\hlf{l_1}!\hms{m_2}.\hlf{l_1}!\hms{m_1},\hlf{l_2}?\hms{m_1}.\hlf{l_2}?\hms{m_2})$ is one such mutant. Here, we have swapped the positions of $\hlf{l_1}!\hms{m_1}$ and $\hlf{l_1}!\hms{m_2}$ on local trace $\mu_{|\{\hlf{l_1}\}}$. 
    \item 
    $M_{sc}(i)$ is obtained by swapping local traces between two distinct slices related to the same component. For instance, given $\mu_1 = (\hlf{l_1}!\hms{m_1},\hlf{l_2}?\hms{m_1})$ and $\mu_2 = (\hlf{l_1}!\hms{m_2},\hlf{l_2}?\hms{m_2})$, $\mu' = (\hlf{l_1}!\hms{m_1},\hlf{l_2}?\hms{m_2})$ is one such mutants. 
    \item $M_{ia}(i)$ is obtained by inserting random actions. For instance, given $\mu = (\hlf{l_1}!\hms{m_1},\hlf{l_2}?\hms{m_1})$, $\mu' = (\hlf{l_1}!\hms{m_1}.\hlf{l_1}!\hms{m_3},\hlf{l_2}?\hms{m_1})$ is one such mutant.  
\end{itemize}
\end{enumerate}

The generated multi-trace sets i.e. $T(i)$, $S(i)$, $M_{sa}(i)$, $M_{sc}(i)$, and $M_{ia}(i)$ are 
used as input to feed our analysis algorithm. The median times (based on $5$ tries) required to analyze each of them are plotted on Fig.\ref{fig:experiments}. Each point corresponds to a multi-trace, with on the $x$ axis its length and on the $y$ axis its median analysis time. This time is in seconds, and for Fig.\ref{fig:experiment_i3}, we use a logarithmic scale. The color of the point corresponds to the verdict of the analysis: blue and orange respectively denote $Pass$ and $Inconc$.

For each interaction example (i.e. $i_1$ on Fig.\ref{fig:experiment_i2} and $i_2$ on Fig.\ref{fig:experiment_i3}), the $5$ plots correspond to the five sets $T(i)$, $S(i)$, $M_{sa}(i)$, $M_{sc}(i)$, and $M_{ia}(i)$. Legends written on each plot describe how the corresponding multi-traces have been generated.
The corresponding interaction diagram is drawn on the right of each figure, and statistics on the analysis time are given.

The experiments were performed using an Intel(R) Core(TM)i5-6360U CPU (2GHz) with 8 GB of RAM.
All the material required to reproduce them is publicly available in \cite{hibou_simu_usecases_slice_recognition}.

We can at first notice that the algorithm's performances are highly dependent on the nature of the input interaction. The more the interaction offers branching choices (loops, alternatives) and possible interleavings (weak sequencing and interleaving), the greater can the size of graph $\mathbb{G}$ be, with, as a consequence, worse performances. While example $i_1$ (Fig.\ref{fig:experiment_i2}) is rather sequential, this is not the case for $i_2$ (Fig.\ref{fig:experiment_i2}).

The dependence w.r.t. the size of the input multi-trace appears linear for interaction $i_1$ (Fig.\ref{fig:experiment_i2}) and is less noticeable for interaction $i_2$ (Fig.\ref{fig:experiment_i3}). The analysis time is much more dependent on the structure of the input multi-trace w.r.t. the interaction rather than on its size.

\subsubsection*{Slices} Our criterion is capable of correctly identifying slices in most cases, as illustrated with $S(i_1)$ and $S(i_2)$ (top middle plot in both Fig.\ref{fig:experiment_i2} and Fig.\ref{fig:experiment_i3}). $S(i_1)$ contains all the possible slices ($7996$) of the multi-traces from $T(i_1)$, and all of them are correctly identified (drawn in 
blue color). 
Even though $S(i_2)$ contains slices of much larger multi-traces (up to $30$ instances of the loops of $i_2$), most of those are still correctly identified. The few $Inconc$ verdicts correspond to cases where the criterion does not allow enough simulation steps, as illustrated in Sec.\ref{ssec:limitations_simulation}.

\subsubsection*{Mutants by adding random actions}
In most cases, when we add a random action to a multi-trace which is conform to a specification, it becomes non-conform. This is reflected on the top right plots of both Fig.\ref{fig:experiment_i2} and Fig.\ref{fig:experiment_i3} by the fact that most multi-trace analyses are inconclusive (
$Inconc$ verdict drawn in orange color). Let us recall that because of the limitations of our criterion illustrated in Sec.\ref{ssec:limitations_simulation}, we cannot return a $Fail$ verdict. In terms of performances, in the case of $M_{ia}(i_1)$ and $M_{ia}(i_2)$, we can obtain the $Inconc$ verdicts more quickly on average than the $Pass$ verdicts. This is because we use some optimizations on the exploration of the graph (we use local analyses to cut parts of $\mathbb{G}$).

\subsubsection*{Mutants by swapping actions}
Suppose the specifying interaction allows many possible interleavings of actions. In that case, swapping the positions of actions of a conform slice is more likely to yield a new multi-trace which is also conform.
This is reflected on the bottom left plots of Fig.\ref{fig:experiment_i2} and Fig.\ref{fig:experiment_i3}, where we can see that those mutants are conform in many cases.

\subsubsection*{Mutants by swapping local components
} 
This family of mutants is quite interesting given that for any such mutant obtained from two conform slices, the local components of the mutant are still all locally conform to the specification. This makes techniques such as using local analyses to reveal non-conformities useless. If the mutant is non-conform, then only a global analysis - i.e. matching the multi-trace to an accepted global scenario - can identify this non-conformity.
In any case, most of those mutants are conform because we only consider a labeled language (i.e. no message passing and no value passing). With message passing, there are likely mismatches between the messages that are passed between non-co-localized lifelines.
Still, in our purely labeled framework, some inconclusive verdicts are present for interaction $i_2$ (Fig.\ref{fig:experiment_i3}).


\subsubsection*{Concluding remarks} 

Let us remark that the presence of inconclusiveness is related to the fact that we have to identify strict slices of multi-traces and to the expressiveness of the specification language (co-regions, parallel loops etc.) and the trace language (multi-traces with co-localized lifelines). If we restrict the prerequisites of the analysis, we can use algorithms that do not return $Inconc$ verdicts.
The algorithms from \cite{revisiting_semantics_of_interactions_for_trace_validity_analysis} and \cite{a_small_step_approach_to_multi_trace_checking_against_interactions} can respectively identify full accepted global traces and full accepted multi-traces.
In~\cite{fsen_toappear2023}, we define an algorithm which, instead of using simulation, applies a lifeline removal operation on no-longer-observed lifelines in order to identify multi-prefixes of multi-traces defined on the discrete partition. Multi-prefixes are multi-slices in which missing actions are only located at the end of the local trace components of multi-traces.

For the algorithm based on simulation, these few experiments show that our tool can analyze representative multi-traces with respect to interactions with reasonable performances. In the presence of an inconclusive verdict, the user may:
\begin{itemize}
    \item revert to using one of the other aforementioned algorithms, if their prerequisites are valid (e.g. synchronization at the start etc.),
    \item start again with a more liberal criterion,
    \item or analyze the multi-trace by hand.
\end{itemize}

One such more liberal criterion would consist in 
having the measure $\lambda = \loopDepthAtPosBase(i)$ being multiplied by the size of the multi-trace $|\mu|$ without resetting the measure whenever rule $\ruleExec$ applies.
In essence, this amounts to instantiating the (nested) loops as many times as the number of actions in the multi-trace. 
With this criterion we would be able to recognize all the examples from Fig.\ref{fig:failsim}. 
However, this criterion can produce exponentially larger analysis graphs $\mathbb{G}$, with a strong impact on the algorithm's performance.
This increases the interest in looking into more selective criteria, such as the one which we proposed, at the expense of the completeness of the analysis.

\section{Related works\label{sec:related}}

The aim of Runtime Verification (RV) \cite{introduction_to_runtime_verification_BartocciFFR18,surveyRV_SanchezSABBCFFK19} is to test the conformity of an implemented system against a formal specification which may define a set of accepted and/or unwanted behaviors. To do so, traces - characterizing system executions - are collected by an instrumentation of the system and then confronted to the formal specification. If one such trace deviates from the specification (i.e. does not characterize an accepted behavior or do characterize an unwanted behavior) then the tester has found a bug in the implemented system. RV techniques include offline and online RV. In online approaches the confrontation to the formal specification takes place at the same time the system is being monitored via the instrumentation. This has the advantage of being able to monitor live reactive system as they are being executed (and expressing behaviors that can be extended arbitrarily many times). However online RV is quite constrained by requirements on the efficiency of the monitoring algorithm. Indeed, observed events must be analyzed quickly enough so that they don't stack and cause a memory overflow.
By contrast, in offline approaches, a trace is collected in a first step and then confronted to the specification. As a result, only finite traces can be analyzed but Offline Runtime Verification (ORV) has fewer constraints than its online counterpart.
In the following, we only consider ORV.

\subsection{RV for DS}

\begin{figure*}[h]
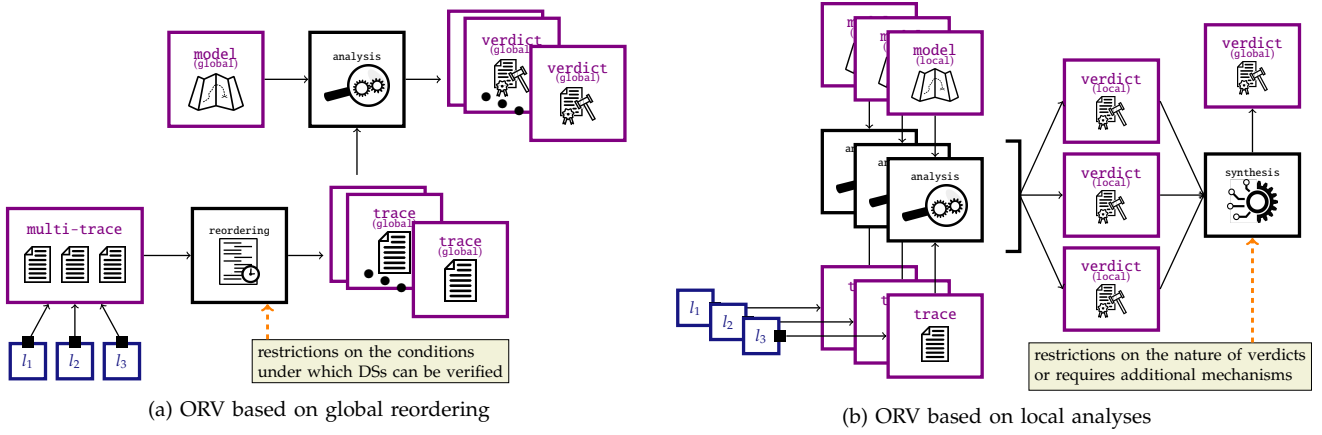

    \centering
    \begin{minipage}{.49\linewidth}
        \centering
        \begin{subfigure}[t]{\linewidth}
            \centering
            \scalebox{.625}{\input{figures/rvkinds/reordering.tex}}
            \caption{ORV based on global reordering}
            \label{fig:orv_reordering}
        \end{subfigure}
    \end{minipage}
    \begin{minipage}{.49\linewidth}
        \centering
        \begin{subfigure}[t]{\linewidth}
            \centering
            \scalebox{.625}{\input{figures/rvkinds/local_analyses.tex}}
            \caption{ORV based on local analyses}
            \label{fig:orv_locana}
        \end{subfigure}
    \end{minipage}
    \caption{ORV approaches from the literature}
    \label{fig:orv}
\end{figure*}

In the literature, most approaches for applying ORV to DS can be categorized as either \textbf{(1)} based on global reordering or \textbf{(2)} based on a synthesis of local analyses. Fig.\ref{fig:orv_reordering} and Fig.\ref{fig:orv_locana} respectively describe those two kinds of approaches.

\textbf{(1)} ORV approaches based on global reordering are often found for solving the Oracle problem in Model Based Testing (MBT). Despite the fact that MBT is supposed to be used at the design phase while RV more naturally concerns the operational phase of a system, algorithms to solve the oracle problem in MBT are technically very similar to offline RV algorithms: they both consist in analyzing an execution to detect non-conformance to a specification. In such MBT approaches, DS behaviors can be modeled using e.g. Input/Output Transition Systems (IOTS)~\cite{dioco_HieronsMN08,scenarios_based_testing_of_systems_with_distributed_ports} or Communicating Sequential Processes (CSP)~\cite{conformance_relations_for_distributed_testing_based_on_CSP_CavalcantiGH11}. 
In those works, local observations are intertwined to associate them with global traces that can be analyzed w.r.t. models.
This is described on Fig.\ref{fig:orv_reordering} as "reordering".
Those approaches however require local observations to be synchronized, based on the states in which each of the logging processes terminate (e.g. based on quiescence states \cite{dioco_HieronsMN08}, termination/deadlocks \cite{conformance_relations_for_distributed_testing_based_on_CSP_CavalcantiGH11} or pre-specified synchronization points \cite{scenarios_based_testing_of_systems_with_distributed_ports}). The works~\cite{the_oracle_problem_when_testing_from_mscs,passive_conformance_testing_of_service_choreographies} focus on verifying distributed executions against models of interaction (While \cite{the_oracle_problem_when_testing_from_mscs} concern MSC, \cite{passive_conformance_testing_of_service_choreographies} considers choreographic languages). Similarly to the MBT approaches, they rely on synchronization hypotheses and on reconstructing a global trace by ordering events occurring at the distributed interfaces (by exploiting the observational power of testers~\cite{the_oracle_problem_when_testing_from_mscs} or timestamp information assuming clock synchronization~\cite{passive_conformance_testing_of_service_choreographies}). In both MBT and RV inspired approaches, synchronization hypotheses restrict the conditions under which DS can be verified.

\textbf{(2)} ORV approaches based on a synthesis of local analyses are often found whenever local monitors are synthesized from a global property. Using temporal logic properties as specifications for RV of DS has been widely studied. In particular, using (variants of) the Linear Temporal Logic (LTL). For example, \cite{SenVAR04} extends the LTL in a framework where formulas relate to sub-systems and what they know about the other sub-systems (e.g. in which local states they are). There is a collection of decentralized observers sharing information about the sub-system executions affecting the validity of the formula. In other works \cite{Falcone16,El-HokayemF17}, the reference properties are expressed at the (global) system level. A collection of decentralized observers is generated from such properties, using LTL formula rewriting, so that there is no need for a global verifier gathering all information on the execution of the system. 
Approaches based on logics derived from modal $\mu$-calculus (e.g. safety Hennessy-Milner Logic sHML) also implement synthesis of monitors. 
In \cite{trace_partitioning_and_local_monitoring_for_asynchronous_components_AttardF17,better_late_than_never_or_verifying_asynchronous_components_at_runtime_AttardAAFIL21}, instrumentation and monitors for programs running on the Erlang Virtual Machine are synthesized from properties written in sHML.
Behavioral models can also be used in such approaches. 
For instance \cite{monitoring_networks_through_multiparty_session_types_BocchiCDHY17} proposes local online RV against projections of multiparty session types satisfying some conditions that enforce intended global behaviors (possibly combined with analysis through a global safe router). More generally, the fact that every local behavior is conform w.r.t. its corresponding local specification, does not generally implies that the global behavior is conform. Hence, the synthesis step described on Fig.\ref{fig:orv_locana} generally requires additional mechanisms. 
In most cases this consists in having local monitors communicate with one another (e.g. in \cite{Falcone16}). This in return requires a very specific instrumentation which goes well beyond simply logging traces.

\cite{a_taxonomy_for_classifying_runtime_verification_tools_FalconeKRT21} provides a taxonomy of RV tools. Specification languages which are most commonly found and fitted with tools derive from Temporal Logic.
Unlike such logics, behavioral models of interactions, which are particularly adapted for specifying DS, are more rarely used.
The works~\cite{the_oracle_problem_when_testing_from_mscs,passive_conformance_testing_of_service_choreographies,monitoring_networks_through_multiparty_session_types_BocchiCDHY17,a_novel_runtime_verification_solution_for_iot_systems_InckiA18} focus on verifying distributed executions against models of interaction (While \cite{the_oracle_problem_when_testing_from_mscs,a_novel_runtime_verification_solution_for_iot_systems_InckiA18} concern MSC, \cite{passive_conformance_testing_of_service_choreographies} considers choreographic languages, \cite{monitoring_networks_through_multiparty_session_types_BocchiCDHY17} session types and \cite{coping_with_bad_agent_interaction_protocols_when_monitoring_partially_observable_multiagent_systems_AnconaFFM18} trace expressions).

\subsection{RV under partial observation}

In this context, {\em partial observation} signifies that the multi-trace logged by the instrumentation does not characterize the entire execution of the DS. More concretely, some events may be missing from the multi-trace w.r.t. an ideal multi-trace which would have been observed with ideal conditions of observation. The notion of {\em multi-trace slice} from Sec.\ref{ssec:slices} proposes a specific definition of partial observation, where events may be missing at the beginning and/or the end of each local trace component of the multi-trace (independently).

In the literature, ORV approaches for DS which are tolerant to partial observation are rare. 
In \cite{coping_with_bad_agent_interaction_protocols_when_monitoring_partially_observable_multiagent_systems_AnconaFFM18}, the authors are interested in generating monitors for distributed RV and in particular, how those monitors can be adapted for a specific definition of partial observation. Here, messages are exchanged via channels which are associated to an observability likelihood. \cite{coping_with_bad_agent_interaction_protocols_when_monitoring_partially_observable_multiagent_systems_AnconaFFM18} uses trace expressions as specifications and proposes transformations that can adapt those expressions to partial observation by removing or making optional a number of identified unobservable events.


The fact that such missing actions may be located anywhere in a globally sequential behavior (as illustrated in Fig.\ref{fig:missing_actions}) is particularly challenging for ORV.
Indeed, if we base ORV on the recognition of a global behavior against a behavioral model (e.g. choreographies \cite{passive_conformance_testing_of_service_choreographies}, MSC \cite{the_oracle_problem_when_testing_from_mscs} or interactions \cite{revisiting_semantics_of_interactions_for_trace_validity_analysis}) then the observed behavior must match a behavior that can be run through the model. This is not possible in the presence of missing actions.

In \cite{fsen_toappear2023}, a partial solution to this problem is given via the use of a lifeline removal operation which is applied on no-longer-observed sub-systems. During the analysis, once a local trace component is entirely re-enacted, the behavioral model is simplified by removing all behavior that concerns the no-longer-observed component. This enables the analysis to be pursued even if the observation on some components has ceased too early.

\subsection{Position of the contribution}

We propose an approach for ORV which confronts observed multi-traces to positive requirements in the form of behavioral models.
Our approach does away with synchronization hypotheses on the beginning and the end of observation in between distant observers.
As a result, the observed multi-traces can be slices of multi-traces that fully characterize executions of the DS (as defined in Sec.\ref{sec:multitraces}).
We use interactions \cite{equivalence_of_denotational_and_operational_semantics_for_interaction_languages} as formal specifications.

Neither monitorability nor implementability are in the scope of this paper. Implementability refers to whether or not a DS specified by an interaction can be implemented concretely e.g. for HMSC in \cite{realizability_of_high_level_message_sequence_charts_closing_the_gaps_Lohrey03,realizability_and_verification_of_msc_graphs}. DS specified by interactions which are not implementable cannot consistently produce traces which are conform to the specification. 
Monitorability is generally associated to properties written in temporal logic or Hennesy-Milner logic and refers to whether or not is possible to validate or invalidate the satisfaction of a property via monitoring (e.g. see partial monitorability in \cite{adventures_in_monitorability_from_branching_to_linear_time_and_back_again_AcetoAFIL19}).

We propose an algorithm which checks whether or not an input multi-trace is a slice of a multi-trace that belong to the semantics of an input interaction. 
Instrumentation or recording methodology, which covers the manner with which such multi-traces can be obtained, is out of the scope of this paper.
Because we do not have strong hypotheses on the synchronization of observers, a very simple and lightweight instrumentation could be used.
Depending on the usecase and its level of abstraction, we could either use logs printed by the different machines (which correspond to co-localizations) or capture and filter network traffic in and out of those machines e.g. with the pcap library \cite{libpcap} or Wireshark \cite{wireshark} as demonstrated in e.g. \cite{converting_pcaps_into_weka_mineable_data}.

\section{The tool\label{sec:tool}}

Our tool HIBOU (for Holistic Interaction Behavioral Oracle Utility)\footnote{The word "hibou" stands for owl in French.} provides utilities for drawing, manipulating and exploiting interaction. The version which we describe in this paper is version~0.8.4.
Its code (written using Rust) is hosted on GitHub in \cite{hibou_label}.

\subsection{Encoding of interactions and traces\label{ssec:tool_encode}}

The specification language of our tool covers different notions:
\begin{itemize}
    \item system signatures are defined in ".hsf" (hibou signature file) files as illustrated on Fig.\ref{fig:encoding_signature}. In those files, we declare the set of messages $M$ and lifelines $L$ that constitute the signature of the Distributed System.
    \item interaction terms are specified in ".hif" (hibou interaction file) files as illustrated on Fig.\ref{fig:encoding_interaction}.
    \item multi-traces are encoded in ".htf" (hibou trace file) files as illustrated on Fig.\ref{fig:encoding_multitrace}
\end{itemize}

\begin{figure}[h]
    \centering
\begin{lstlisting}[language=hibou_hsf,style=coloured_hibou_hsf]
@message{
	<@\hms{m1}@>;<@\hms{m2}@>;<@\hms{m3}@>;<@\hms{m4}@>;<@\hms{m5}@>
}
@lifeline{
	<@\hlf{l1}@>;<@\hlf{l2}@>;<@\hlf{l3}@>
}
\end{lstlisting}
    \caption{Declaration of a signature.}
    \label{fig:encoding_signature}
\end{figure}

On Fig.\ref{fig:encoding_interaction} we provide the encoding of the interaction model that is drawn on Fig.\ref{fig:interaction_diagram}. This textual format (found in ".hif" files) is similar to the notations used in Sec.\ref{sec:interactions}. The basic building blocks are the empty interaction $\varnothing$ encoded as "\lstinline[language=hibou_hsf,style=coloured_hibou_hsf_inline]{o}" (lowercase letter o) and communication actions. For actions, we use notations inspired by WebSequenceDiagrams \cite{website_websequencediagrams} / PlantUML \cite{website_plantuml}:
\begin{itemize}
    \item instead of $l_1!m$ we write \lstinline[language=hibou_hsf,style=coloured_hibou_hsf_inline]{l1 -- m ->|}
    \item instead of $l_1?m$ we write \lstinline[language=hibou_hsf,style=coloured_hibou_hsf_inline]{m -> l1}
    \item $strict(l_1!m,seq(l_2?m,l_3?m))$ becomes \lstinline[language=hibou_hsf,style=coloured_hibou_hsf_inline]{l1 -- m ->(l1,l2)}
\end{itemize}
The $strict$, $seq$, $par$, $coreg_r$, $alt$, $loop_S$, $loop_W$ and $loop_P$ constructors directly match text labels in the encoding (we use parentheses to specify concurrent regions e.g. \lstinline[language=hibou_hsf,style=coloured_hibou_hsf_inline]{coreg(l2)} for $coreg_{\{l_2\}}$)
and we use parentheses to enclose sub-interactions.
For any associative operator, we allow $n$-ary notations. For instance, $seq(i_1,i_2,i_3)$ is interpreted as $seq(i_1,seq(i_2,i_3))$.

\begin{figure}[h]
    \centering
\begin{lstlisting}[language=hibou_hsf,style=coloured_hibou_hsf]
seq(
  coreg(<@\hlf{l2}@>)(
    alt(
      <@\hlf{l1}@> -- <@\hms{m1}@> -> (<@\hlf{l2}@>,<@\hlf{l3}@>),
      o
    ),
    loopW(
      alt(
        <@\hlf{l1}@> -- <@\hms{m2}@> -> <@\hlf{l2}@>,
        <@\hlf{l2}@> -- <@\hms{m3}@> -> <@\hlf{l3}@>
    ) )
  ),
  loopP(
    seq(
      <@\hlf{l3}@> -- <@\hms{m4}@> -> <@\hlf{l2}@>,
      <@\hlf{l2}@> -- <@\hms{m5}@> -> <@\hlf{l3}@>
) ) )
\end{lstlisting}
    \caption{Encoding of the interaction from Fig.\ref{fig:interaction_example}}
    \label{fig:encoding_interaction}
\end{figure}

The encoding of multi-traces (within ".htf" files) is straightforward, as illustrated on Fig.\ref{fig:encoding_multitrace}. Co-localizations are declared between square brackets:
\begin{itemize}
    \item either in plain text or via a keyword:
    \item \lstinline[language=hibou_htf,style=coloured_hibou_htf_inline]{[#all]} signifies that all lifelines are in this co-localization and thus we have a global trace
    \item \lstinline[language=hibou_htf,style=coloured_hibou_htf_inline]{[#any]} signifies that the lifelines appearing on the actions of the trace component are taken into account to define the co-localization
    \item if a lifeline from $L$ is left without any co-localization that contains it, then a dedicated co-localization with an empty trace component is created for it
\end{itemize}

\begin{figure}[h]
    \centering
\begin{lstlisting}[language=hibou_htf,style=coloured_hibou_htf]
[#any] <@\hlf{l1}@>!<@\hms{m1}@>.<@\hlf{l2}@>?<@\hms{m1}@>.<@\hlf{l2}@>?<@\hms{m4}@>;
[<@\hlf{l3}@>] <@\hlf{l3}@>?<@\hms{m1}@>.<@\hlf{l3}@>!<@\hms{m4}@>
\end{lstlisting}
    \caption{Encoding of the multi-trace from Fig.\ref{fig:mu_coloc2clocks}}
    \label{fig:encoding_multitrace}
\end{figure}

\subsection{Command line interface}


The HIBOU tool takes the form of a Command Line Interface which includes the following commands:
\begin{itemize}
    \item "\texttt{hibou draw <.hsf file> <.hif file>}" draws an interaction in a graphical form (which we have used in this paper);
    \item "\texttt{hibou explore} <.hsf file> <.hif file> <.hcf file>?" explores the semantics of an interaction i.e. it computes and may display parts of the execution tree of that interaction as well as generate accepted multi-traces;
    \item "\texttt{hibou analyze} <.hsf file> <.hif file> <.htf file> <.hcf file>?" analyzes a multi-trace against an interaction i.e. it computes and may display parts of an analysis graph related to that analysis and returns a global verdict.
\end{itemize}

For the "\texttt{explore}" and "\texttt{analyze}" commands we may also provide a ".hcf" (hibou configuration file) file to further configure the exploration or analysis process (so as to replace the default configuration).

\subsection{Semantics exploration \& heuristics\label{ssec:tool_explo}}

The operational semantics from Sec.\ref{ssec:semantics} is implemented in the tool which enables exploring execution trees for any interaction via the "\texttt{hibou explore}" command. This exploration can be configured via a ".hcf" file, within a \lstinline[language=hibou_hsf,style=coloured_hibou_hsf_inline]{@explore_option} section, as illustrated on Fig.\ref{fig:explo_config}. 

\begin{figure}[h]
    \centering
\begin{lstlisting}[language=hibou_hsf,style=coloured_hibou_hsf]
@explore_option{
  loggers = [graphic[svg,vertical],
             tracegen[generation = exact,
               partition={(l1,l2),(l3)}]
            ];
  strategy = HCS;
  filters = [max_depth = 35,
             max_loop_depth = 4,
             max_node_number = 250];
  priorities = random
}
\end{lstlisting}
    \caption{An example configuration for an exploration}
    \label{fig:explo_config}
\end{figure}

From each new node of the tree, immediately executable actions are determined and their execution scheduled for the exploration. This scheduling order is by default the lexicographic order of their positions (hence actions "at the top" are prioritized over those "at the bottom"). We can change this by setting either a random order or some priorities (e.g. prioritizing emissions over receptions or actions outside or within loops). The exploration of the tree is then performed according to a certain (deterministic) search strategy, which can be Breadth First, Depth First or a High Coverage Search that favors paths sharing fewer common prefixes. As interactions may contain loops, an exploration without constraints would not terminate. With HIBOU, we can set some limits on the exploration (e.g. a maximum depth, a maximum number of loops or of nodes).

This process can be observed by loggers, which can write outputs describing the exploration of the execution tree.
A graphic logger provides a graphical representation of the analysis graph.
It is enabled via \lstinline[language=hibou_hsf,style=coloured_hibou_hsf_inline]{loggers = [graphic[svg,vertical]]}. 
We can configure its output format (svg or png) and we can decide whether the drawing is drawn from top to bottom (vertical) or left to right (horizontal). 
A trace generation logger can be set up via the \lstinline[language=hibou_hsf,style=coloured_hibou_hsf_inline]{tracegen[generation = exact,partition=..]}. 
For each path in the tree it may generate a ".htf" file containing a multi-trace which corresponds to this path. The keyword \lstinline[language=hibou_hsf,style=coloured_hibou_hsf_inline]{exact} signifies that only exactly accepted traces will be generated (as opposed to using \lstinline[language=hibou_hsf,style=coloured_hibou_hsf_inline]{prefix} for generating all prefixes or \lstinline[language=hibou_hsf,style=coloured_hibou_hsf_inline]{terminal} for generating a trace only on terminal nodes of the explored tree).
\lstinline[language=hibou_hsf,style=coloured_hibou_hsf_inline]{partition=..} is used to specify the partition of lifelines on which the multi-traces are generated (we can use the \lstinline[language=hibou_hsf,style=coloured_hibou_hsf_inline]{discrete} and  \lstinline[language=hibou_hsf,style=coloured_hibou_hsf_inline]{trivial} keywords for the $C_d$ and $C_t$ partitions).

\subsection{Multi-trace analysis\label{ssec:tool_ana}}

HIBOU implements several configurable algorithms for analyzing multi-traces w.r.t. interactions.
This process corresponds to the "\texttt{hibou analyze}" command and can be configured via a \lstinline[language=hibou_hsf,style=coloured_hibou_hsf_inline]{@analyze_option} section as illustrated on Fig.\ref{fig:ana_config}. Most options are defined in common with those of the exploration process. However, some are specific to analyses.

\begin{figure}[h]
    \centering
\begin{lstlisting}[language=hibou_hsf,style=coloured_hibou_hsf]
@analyze_option{
  loggers = [graphic[svg]];
  analysis_kind = simulate
                     [before = true, 
                      loop max depth, 
                      reset = true,
                      multiply = false,
                      act num = 10];
  strategy = DFS;
  priorities = [simu = -1];
  goal = WeakPass
}
\end{lstlisting}
    \caption{An example configuration for an analysis}
    \label{fig:ana_config}
\end{figure}

"\lstinline[language=hibou_hsf,style=coloured_hibou_hsf_inline]{analysis_kind}" can be set to specify which analysis algorithm should be used. Among others it proposes:
\begin{itemize}
    \item "\lstinline[language=hibou_hsf,style=coloured_hibou_hsf_inline]{analysis_kind = accept}" which identifies exactly accepted multi-traces and corresponds to the algorithm from \cite{a_small_step_approach_to_multi_trace_checking_against_interactions}. It returns $Pass$ if this is the case and $Fail$ otherwise.
    \item "\lstinline[language=hibou_hsf,style=coloured_hibou_hsf_inline]{analysis_kind = prefix}" which identifies behaviors which are projections of prefixes of accepted global traces. It is adapted from \cite{a_small_step_approach_to_multi_trace_checking_against_interactions} and returns $Pass$ if the behavior is exactly accepted, $WeakPass$ if it corresponds to such a prefix and $Fail$ otherwise.
    \item "\lstinline[language=hibou_hsf,style=coloured_hibou_hsf_inline]{analysis_kind = eliminate}" which identifies prefixes of accepted multi-traces (events missing at the end of local components) defined over the discrete partition. 
    Prefixes in the sense of multi-traces are not necessarily projections of prefixes of global traces. As such, to underline this difference, we may call them multi-prefixes.
    It corresponds to the algorithm from \cite{fsen_toappear2023} which is based on the use of a lifeline removal operation which is applied on no-longer-observed sub-systems.
    This algorithm returns $Pass$ if the behavior is exactly accepted, $WeakPass$ if it corresponds to a multi-prefix and $Fail$ otherwise.
    \item "\lstinline[language=hibou_hsf,style=coloured_hibou_hsf_inline]{analysis_kind = simulate[...]}" which implements the algorithm from this paper. It can be configured with a number of options:
    \begin{itemize}
        \item the "\lstinline[language=hibou_hsf,style=coloured_hibou_hsf_inline]{loop}" option sets up a stopping criterion on the number of loops that can be instantiated in simulation. It can either correspond to the maximum number of loops in the interaction, the maximum depth of its nested loops or to a specific ad-hoc number.
        \item the "\lstinline[language=hibou_hsf,style=coloured_hibou_hsf_inline]{act}" option sets up a stopping criterion on the number of actions that can be executed in simulation. It can either correspond to the number of actions outside loops or a specific ad-hoc number.
        \item if set, the "\lstinline[language=hibou_hsf,style=coloured_hibou_hsf_inline]{before}" option activates simulation before the beginning of observation in addition to after the end of observation on local components. With this option we can recognize slices and without it we can only recognise multi-prefixes.
        \item if set, the "\lstinline[language=hibou_hsf,style=coloured_hibou_hsf_inline]{multiply}" option multiplies the criteria on loops and actions by the size of the multi-trace to analyze.
        \item if set, the "\lstinline[language=hibou_hsf,style=coloured_hibou_hsf_inline]{reset}" option makes so that the measure is reset after a step of execution (if not, then we have a set number of simulation steps for the whole analysis, independently of the use of $\ruleExec$)
    \end{itemize}
    This algorithm may return a $Pass$ or $WeakPass$ if the multi-trace is exactly accepted or either a prefix or a slice of an accepted one or return $WeakFail$ if this is not the case.
\end{itemize}

With the "\lstinline[language=hibou_hsf,style=coloured_hibou_hsf_inline]{goal}" option, we can force the termination of the search once a sufficient verdict is found e.g. $WeakPass$.

\section{Conclusion}

In this paper, we propose a new ORV approach which can be adapted to various observation architectures and is tolerant to the absence of synchronization between local observers. Different sub-systems deployed on the same computer can be modeled by several co-localized lifelines. 
The notion of co-localization allows us to use the same approach to treat the analysis of both global traces (in the case where the system is centralized i.e. all sub-systems are deployed on the same machine and hence all corresponding lifelines are co-localized) and multi-traces (when the system is fully distributed i.e. all sub-systems are deployed on distinct machines and no two lifelines are co-localized). 
Moreover, our handling of partial observation allows taking into account situations where the executions of some sub-systems are not (or are partially) observable due to technical limitations, in particular related to missing observers or the absence of synchronization mechanisms between distant observers. 

Multi-trace analysis combines steps of consumption of actions present in the multi-trace and simulation steps to guess missing (i.e. unobserved) actions. The simulation steps are controlled by a criterion to ensure that the analysis is performed in a finite time. 
At then end of the analysis, two verdicts may be emitted: $Pass$ when a slice is recognized and $Inconc$ if not.
Because actions may be missing at the beginning of local traces and because interactions may include loops, further characterizing the inconclusive verdict may require arbitrarily many simulation steps and is thus not tractable in all generality.
Since accepted multi-traces of an interaction are determined via an operational-style semantics, we can visualize follow-up interactions. Our approach being implemented into the HIBOU tool, this allows users to write and debug interactions at the design stage.

\subsubsection*{Acknowledgements} The research leading to these results has received funding from the European Union’s Horizon Europe programme under grant agreement No 101069748 – SELFY project. 

\bibliographystyle{IEEEtranS}
\bibliography{
biblio/biblio_interactions,
biblio/biblio_sd,
biblio/biblio_rv,
biblio/biblio_test,
biblio/biblio_semantics,
biblio/biblio_others
}

\appendix

\section{Details on the denotational semantics with co-regions}

\subsection{Operators on sets of traces\label{anx:proof_conditional_sequencing}}

\begin{definition}[Interleaving]
\label{anx_def:interleaving}
The set $t_1 \globalInterleaving t_2$ of interleavings of traces $t_1$ and $t_2$ is defined by:
\[
\begin{array}{ccl}
\varepsilon \globalInterleaving t_2 & =  & \{ t_2 \} \\
t_1 \globalInterleaving \varepsilon  & =  & \{ t_1 \} \\
(a_1.t_1) \globalInterleaving (a_2.t_2) & = & 
\{a_1.t ~|~ t \in t_1 \globalInterleaving (a_2.t_2)\}
\cup \{a_2.t ~|~ t \in (a_1.t_1) \globalInterleaving t_2\}
\end{array}
\]
\end{definition}

\begin{property}
For any traces $t_1$ and $t_2$ we have:
\[
t_1 \globalInterleaving t_2 = t_1 ~\globalConditionalSequencing{L}~ t_2
\]
\end{property}

\begin{proof}
For proving $t_1 \globalInterleaving t_2 \subseteq t_1 ~\globalConditionalSequencing{L}~ t_2$ we can reason by induction on trace $t_1$.
See lemma \texttt{cond\_seq\_covers\_interleaving\_1} in Coq proof \cite{coq_hibou_label_eqsem_with_coregions}.\\

For proving $t_1 ~\globalConditionalSequencing{L}~ t_2 \subseteq t_1 \globalInterleaving t_2$ we can reason by induction on the conditions for $t \in t_1 ~\globalConditionalSequencing{L}~ t_2$.
See lemma \texttt{cond\_seq\_covers\_interleaving\_2} in Coq proof \cite{coq_hibou_label_eqsem_with_coregions}.\\
\end{proof}

\begin{definition}[Conflict]
\label{anx_def:conflict}
We define a conflict operator $\doubleVerticalTimes : \mathbb{T}_{|L} \times L \rightarrow \{\top,\bot\}$ such that:
\[
\varepsilon \doubleVerticalTimes l = \bot \hspace{1cm} \mbox{and} \hspace{1cm} (a.t) \doubleVerticalTimes l =  (\theta(a) = l) \vee (t \doubleVerticalTimes l)
\]
\end{definition}

\begin{property}
[Conflict is conditional conflict with empty concurrent region]
\label{anx_lem:conflict_is_condconf_nolifeline}
For any traces $t$ and lifeline $l$ we have:
\[
(t \doubleVerticalTimes l) \Leftrightarrow (t ~\globalConditionalSequencing{\emptyset}~ l) 
\]
\end{property}

\begin{proof}
We can reason by induction on trace $t$.
See lemma \texttt{no\_condconf\_no\_lifelines\_charac} in Coq proof \cite{coq_hibou_label_eqsem_with_coregions}.
\end{proof}

\begin{definition}[Weak Sequencing]
\label{anx_def:weak_sequencing}
The set $t_1 \globalWeakSeq t_2$ of weak sequencings of traces $t_1$ and $t_2$ is defined by:
\[
\begin{array}{cccl}
\varepsilon \globalWeakSeq t_2 & =  & \{ t_2 \} &\\
t_1 \globalWeakSeq \varepsilon & =  & \{ t_1 \} & \\
(a_1.t_1) \globalWeakSeq (a_2.t_2) & =  & &
\{ a_1.t ~|~ t \in t_1 \globalWeakSeq (a_2.t_2) \} \\
& & \cup & \{ a_2.t ~|~ t \in (a_1.t_1) \globalWeakSeq t_2,\; \neg (a_1.t_1 \doubleVerticalTimes \theta(a_2))\}
\end{array}
\]
\end{definition}

\begin{property}
For any traces $t_1$ and $t_2$ we have:
\[
t_1 \globalWeakSeq t_2 = t_1 ~\globalConditionalSequencing{\emptyset}~ t_2
\]
\end{property}

\begin{proof}
For proving $t_1 \globalWeakSeq t_2 \subseteq t_1 ~\globalConditionalSequencing{\emptyset}~ t_2$ we can reason by induction on trace $t_1$.
See lemma \texttt{cond\_seq\_covers\_weak\_seq\_1} in Coq proof \cite{coq_hibou_label_eqsem_with_coregions}.\\

For proving $t_1 ~\globalConditionalSequencing{\emptyset}~ t_2 \subseteq t_1 \globalWeakSeq t_2$ we can reason by induction on the conditions for $t \in t_1 ~\globalConditionalSequencing{\emptyset}~ t_2$.
See lemma \texttt{cond\_seq\_covers\_weak\_seq\_2} in Coq proof \cite{coq_hibou_label_eqsem_with_coregions}.\\
\end{proof}

\subsection{Characterization of the pruning relations\label{anx:proof_pruning}}

\begin{theorem}
[An interaction that can't be pruned has conflicts on all accepted traces]
\label{anx_th:cannot_prune_denotational}
For any $L' \subseteq L$ and any $i \in \mathbb{I}$ we have:
\[
(i \isNotPruneOf{L'})
\Rightarrow 
(\forall~t \in \rho(i),~\exists~l \in L',~ t ~\globalConditionalSequencing{\emptyset}~ l)
\]
\end{theorem}

\begin{proof}
At first we remark that as per Prop.\ref{anx_lem:conflict_is_condconf_nolifeline}, 
$\globalConditionalSequencing{\emptyset} = \doubleVerticalTimes$.
We can then reason by induction on the structure of interaction $i$.
See theorem \texttt{cannot\_prune\_characterisation\_with\_sem\_de} in Coq proof \cite{coq_hibou_label_eqsem_with_coregions}.
\end{proof}

\begin{property}
[An interaction which can be pruned on all lifelines accepts the empty trace]
\label{anx_lem:prune_all_equiv_accept_nil_1}
For any $i,i' \in \mathbb{I}$ we have:
\[
(i \isPruneOf{L} i') \Rightarrow (\varepsilon \in \rho(i))
\]
\end{property}

\begin{proof}
We can reason by induction on the structure of interaction $i$.
See lemma \texttt{prune\_all\_equiv\_accept\_nil\_1} in Coq proof \cite{coq_hibou_label_eqsem_with_coregions}.
\end{proof}

\begin{property}
[An interaction which accepts the empty trace can be pruned on all lifelines]
\label{anx_lem:prune_all_equiv_accept_nil_2}
For any $i\in \mathbb{I}$ we have:
\[
(\varepsilon \in \rho(i)) \Rightarrow  (\exists~i' \in \mathbb{I},~i \isPruneOf{L} i') 
\]
\end{property}

\begin{proof}
We can reason by induction on the structure of interaction $i$.
See lemma \texttt{prune\_all\_equiv\_accept\_nil\_2} in Coq proof \cite{coq_hibou_label_eqsem_with_coregions}.
\end{proof}

\begin{property}
[Prune does not introduce new behaviors]
\label{anx_lem:prune_does_not_introduce_new_behaviors}
For any $L' \subseteq L$, any $i,i' \in \mathbb{I}$ and any $t \in \mathbb{T}_{|L}$ we have:
\[
\left(
\begin{array}{l}
(i \isPruneOf{L'} i')\\
\wedge (t \in \rho(i'))
\end{array}
\right)
\Rightarrow
(t \in \rho(i))
\]
\end{property}

\begin{proof}
We can reason by induction on the structure of interaction $i$.
See lemma \texttt{prune\_characterisation\_with\_sem\_de\_1} in Coq proof \cite{coq_hibou_label_eqsem_with_coregions}.
\end{proof}

\begin{property}
[Prune conserves behaviors without conflicts]
\label{anx_lem:prune_conserve_traces_without_conflicts}
For any $L' \subseteq L$, any $i,i' \in \mathbb{I}$ and any $t \in \mathbb{T}_{|L}$ we have:
\[
\left(
\begin{array}{l}
(i \isPruneOf{L'} i')\\
\wedge (t \in \rho(i))\\
\wedge (\forall~l\in L',~\neg(t \doubleVerticalTimes l))
\end{array}
\right)
\Rightarrow
(t \in \rho(i'))
\]
\end{property}

\begin{proof}
We can reason by induction on the structure of interaction $i$ and use Th.\ref{anx_th:cannot_prune_denotational}.
See lemma \texttt{prune\_characterisation\_with\_sem\_de\_2} in Coq proof \cite{coq_hibou_label_eqsem_with_coregions}.
\end{proof}

\begin{property}
[Prune removes conflicts]
\label{anx_lem:prune_removes_conflicts}
For any $L' \subseteq L$, any $i,i' \in \mathbb{I}$ and any $t \in \mathbb{T}_{|L}$ we have:
\[
\left(
\begin{array}{l}
(i \isPruneOf{L'} i')\\
\wedge (t \in \rho(i'))
\end{array}
\right)
\Rightarrow
(\forall~l\in L',~\neg(t \doubleVerticalTimes l))
\]
\end{property}

\begin{proof}
We can reason by induction on the conditions that make hypothesis $(i \isPruneOf{L'} i')$ valid.
See lemma \texttt{prune\_removes\_conflicts} in Coq proof \cite{coq_hibou_label_eqsem_with_coregions}.
\end{proof}

\begin{theorem}
[Prune characterization in denotational semantics]
\label{anx_def:prune_characterisation}
For any $i,i' \in \mathbb{I}$, for any $L' \subseteq L$ we have:
\[
(i \isPruneOf{L'} i')
\Rightarrow 
(\rho(i') = \{ t \in \rho(i) ~|~ \forall~l\in L',~ \neg (t ~\globalConditionalSequencing{\emptyset}~ l) \})
\]
\end{theorem}

\begin{proof}
At first we remark that as per Prop.\ref{anx_lem:conflict_is_condconf_nolifeline}, 
$\globalConditionalSequencing{\emptyset} = \doubleVerticalTimes$.
Then we use Prop.\ref{anx_lem:prune_does_not_introduce_new_behaviors}, Prop.\ref{anx_lem:prune_conserve_traces_without_conflicts} and Prop.\ref{anx_lem:prune_removes_conflicts}.
See theorem \texttt{prune\_characterisation\_with\_sem\_de} in Coq proof \cite{coq_hibou_label_eqsem_with_coregions}.
\end{proof}

\subsection{Characterization of the execution relation\label{anx:proof_execution}}

\begin{property}
\label{anx_lem:execution_characterisation_with_sem_de_1}
For any $i,i' \in \mathbb{I}$, for any $t \in \mathbb{T}_{|L}$ and $a \in \mathbb{A}_{|L}$ we have:
\[
\left(
\begin{array}{l}
(i \xrightarrow{a} i')\\
\wedge (t \in \rho(i') )
\end{array}
\right) 
\Rightarrow 
(a.t \in \rho(i))
\]
\end{property}

\begin{proof}
We can reason by induction on the conditions that make hypothesis $(i \xrightarrow{a} i')$ hold.
This proof uses Prop.\ref{anx_lem:prune_all_equiv_accept_nil_1}, Prop.\ref{anx_lem:prune_does_not_introduce_new_behaviors}, Prop.\ref{anx_lem:prune_removes_conflicts}.
See lemma \texttt{execution\_characterisation\_with\_sem\_de\_1} in Coq proof \cite{coq_hibou_label_eqsem_with_coregions}.
\end{proof}

\begin{property}
\label{anx_lem:execution_characterisation_with_sem_de_2}
For any $i \in \mathbb{I}$, for any $t \in \mathbb{T}_{|L}$ and $a \in \mathbb{A}_{|L}$ we have:
\[
(a.t \in \rho(i))
\Rightarrow
\left(
\exists~i' \in \mathbb{I},~s.t.~
\left(
\begin{array}{l}
(i \xrightarrow{a} i')\\
\wedge (t \in \rho(i') )
\end{array}
\right) 
\right)
\]
\end{property}

\begin{proof}
We can reason by induction on the structure of interaction $i$.
This proof uses Prop.\ref{anx_lem:prune_all_equiv_accept_nil_2}, Prop.\ref{anx_lem:prune_conserve_traces_without_conflicts}.
See lemma \texttt{execution\_characterisation\_with\_sem\_de\_2} in Coq proof \cite{coq_hibou_label_eqsem_with_coregions}.
\end{proof}

\subsection{Equivalence of the semantics\label{anx:proof_eqsem}}

\begin{theorem}
For any $i \in \mathbb{I}$ we have:
\[
\sigma_{C_t}(i) \subseteq \rho(i)
\]
\end{theorem}

\begin{proof}
We can reason by induction on a trace $t \in \sigma_{C_t}(i)$ and use Prop.\ref{anx_lem:prune_all_equiv_accept_nil_1} for the case $t = \varepsilon$ and Prop.\ref{anx_lem:execution_characterisation_with_sem_de_1} for the case $t = a.t'$.
See theorem \texttt{op\_implies\_de} in Coq proof \cite{coq_hibou_label_eqsem_with_coregions}.
\end{proof}

\begin{theorem}
For any $i \in \mathbb{I}$ we have:
\[
\rho(i) \subseteq \sigma_{C_t}(i) 
\]
\end{theorem}

\begin{proof}
We can reason by induction on a trace $t \in \rho(i)$ and use Prop.\ref{anx_lem:prune_all_equiv_accept_nil_2} for the case $t = \varepsilon$ and Prop.\ref{anx_lem:execution_characterisation_with_sem_de_2} for the case $t = a.t'$.
See theorem \texttt{de\_implies\_op} in Coq proof \cite{coq_hibou_label_eqsem_with_coregions}.
\end{proof}

\section{Details on the proposal criterion\label{anx:details_criterion}}

In this paper we propose in Sec.\ref{sec:criterion} a specific criterion to limit the amount of simulation steps which can be taken during the analysis of a multi-trace. This criterion is a concretization of the abstract criterion used in Sec.\ref{sec:algo} for defining the analysis algorithm.

\subsection{Maximum loop depth}

We set a maximum number of loops which can be instantiated during a continuous sequence of simulation steps. A loop is instantiated if an action within it is executed. For nested loops, we consider that the number of loops which are instantiated corresponds to the depth of the action within these nested loops. For instance, within $loop_S(seq(alt(a_1,o),loop_S(a_2)))$, $a_1$ is at a loop depth of $1$ while $a_2$ is at a loop depth of $2$. An action that is not within any loop is at loop depth $0$. This notion of loop depth of a certain action is captured by the $\loopDepthAtPosBase$ function from Def.\ref{def:loop_depth}, where, for any interaction $i\in \mathbb{I}$ and any one of its position $p \in pos(i)$, $\loopDepthAtPos{i}{p}$ gives the loop depth of the node (within the tree structure of the interaction) at position $p$ (hence this is particularly true for actions).

At the beginning of a sequence of simulation steps, starting from a certain interaction $i_0$, this maximum number of loops is initialized at a value which corresponds to the maximum depth of nested loops in $i_0$ which we denote $\loopDepthAtPosBase(i_0)$ (see Def.\ref{def:loop_depth}). This allows every action occurring in the interaction to be simulated at least once in at least one path starting from this initial interaction $i_0$.

\begin{definition}[Loop depth\label{def:loop_depth}]
$\loopDepthAtPosBase$ defined over $\bigcup_{i \in \mathbb{I}} (\{i\}\times pos(i))$ is the function s.t.:
\begin{itemize}
    \item for any $i \in \mathbb{I}$, $\loopDepthAtPos{i}{\varepsilon} = 0$
    \item for any $i_1,i_2 \in \mathbb{I}^2$ and any $p_1 \in pos(i_1)$, $p_2 \in pos(i_2)$, for any $f \in \{str,alt\}\cup\bigcup_{r \subseteq L} \{coreg_r\}$:
    \begin{itemize}
        \item $\loopDepthAtPos{f(i_1,i_2)}{1.p_1} = \loopDepthAtPos{i_1}{p_1}$
        \item $\loopDepthAtPos{f(i_1,i_2)}{2.p_2} = \loopDepthAtPos{i_2}{p_2}$
    \end{itemize}
    \item for any $i \in \mathbb{I}$, $p \in pos(i)$ and $k \in \{S\}\cup\bigcup_{r\subseteq L}\{C_r\}$, $\loopDepthAtPos{loop_k(i)}{1.p} = \loopDepthAtPos{i}{p} + 1$
\end{itemize}
We then define:
$\loopDepthAtPosBase : \mathbb{I} \rightarrow \mathbb{N}$ s.t. $\forall~i \in \mathbb{I}$:
\[
\loopDepthAtPosBase(i) = max_{p \in pos(i)} \loopDepthAtPos{i}{p}
\]
\end{definition}

We illustrate this with Fig.\ref{fig:loop_depth} in which we explore the semantics of an initial interaction $i_0$ (i.e. its execution tree) with a limitation on the number of loops which can be instantiated. Here this limit is initialized at $2$ which is the maximum loop depth of the initial interaction i.e. $\loopDepthAtPosBase(i_0)=2$. We can see that this allows all the actions occurring in the initial interaction to be expressed at least once in at least one path of the part of the tree that is explored. Within the context of multi-trace analysis with simulation, if the next action in the multi-trace which might be executed needs to be "unlocked" via performing some simulation steps, because this action must in any case appear in the interaction, it might then suffice to take advantage of this remark to parameterize simulation.

\begin{figure}[h]
    \centering
    \scalebox{.630}{\input{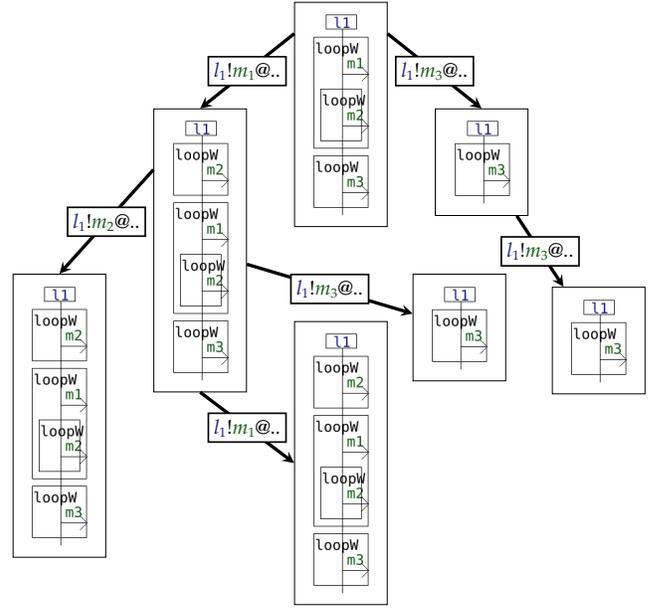}}
    \caption{Semantics exploration (execution tree) with a limitation on the loop depth (here $2$, which is the maximum loop depth of the initial interaction $\loopDepthAtPosBase(i_0)=2$).}
    \label{fig:loop_depth}
\end{figure}

\subsection{Number of actions outside loops}

Setting this limitation on the number of loops is sufficient to ensure termination of the algorithm because an interaction term being finite, there can only be a finite number of actions outside loops which may be simulated and because once such an action is simulated, it disappears from the follow-up interaction, in which the number of actions outside loops therefore diminishes by (at least) once. Yet by only considering the number of loop instantiation in our measure (i.e. by only considering $\lambda$), we do not have a strictly decreasing measure. Indeed, after steps where an action outside a loop is simulated, $\lambda$ stays the same. In order to reflect this decreasing number of actions, which appear in the interaction but not in $\lambda$, we introduce $\alpha$ as the number of actions outside loops, as defined in Def.\ref{def:num_act_outside}. 

\begin{definition}[Number of actions outside loops\label{def:num_act_outside}]
We define $\numActOutsideBase : \mathbb{I} \rightarrow \mathbb{N}$ as follows:
\begin{itemize}
    \item $\numActOutside{\varnothing} = 0$ and for any $a \in \mathbb{A}$, $\numActOutside{a} = 1$
    \item for any $i_1,i_2 \in \mathbb{I}^2$:
        \begin{itemize}
            \item for any $f \in \{strict\}\cup\bigcup_{r \subseteq L} \{coreg_r\}$,\\
            $\numActOutside{f(i_1,i_2)} = \numActOutside{i_1} + \numActOutside{i_2}$
            \item 
            $\numActOutside{alt(i_1,i_2)} = max(\numActOutside{i_1},\numActOutside{i_2})$
        \end{itemize}
    \item for any $i \in \mathbb{I}$ and $k \in \{S\}\cup\bigcup_{r\subseteq L}\{C_r\}$, $\numActOutside{loop_k(i)} = 0$
\end{itemize}
\end{definition}

By considering the tuple $(\lambda, \alpha)$ as a measure, we guarantee that successive simulation steps are bounded by a strictly decreasing measure (see Sec.\ref{sec:criterion}).

\end{document}